\documentclass[a4paper,USenglish,cleveref,numberwithinsect,thm-restate,pdfa]{lipics-v2021}

\def\showTechReport{}

\title{Failure Transparency in Stateful Dataflow Systems (Technical Report)\footnote{This technical report is an extended version of the paper published in the 38th European Conference on Object-Oriented Programming (ECOOP 2024): \texttt{\href{https://doi.org/10.4230/LIPIcs.ECOOP.2024.42}{doi:10.4230/LIPIcs.ECOOP.2024.42}.}}}
\hideLIPIcs 
\pdfoutput=1 

\author{Aleksey Veresov\footnote{Both authors contributed equally to this research.}}{EECS and Digital Futures, KTH Royal Institute of Technology, Stockholm, Sweden \and \url{https://veresov.pro}}{veresov@kth.se}{https://orcid.org/0000-0002-5091-9811}{}

\author{Jonas Spenger\footnotemark[1]}{EECS and Digital Futures, KTH Royal Institute of Technology, Stockholm, Sweden}{jspenger@kth.se}{https://orcid.org/0000-0002-7119-5234}{}

\author{Paris Carbone}{EECS and Digital Futures, KTH Royal Institute of Technology, Stockholm, Sweden \and Digital Systems, RISE Research Institutes of Sweden, Stockholm, Sweden}{parisc@kth.se}{https://orcid.org/0000-0002-9351-8508}{}

\author{Philipp Haller}{EECS and Digital Futures, KTH Royal Institute of Technology, Stockholm, Sweden}{phaller@kth.se}{https://orcid.org/0000-0002-2659-5271}{}

\authorrunning{A. Veresov, J. Spenger, P. Carbone, and P. Haller}

\Copyright{Aleksey Veresov, Jonas Spenger, Paris Carbone, and Philipp Haller}



\funding{This work was partially funded by Digital Futures under a Research Pairs Consolidator grant (PORTALS).}

\nolinenumbers 

\EventEditors{Jonathan Aldrich and Guido Salvaneschi}
\EventNoEds{2}
\EventLongTitle{38th European Conference on Object-Oriented Programming (ECOOP 2024)}
\EventShortTitle{ECOOP 2024}
\EventAcronym{ECOOP}
\EventYear{2024}
\EventDate{September 16--20, 2024}
\EventLocation{Vienna, Austria}
\EventLogo{}
\SeriesVolume{313}
\ArticleNo{15}

\usepackage{anyfontsize}
\usepackage{stmaryrd}
\SetSymbolFont{stmry}{bold}{U}{stmry}{m}{n}
\usepackage{mathpartir}
\usepackage{mathtools}
\usepackage{subcaption}
\usepackage{tikz}
\usetikzlibrary{positioning, arrows, arrows.meta, shapes.arrows, matrix, shapes.geometric, fit, backgrounds, chains, bending, calc, fit, decorations, decorations.pathmorphing, patterns, svg.path, shapes.misc}

\begin{document}

\input{utilities}

\lstdefinelanguage{scala}{
  alsoletter={@},
  morekeywords={abstract, case, catch, class, def, do, else, extends, false, final, finally, for, if, implicit, import, match, new, null, object, 
override, package, private, protected, requires, return, sealed, super, this, throw, trait, try, true, type, val, var, while, with, yield, domain, 
postcondition, precondition, invariant, constraint, assert, forAll, _, return, @generator, ensure, require, ensuring, Source, Sink, Task, Reset, Nil, Event},
  sensitive=true,
  morecomment=[l]{//},
  morecomment=[s]{/*}{*/},
  morestring=[b]"
}

\lstdefinelanguage{eventhandler}{
  alsoletter={@},
  morekeywords={EventHandler, Def, With, Context, On, Event, Fail, Receive, Do, If, And, Init, Then, Function, Repeated, Recover, Vars},
  sensitive=true,
  morecomment=[l]{//},
  morecomment=[s]{/*}{*/},
  morestring=[b]",
}

\definecolor{additionalSyntaxColor}{HTML}{535353}

\lstset{
  frame=tb,
  showstringspaces=false,
  escapeinside=``,
  columns=fullflexible,
  mathescape=true,
  basicstyle=\bfseries\ttfamily,
  keywordstyle=,
  commentstyle=\mdseries\color{additionalSyntaxColor}\itshape,
  stringstyle=\mdseries\color{additionalSyntaxColor},
  identifierstyle=\mdseries,
  breakatwhitespace=true,
  tabsize=3
}
\definecolor{colorFailure}{HTML}{E0526F}

\definecolor{colorRedBase}{HTML}{c35c71}
\definecolor{colorRedEmph}{HTML}{E0526F}
\definecolor{colorRedBack}{HTML}{F6EBEE}

\definecolor{colorBlueBase}{HTML}{5b89c3}
\definecolor{colorBlueEmph}{HTML}{5291e0}
\definecolor{colorBlueBack}{HTML}{E7F1FE}

\definecolor{colorGreenBase}{HTML}{629838}
\definecolor{colorGreenEmph}{HTML}{72c830}
\definecolor{colorGreenBack}{HTML}{f0f6eb}

\definecolor{colorGrayBase}{HTML}{535353}
\definecolor{colorGrayBack}{HTML}{F5F5F5}

\newcommand{\snapshotImage}[1]{\color{white}\pgfsetstrokecolor{#1}\pgfsetlinewidth{2}\pgfsetroundcap\pgfsetroundjoin\pgfpathsvg{M5 5h-3v-1h3v1zm8 5c-1.654 0-3 1.346-3 3s1.346 3 3 3 3-1.346 3-3-1.346-3-3-3zm11-4v15h-24v-15h5.93c.669 0 1.293-.334 1.664-.891l1.406-2.109h8l1.406 2.109c.371.557.995.891 1.664.891h3.93zm-19 4c0-.552-.447-1-1-1-.553 0-1 .448-1 1s.447 1 1 1c.553 0 1-.448 1-1zm13 3c0-2.761-2.239-5-5-5s-5 2.239-5 5 2.239 5 5 5 5-2.239 5-5z}\pgfusepath{fill,stroke}}
\newcommand{\failureImage}{\color{colorFailure}\pgfsetstrokecolor{colorFailure}\pgfsetlinewidth{2}\pgfsetroundcap\pgfsetroundjoin\pgfpathsvg{M8 24l3-9h-9l14-15-3 9h9l-14 15z}\pgfusepath{fill,stroke}}
\newcommand{\storageImage}[1]{\color{#1}\pgfpathsvg{M22 18.055v2.458c0 1.925-4.655 3.487-10 3.487-5.344 0-10-1.562-10-3.487v-2.458c2.418 1.738 7.005 2.256 10 2.256 3.006 0 7.588-.523 10-2.256zm-10-3.409c-3.006 0-7.588-.523-10-2.256v2.434c0 1.926 4.656 3.487 10 3.487 5.345 0 10-1.562 10-3.487v-2.434c-2.418 1.738-7.005 2.256-10 2.256zm0-14.646c-5.344 0-10 1.562-10 3.488s4.656 3.487 10 3.487c5.345 0 10-1.562 10-3.487 0-1.926-4.655-3.488-10-3.488zm0 8.975c-3.006 0-7.588-.523-10-2.256v2.44c0 1.926 4.656 3.487 10 3.487 5.345 0 10-1.562 10-3.487v-2.44c-2.418 1.738-7.005 2.256-10 2.256z}\pgfusepath{fill}}

\makeatletter
\pgfset{
  /pgf/decoration/randomness/.initial=5,
  /pgf/decoration/wavelength/.initial=73
}
\pgfdeclaredecoration{sketch}{init}{
  \state{init}[width=0pt,next state=draw,persistent precomputation={
    \pgfmathsetmacro\pgf@lib@dec@sketch@t0
  }]{}
  \state{draw}[width=\pgfdecorationsegmentlength,
  auto corner on length=\pgfdecorationsegmentlength,
  persistent precomputation={
    \pgfmathsetmacro\pgf@lib@dec@sketch@t{mod(\pgf@lib@dec@sketch@t+pow(\pgfkeysvalueof{/pgf/decoration/randomness},rand),\pgfkeysvalueof{/pgf/decoration/wavelength})}
  }]{
    \pgfmathparse{sin(2*\pgf@lib@dec@sketch@t*pi/\pgfkeysvalueof{/pgf/decoration/wavelength} r)}
    \pgfpathlineto{\pgfqpoint{\pgfdecorationsegmentlength}{\pgfmathresult\pgfdecorationsegmentamplitude}}
  }
  \state{final}{}
}
\tikzset{handdrawn/.style={decorate,decoration={sketch,segment length=0.5pt,amplitude=0.5pt}}}
\makeatother
\pgfmathsetseed{37}

\pgfdeclarelayer{foreground}\pgfdeclarelayer{background}
\pgfsetlayers{background,main,foreground}
\makeatletter\ExplSyntaxOn 
\tl_replace_once:Nnn \tikz@fig@continue { \setbox } { \tikz@setbox@which }
\ExplSyntaxOff
\let\tikz@setbox@which\setbox
\tikzset{node on layer/.code={%
  \expandafter\def\expandafter\tikz@whichbox\expandafter
    {\csname pgf@layerbox@#1\endcsname}%
  \def\tikz@setbox@which{\global\setbox}}}
\makeatother

\maketitle

\begin{abstract}
Failure transparency enables users to reason about distributed systems at a higher level of abstraction, where complex failure-handling logic is hidden. This is especially true for stateful dataflow systems, which are the backbone of many cloud applications. In particular, this paper focuses on proving failure transparency in Apache Flink, a popular stateful dataflow system. Even though failure transparency is a critical aspect of Apache Flink, to date it has not been formally proven. Showing that the failure transparency mechanism is correct, however, is challenging due to the complexity of the mechanism itself. Nevertheless, this complexity can be effectively hidden behind a failure transparent programming interface. To show that Apache Flink is failure transparent, we model it in small-step operational semantics. Next, we provide a novel definition of failure transparency based on observational explainability, a concept which relates executions according to their observations. Finally, we provide a formal proof of failure transparency for the implementation model; i.e., we prove that the failure-free model correctly abstracts from the failure-related details of the implementation model. We also show liveness of the implementation model under a fair execution assumption. These results are a first step towards a verified stack for stateful dataflow systems.
\end{abstract}

\section{Introduction}\label{sec:introduction}
Stateful dataflow systems have seen wide adoption in the modern cloud infrastructure due to their ability to process large amounts of event-based data at ingestion time~\cite{fragkoulis2023survey}.
Apache Flink~\cite{carbone2015apache}, for example, is used to power tens-of-thousands of streaming jobs with up to nine billion events per second at ByteDance~\cite{mao2023streamops}, and several thousand streaming jobs at Uber~\cite{fu2021real}.
An essential aspect of stateful dataflow systems is the recovery from failures, as failures are to be expected in any long-running streaming job~\cite{dean2008mapreduce}.
However, failure recovery is non-trivial.
For example, simply recovering from a failure by restarting a job from the very beginning would discard all progress up to that point, making the recovery prohibitively expensive.
To balance the need for efficiency and reliability, stateful dataflow systems have to embrace complex failure recovery protocols.
Because of their complexity, the correctness of these recovery protocols is a crucial problem for the reliability of stateful dataflow systems.

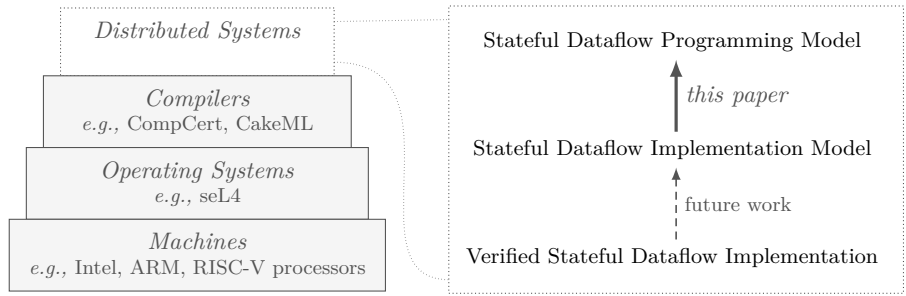
\begin{figure}[t]\centering
\centering
\ifx\hideTikz\undefined
\begin{tikzpicture}[
  scale=0.9, every node/.style={scale=0.9},
  align=center,
  rect/.style={rectangle, draw, colorGrayBase, thin, outer sep=0, inner ysep=2mm},
  proof/.style={colorGrayBase,-latex},
  node distance=0
]
\node[rect, densely dotted, minimum width=4cm] at (0,0) (distributed) {
  \emph{Distributed Systems}\\[-1mm]
  \footnotesize \phantom{x}
};

\node[minimum width=6cm] at (7,0) (sdpm) {
  \small{Stateful Dataflow Programming Model}
};
\node[minimum width=6cm] at (7,-3.175/2) (fim) {
  \small{Stateful Dataflow Implementation Model}
};
\node[minimum width=6cm] at (7,-3.175) (fai) {
  \small{Verified Stateful Dataflow Implementation}
};

\draw[proof, very thick] (fim) -- (sdpm) node[midway,right] {\emph{this paper}};
\draw[proof, semithick, densely dashed] (fai) -- (fim) node[midway,right] {\footnotesize future work};

\node[rectangle, draw, colorGrayBase, thin, densely dotted, outer sep=0, inner sep=4.5mm, fit=(sdpm) (fim) (fai)] (box) {};

\draw[colorGrayBase, very thin, densely dotted] ($(distributed.north east)+(0,-0.2)$) -- ($(box.north west)+(0,-0.2)$);
\draw[colorGrayBase, very thin, densely dotted] ($(distributed.south east)+(0,0.2)$) .. controls +(right:2) and +(left:1.5) .. ($(box.south west)+(0,0.2)$);

\node[rect,below=of distributed, minimum width=4.5cm,fill=colorGrayBack] (compiler) {
  \emph{Compilers}\\[-1mm]
  \footnotesize \eg CompCert, CakeML
};
\node[rect,below=of compiler, minimum width=5cm,fill=colorGrayBack] (os) {
  \emph{Operating Systems}\\[-1mm]
  \footnotesize \eg seL4
};
\node[rect,below=of os, minimum width=5.5cm,fill=colorGrayBack] (machine) {
  \emph{Machines}\\[-1mm]
  \footnotesize \eg Intel, ARM, RISC-V processors
};
\end{tikzpicture}
\else
hideTikz is set
\fi
\caption{This work in the context of a fully verified stack for distributed programming.}
\label{fig:verified-stack}
\end{figure}

In previous work~\cite{elnozahy2002survey,lowell99theory,gartner1999fundamentals}, a failure-masking recovery protocol is considered to be correct, if failures can be masked such that the user cannot observe the failures.
This property is also known as \emph{failure transparency}, \ie a user should be able to ignore failures as if they do not occur.
Failure transparency has been shown for some distributed systems, such as Durable Functions~\cite{burckhardt2021durable}, Reliable State Machines~\cite{mukherjee2019reliable}, reliable actors (KAR)~\cite{tardieu2023reliable}, and serverless microservices (\textmu{}2sls)~\cite{kallas2023executing}.
As for stateful dataflow, the core mechanism used in Apache Flink's~\cite{carbone2015apache} recovery protocol, namely Asynchronous Barrier Snapshotting (ABS)~\cite{carbone2015lightweight}, has been shown to be a correct snapshotting protocol~\cite{carbone2018thesis}.
However, the proof does not reason about failure transparency and its related aspects, such as modelling failures and the recovery from failures, as well as about the equivalence of observed executions.
That is, there has been no formal proof that Apache Flink's entire failure recovery protocol provides failure transparency.
Furthermore, the literature lacks a formal definition of failure transparency for systems described with distinct failure-related rules using small-step operational semantics, a widely-used method for defining program execution in programming languages theory.

An important approach for ensuring reliability and correctness is machine-checked formal verification, \ie proving that a system implements its specification.
There is well-known prior work on verified compilers~\cite{DBLP:journals/cacm/Leroy09,DBLP:conf/popl/KumarMNO14,patterson2019next}, operating systems~\cite{DBLP:conf/sosp/KleinEHACDEEKNSTW09}, as well as processors~\cite{DBLP:journals/pacmpl/ChoiVSCA17,DBLP:conf/cav/KaivolaGNTWPSTFRN09,DBLP:conf/cav/ReidCDGHKPSVZ16} (\Cref{fig:verified-stack}, left).
However, there is an apparent lack of verified distributed systems, particularly, there is no verified stateful dataflow system.
We believe that it is essential to address this gap in order to prevent disastrous outages of distributed infrastructure as known today.

This work is a first step towards the grand goal of providing a fully verified reliable stack for distributed programming, as shown in \Cref{fig:verified-stack}.
It addresses the highlighted gap by:
(1)~providing a definition of failure transparency,
(2)~formalizing a stateful dataflow system as a model in small-step operational semantics, under the assumptions of crash-recovery failures and FIFO-ordered channels, and
(3)~formally proving that the model permits abstracting from failures, \ie that it is failure transparent.
Our definition of failure transparency is based on \emph{observational explainability}, a property which, informally, says that the \emph{explainable} implementation model generates the same observable output as is possible in the \emph{explaining} abstract model.
Using this property, a system is defined as failure transparent if the whole system is observationally explainable by its explicitly separated failure-free part.
Finally, we prove that our formal model of a stateful dataflow system based on Asynchronous Barrier Snapshotting~\cite{carbone2015apache,carbone2018thesis} is failure transparent.
This abstraction from failures is designed to serve the end users of the modelled system with less interest in its implementation details.

\subparagraph*{Contributions.}
In summary, this paper makes the following contributions.
\begin{itemize}
  \item We provide the first small-step operational semantics of the \emph{Asynchronous Barrier Snapshotting} protocol within a stateful dataflow system, as used in Apache Flink (\Cref{sec:implementation-model}).
  \item We provide a novel definition of \emph{failure transparency} for programming models expressed in small-step operational semantics with explicit failure rules and the intuitions behind it (\Cref{sec:failure-transparency}).
  It is the first attempt to define failure transparency in the context of stateful dataflow systems.
  \item We prove that the provided implementation model is failure transparent and guarantees liveness (\Cref{sec:failure-transparency-of-the-implementation}).
  \item We provide a mechanization of the definitions, theorems, and models in Coq.\footnote{\url{https://github.com/aversey/abscoq}}
\end{itemize}

\subparagraph*{Outline.}
\Cref{sec:background} introduces background on failures, distributed systems, stateful dataflow, as well as some basic notation used throughout this paper.
\Cref{sec:stateful-dataflow} informally introduces the stateful dataflow programming model and failure recovery via the Asynchronous Barrier Snapshotting (ABS) protocol.
\Cref{sec:implementation-model} provides a small-step operational semantics of a stateful dataflow system based on ABS.
\Cref{sec:failure-transparency} defines failure transparency and observational explainability for programming models expressed in small-step operational semantics.
\Cref{sec:failure-transparency-of-the-implementation} proves that the implementation model is failure transparent. \Cref{sec:related-work} discusses related work, and \Cref{sec:conclusions-and-future-work} concludes this paper.

\section{Background}\label{sec:background}

\subsection{Failures in Distributed Systems}\label{sec:failures-in-distributed-systems}

A distributed system is a system of many processes communicating over a network~\cite{cachin2011introduction}.
The kind of distributed systems which are related to this work are event-based processing systems such as stateful dataflow systems~\cite{dean2008mapreduce,zaharia2010spark,carbone2015apache,murray2013naiad,zeuch2019nebulastream,spenger2022portals}.
Failures within such systems are expected to happen, due to their typical large scale and longevity~\cite{dean2008mapreduce}.
However, failures are notoriously hard to deal with within distributed systems.
For this reason, \emph{failure transparency} is a necessary abstraction, as it enables the user to abstract from failures.
Failure transparency as a general concept, and failure recovery protocols are both well-studied topics in distributed systems~\cite{neumann1956probabilistic,DBLP:conf/afips/Wensley72,Lee1990,lowell99theory,DBLP:phd/us/Lowell99,gartner1999fundamentals,elnozahy2002survey}.
Moreover, failure transparency has seen an increase in interest within the programming languages community in recent years~\cite{burckhardt2021durable,kallas2023executing,mukherjee2019reliable,tardieu2023reliable}.
The goal of failure recovery is to provide automatic system means to recover from system failures, in ways which the system user may or may not notice.
In contrast, the goal of failure transparency is to provide an abstraction of the system, such that the abstraction hides the internals of failures and failure recovery, masking the failures from the user~\cite{gartner1999fundamentals}.
For this reason, failure transparency greatly simplifies the programming model to the benefit of the end user.

\subsection{Stateful Dataflow and Apache Flink}\label{sec:stateful-dataflow-and-apache-flink}
Stateful dataflow systems, sometimes also called stream processing or dataflow streaming systems, such as Apache Flink~\cite{carbone2015apache}, have become ubiquitous for real-time processing of large amounts of data~\cite{mao2023streamops,fu2021real}.
Other well known dataflow systems include Google Dataflow~\cite{akidau2015dataflow}, IBM Streams~\cite{jacques2016consistent}, Apache Spark~\cite{ZahariaCDDMMFSS12} and Spark Streaming~\cite{zaharia2013discretized}, Timely Dataflow~\cite{murray2013naiad}, NebulaStream~\cite{zeuch2019nebulastream}, Portals~\cite{spenger2022portals}, and more~\cite{balazinska2005fault,shah2004highly}.
The popularity and wide-spread use of dataflow systems~\cite{fu2021real,mao2023streamops} is due to their ability to scale-out production workloads.
In particular, they provide high throughput, low latency, and strong guarantees (such as failure transparency, sometimes referred to as exactly-once processing).
The programming model of most stateful dataflow systems is based on acyclic dataflow graphs~\cite{fragkoulis2023survey}.
In these graphs, the nodes are stateful processing tasks, and the edges are streams of data.
As failure transparency is an important aspect of the stateful dataflow programming model, it and its failure recovery protocol is the focus of this paper.

\subsection{Asynchronous Barrier Snapshotting}\label{sec:asynchronous-barrier-snapshotting}
The failure recovery protocol used in Apache Flink~\cite{carbone2015apache} is a checkpointing-based rollback recovery protocol~\cite{elnozahy2002survey}, in which the system regularly takes checkpoints and, after a failure, recovers to the latest completed checkpoint.
For batch execution systems, such as MapReduce~\cite{dean2008mapreduce}, the general approach is to atomically execute one batch at a time, and if a failure occurs, the system restarts from the beginning of the current batch.
In contrast, computation on stateful dataflow streaming systems is continuous~\cite{fragkoulis2023survey}, without predefined recovery points in its execution, complicating the failure recovery.
The solution to recovery in continuous computations is the acquisition of causally consistent snapshots~\cite{chandy1985distributed}, which can be used for the recovery to a consistent system state after a failure~\cite{elnozahy2002survey}.
The specific implementation of Apache Flink~\cite{carbone2015apache} and other stateful dataflow systems~\cite{spenger2022portals} use the Asynchronous Barrier Snapshotting (ABS) protocol~\cite{carbone2015lightweight}, an extended and optimized variant for data processing graphs of the Chandy-Lamport snapshotting protocol~\cite{chandy1985distributed}, for taking causally-consistent snapshots.
In contrast to the Chandy-Lamport snapshotting protocol, the ABS protocol is tailored to acyclic dataflow graphs and its snapshots do not contain any in-flight events.
In contrast to batching protocols, the ABS protocol is fully asynchronous, and does not require blocking coordination.
For these reasons, the ABS protocol greatly benefits the end-to-end latency and throughput of the system.

\subsection{Basic Notation}\label{sec:basic-notation}

\subparagraph*{Functions.}\parlabel{par:functions}
We denote a function $f$ similarly to set-builder notation as: $\mapbuild{k}{t}{k \in \dom{f}}$.
The part after the bar defines the domain of the function.
The part before the bar defines the value of the function at point $k$ by the expression $t$.
The expression $t$ captures all variables defined on the right side of the bar, including $k$.
A function with only one element in its domain is represented as $\map{x}{x'}$, for example, $\map{3}{7}$ is such a function that $\dom{\map{3}{7}} = \{3\}$ and $\map{3}{7}(3) = 7$.
We denote function update as $f\, g$, such that:
\[(f\,g)(x) =
\begin{cases}
g(x) & \text{if } x \in \dom{g}\\
f(x) & \text{if } x \notin \dom{g}
\end{cases}\]

\subparagraph*{Sequences.}\parlabel{par:sequences}
We represent a sequence $S$ as a function $f$ with domain $\setbuild{i}{i \in \mathbb{N} \land i < \size{S}}$.
The length of the sequence is represented by $\size{S}$ and may be infinite.
The notation $S_i$ stands for the $i$-th element of the sequence $S$ and equals $f(i)$.
To simplify our analysis of sequences, we use $\seq{t}{i}{n}$ as a shorthand for $\mapbuild{i}{t}{i \in \mathbb{N} \land i < n}$, where $t$ is an expression that captures $i$ and represents the $i$-th element of the sequence.
Therefore, for any sequence $S$, we have that $S = \seq{S_i}{i}{\size{S}}$.

The usage of indices for variables standing for sequences may differ from other variables.
If $S$ stands for a sequence, then $S_i$ corresponds to the $i$-th element of $S$.
If, in contrast, $x$ is not a sequence, then $x_i$ is an independent variable and is not connected to $x$ or any $x_j$.
To avoid confusion, we name sets and sequences using uppercase and individual elements using lowercase.

Sequence concatenation can be used to extend or shrink existing sequences.
We include a shorthand notation for sequence concatenation, concatenating $S$ with $S'$ as follows ${S \concat S'} \equiv \mapbuild{i}{S_i}{i \in \mathbb{N} \land i < \size{S}}\mapbuild{j+\size{S}}{S'_j}{j \in \mathbb{N} \land j < \size{S'}}$.
To simplify extraction and addition of single elements, we denote single-element sequences $\seq{x}{i}{1}$ as $\single{x}$, where $x$ is the only value in the sequence.
The empty sequence is represented as $\varepsilon$.

\section{Stateful Dataflow}\label{sec:stateful-dataflow}
\newcommand{\integerEvent}[1]{\ensuremath{\mathtt{E}\langle#1\rangle}}
\newcommand{\resetEvent}{\ensuremath{\mathtt{Reset}}}

\emph{Stateful dataflow} systems, sometimes also called \emph{distributed dataflow}, \emph{dataflow streaming}, or \emph{stream processing} systems, are widely used for real-time processing of large amounts of streaming data.
This section informally introduces the stateful dataflow programming model and its failure recovery mechanism, which we formalize and prove correct in later sections.
It is mostly based on Apache Flink~\cite{carbone2015apache}, a stateful dataflow system,
however, the core concepts and techniques involved also apply to other similar systems~\cite{dean2008mapreduce,ZahariaCDDMMFSS12,zaharia2013discretized,akidau2015dataflow,jacques2016consistent,murray2013naiad,zeuch2019nebulastream,spenger2022portals}.

\subsection{A Taste of Programming in Stateful Dataflow}\label{sec:a-taste-of-programming-in-stateful-dataflow}

\begin{figure}[t]\centering
\ifx\hideTikz\undefined
\begin{tikzpicture}[
  scale=0.9, every node/.style={scale=0.9},
  thick, colorGrayBase, align=center,
  node distance=0.5cm and 1cm,
  task/.style={rectangle, rounded corners=1mm, very thick, draw,text=black,inner sep=2mm,fill=colorGrayBack},
  taskConnection/.style={-{Triangle[angle=45:5]}, very thick, controls={+(right:0.5) and +(left:0.5)}},
  inlineImage/.style={scale=0.35, yscale=-1, xshift=-12, yshift=-12}
]

\node[task,label={\footnotesize\emph{source}}] (t0) at (0,0) {integers};
\draw[taskConnection,{>[angle=45:5]}-{Triangle[angle=45:5]}]($(t0.west)+(-3,0)$) -- (t0.west) node[midway,above] {\footnotesize \integerEvent{5} \integerEvent{3} \integerEvent{1}};
\node[task,above=of t0,label={\footnotesize\emph{source}}] (t1) {reset};
\draw[taskConnection,{>[angle=45:5]}-{Triangle[angle=45:5]}]($(t1.west)+(-1.5,0)$) -- (t1.west) node[midway,above] {\footnotesize \resetEvent};

\node[task,right=of t0,label={\footnotesize\emph{task}}] (t2) {incremental\\average};
\draw[taskConnection](t0.east) to (t2.west);
\draw[taskConnection](t1.east) to ($(t2.west)+(0,0.3)$);

\node[task,right=of t2,label={\footnotesize\emph{sink}}] (t3) {result};
\draw[taskConnection](t2.east) to (t3.west);
\draw[taskConnection](t3.east) -- ($(t3.east)+(1.5,0)$);
\end{tikzpicture}
\else
hideTikz is set
\fi
\caption{Example stateful dataflow program calculating the incremental average of a data stream of integers. Another stream is used to transfer control messages resetting the state of the program.}
\label{fig:incremental-average}
\end{figure}
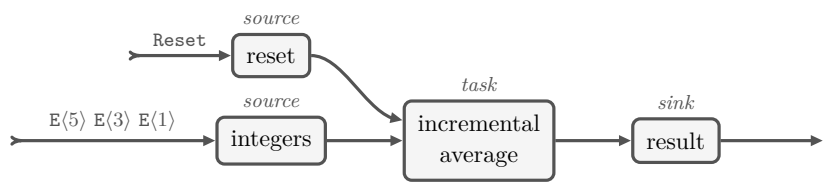

\Cref{fig:incremental-average} shows a stateful dataflow example calculating the incremental average of a stream of integers.
The example consists of two \emph{sources} ingesting streams of events into the system.
One source ingests a stream of integers \integerEvent{i}, and the other ingests a stream of \resetEvent~ events.
The term \emph{stream} can be understood as an unbounded sequence of events, it may in general continue forever.
The example also consists of a \emph{task}, an internal processing unit, which calculates an incremental average of the integers.
The incrementally computed averages are emitted to a \emph{sink}, which is the output of the system.

\begin{figure}[t!]\centering
\begin{lstlisting}[language=scala, label=lst:incremental-average, caption=A stateful dataflow program calculating the incremental average of a data stream of integers (see \Cref{fig:incremental-average}).]
Source(input = "src_reset", output = "reset")
Source(input = "src_ints", output = "ints")
Sink(inputs = { "avgs" }, output = "sink_avgs")
Task(inputs = { "src_ints", "src_reset" }, output = "avgs",
  f = (event, state) => event match {
    case Reset =>
      val new_state = {sum = 0, count = 0}
      return (Nil, new_state)
    case E$\langle$value$\rangle$ =>
      val new_state = {sum = state.sum + value, count = state.count + 1}
      val average = E$\langle$value = new_state.sum / new_state.count$\rangle$
      return (average : Nil, new_state) })
\end{lstlisting}
\end{figure}

A more detailed representation of the example is shown in \Cref{lst:incremental-average}.
Sources, tasks, and sinks are created using corresponding functions.
The API enables users to: (1) create sources, tasks, and sinks; (2) specify the connections in the graph by providing input and output streams; and (3) to specify how the tasks process events by providing their processing functions.
In this example, when the task receives an integer event \integerEvent{i}, it updates the average and emits the new average.
When it receives a \resetEvent~event, it resets its local state, such that the average is reset to its initial state.
To note is that the task is considered \emph{stateful}, as it maintains local state for its computation of the incremental average, even though the processing function $f$ is a pure function.
Also to note is that it is possible to provide an easier-to-use API above this core API, for example an API based on higher-order functions (\texttt{map}, \texttt{flatMap}, etc.)~\cite{carbone2015apache,spenger2022portals,akidau2015dataflow,murray2013naiad}.

\subsection{Failure Recovery via Asynchronous Barrier Snapshotting}\label{sec:failure-recovery-via-asynchronous-barrier-snapshotting}

Failure recovery is a crucial aspect of stateful dataflow systems.
In this section, we describe the failure recovery mechanism of the Asynchronous Barrier Snapshotting (ABS) protocol~\cite{carbone2015lightweight} as used in Apache Flink.
More specifically, ABS is a \emph{distributed snapshotting} protocol~\cite{chandy1985distributed} which is used for the checkpointing-based rollback-recovery protocol~\cite{elnozahy2002survey} within Apache Flink~\cite{carbone2015lightweight}.
After a failure, a checkpointing-based recovery will restart the system from the latest valid snapshot of the system~\cite{elnozahy2002survey}.

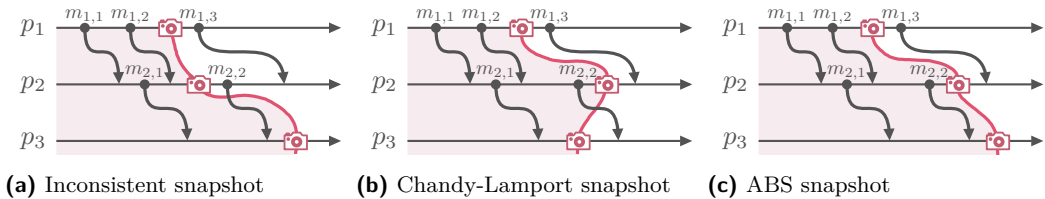
\begin{figure}[t]\centering
\begin{subfigure}{0.33\textwidth}
\centering
\ifx\hideTikz\undefined
\begin{tikzpicture}[
  scale=0.75, every node/.style={scale=1},
  thick, colorGrayBase,
  timeline/.style={-{Triangle[angle=45:5]}, thick},
  message/.style={-{Triangle[angle=45:5]}, very thick, controls={+(down:0.9) and +(up:1.1)}},
  messageNode/.style={scale=0.4, circle, fill, prefix after command={\pgfextra{\tikzset{every label/.style={label distance=-1.5mm, scale=0.75}}}}},
  snapshot/.style={scale=0.35, yscale=-1, xshift=-12, yshift=-12},
  snapshotEdge/.style={colorRedEmph, very thick, handdrawn}
]

\coordinate (s1) at (2,2);
\coordinate (s2) at (2.5,1);
\coordinate (s3) at (4.2,0);

\begin{scope}
\clip (0,-0.25)--(0,2)--(s1) .. controls +(down:0.5) and +(up:0) .. (s2) .. controls +(down:0.5) and +(up:1.25) .. (s3)--($(s3)+(0,-0.25)$)--(0,-0.25);
\fill[colorRedBack] (0,-0.25)--(0,2.1)--(5,2.1)--(5,-0.25)--(0,-0.25);
\end{scope}

\draw[snapshotEdge](s1) .. controls +(down:0.5) and +(up:0) .. (s2) .. controls +(down:0.5) and +(up:1.25) .. (s3)--($(s3)+(0,-0.25)$);

\node[label={left:{$p_1$}},scale=0] at (0,2) {};
\draw[timeline](0,2) -- (5,2);
\node[label={left:{$p_2$}},scale=0] at (0,1) {};
\draw[timeline](0,1) -- (5,1);
\node[label={left:{$p_3$}},scale=0] at (0,0) {};
\draw[timeline](0,0) -- (5,0);

\node[messageNode, label={$m_{1,1}$}] (m11) at (0.5,2) {};
\draw[message](m11.center) to (1.1,1);
\node[messageNode, label={$m_{2,1}$}] (m21) at (1.55,1) {};
\draw[message](m21.center) to (2.3,0);
\node[messageNode, label={$m_{1,2}$}] (m12) at (1.3,2) {};
\draw[message](m12.center) to (2,1);
\node[messageNode, label={$m_{1,3}$}] (m13) at (2.5,2) {};
\draw[message](m13.center) to (4,1);
\node[messageNode, label={$m_{2,2}$}] (m22) at (3,1) {};
\draw[message](m22.center) to (3.7,0);

\node[snapshot] at (s1) {\snapshotImage{colorRedBase}};
\node[snapshot] at (s2) {\snapshotImage{colorRedBase}};
\node[snapshot] at (s3) {\snapshotImage{colorRedBase}};
\end{tikzpicture}
\else
hideTikz is set
\fi
\caption{Inconsistent snapshot}\label{fig:snapshots:inconsistent}
\end{subfigure}%
\begin{subfigure}{0.33\textwidth}
\centering
\ifx\hideTikz\undefined
\begin{tikzpicture}[
  scale=0.75, every node/.style={scale=1},
  thick, colorGrayBase,
  timeline/.style={-{Triangle[angle=45:5]}, thick},
  message/.style={-{Triangle[angle=45:5]}, very thick, controls={+(down:0.9) and +(up:1.1)}},
  messageNode/.style={scale=0.4, circle, fill, prefix after command={\pgfextra{\tikzset{every label/.style={label distance=-1.5mm, scale=0.75}}}}},
  snapshot/.style={scale=0.35, yscale=-1, xshift=-12, yshift=-12},
  snapshotEdge/.style={colorRedEmph, very thick, handdrawn}
]

\coordinate (s1) at (2,2);
\coordinate (s2) at (3.5,1);
\coordinate (s3) at (3,0);

\begin{scope}
\clip (0,-0.25)--(0,2)--(s1) .. controls +(-0.2,-0.9) and +(0.2,0.7) .. (s2) .. controls +(down:0.5) and +(up:0.5) .. (s3)--($(s3)+(0,-0.25)$)--(0,-0.25);
\fill[colorRedBack] (0,-0.25)--(0,2.1)--(5,2.1)--(5,-0.25)--(0,-0.25);
\end{scope}

\draw[snapshotEdge](s1) .. controls +(-0.2,-0.9) and +(0.2,0.7) .. (s2) .. controls +(down:0.5) and +(up:0.5) .. (s3)--($(s3)+(0,-0.25)$);

\node[label={left:{$p_1$}},scale=0] at (0,2) {};
\draw[timeline](0,2) -- (5,2);
\node[label={left:{$p_2$}},scale=0] at (0,1) {};
\draw[timeline](0,1) -- (5,1);
\node[label={left:{$p_3$}},scale=0] at (0,0) {};
\draw[timeline](0,0) -- (5,0);

\node[messageNode, label={$m_{1,1}$}] (m11) at (0.5,2) {};
\draw[message](m11.center) to (1.1,1);
\node[messageNode, label={$m_{2,1}$}] (m21) at (1.55,1) {};
\draw[message](m21.center) to (2.3,0);
\node[messageNode, label={$m_{1,2}$}] (m12) at (1.3,2) {};
\draw[message](m12.center) to (2,1);
\node[messageNode, label={$m_{1,3}$}] (m13) at (2.5,2) {};
\draw[message](m13.center) to (4,1);
\node[messageNode, label={$m_{2,2}$}] (m22) at (3,1) {};
\draw[message](m22.center) to (3.7,0);

\node[snapshot] at (s1) {\snapshotImage{colorRedBase}};
\node[snapshot] at (s2) {\snapshotImage{colorRedBase}};
\node[snapshot] at (s3) {\snapshotImage{colorRedBase}};
\end{tikzpicture}
\else
hideTikz is set
\fi
\caption{Chandy-Lamport snapshot}\label{fig:snapshots:chandy-lamport}
\end{subfigure}%
\begin{subfigure}{0.33\textwidth}
\centering
\ifx\hideTikz\undefined
\begin{tikzpicture}[
  scale=0.75, every node/.style={scale=1},
  thick, colorGrayBase,
  timeline/.style={-{Triangle[angle=45:5]}, thick},
  message/.style={-{Triangle[angle=45:5]}, very thick, controls={+(down:0.9) and +(up:1.1)}},
  messageNode/.style={scale=0.4, circle, fill, prefix after command={\pgfextra{\tikzset{every label/.style={label distance=-1.5mm, scale=0.75}}}}},
  snapshot/.style={scale=0.35, yscale=-1, xshift=-12, yshift=-12},
  snapshotEdge/.style={colorRedEmph, very thick, handdrawn}
]

\coordinate (s1) at (2,2);
\coordinate (s2) at (3.5,1);
\coordinate (s3) at (4.2,0);

\begin{scope}
\clip (0,-0.25)--(0,2)--(s1) .. controls +(-0.2,-0.9) and +(0.2,0.7) .. (s2) .. controls +(down:0.5) and +(up:0.5) .. (s3)--($(s3)+(0,-0.25)$)--(0,-0.25);
\fill[colorRedBack] (0,-0.25)--(0,2.1)--(5,2.1)--(5,-0.25)--(0,-0.25);
\end{scope}

\draw[snapshotEdge](s1) .. controls +(-0.2,-0.9) and +(0.2,0.7) .. (s2) .. controls +(down:0.5) and +(up:0.5) .. (s3)--($(s3)+(0,-0.25)$);

\node[label={left:{$p_1$}},scale=0] at (0,2) {};
\draw[timeline](0,2) -- (5,2);
\node[label={left:{$p_2$}},scale=0] at (0,1) {};
\draw[timeline](0,1) -- (5,1);
\node[label={left:{$p_3$}},scale=0] at (0,0) {};
\draw[timeline](0,0) -- (5,0);

\node[messageNode, label={$m_{1,1}$}] (m11) at (0.5,2) {};
\draw[message](m11.center) to (1.1,1);
\node[messageNode, label={$m_{2,1}$}] (m21) at (1.55,1) {};
\draw[message](m21.center) to (2.3,0);
\node[messageNode, label={$m_{1,2}$}] (m12) at (1.3,2) {};
\draw[message](m12.center) to (2,1);
\node[messageNode, label={$m_{1,3}$}] (m13) at (2.5,2) {};
\draw[message](m13.center) to (4,1);
\node[messageNode, label={$m_{2,2}$}] (m22) at (3,1) {};
\draw[message](m22.center) to (3.7,0);

\node[snapshot] at (s1) {\snapshotImage{colorRedBase}};
\node[snapshot] at (s2) {\snapshotImage{colorRedBase}};
\node[snapshot] at (s3) {\snapshotImage{colorRedBase}};
\end{tikzpicture}
\else
hideTikz is set
\fi
\caption{ABS snapshot}\label{fig:snapshots:abs}
\end{subfigure}
\caption{Examples of snapshots obtained in a distributed stateful dataflow system with three processes $p_1 \rightarrow p_2 \rightarrow p_3$.}
\label{fig:snapshots}
\end{figure}

\subparagraph*{Distributed Snapshotting Protocols.}
A distributed snapshotting protocol is considered \emph{causally consistent} if it captures snapshots that do not violate causality~\cite{chandy1985distributed}.
Causality, here, refers to the causal order relation~\cite{DBLP:journals/cacm/Lamport78}, informally: two events are causally ordered if one event was part of a causal chain leading to the other event.
Consequently, a causally consistent snapshot captures the state of a system such that all events causally preceding any other event in the snapshot are included.
This definition is illustrated by three example executions of different snapshotting protocols for a dataflow graph consisting of three nodes, shown in \Cref{fig:snapshots}.
An incorrect implementation (\Cref{fig:snapshots:inconsistent}) would be to let the processes periodically capture a snapshot of their state without coordination.
A snapshot captured with this method can be inconsistent, thus not suitable for recovery, as it may violate causality.
In the example, the incorrect snapshot has captured that $m_{2,2}$ was received by $p_3$ but never sent by $p_2$, this is a violation of causality, and recovery from such a snapshot would be considered erroneous.
In contrast, consistent snapshotting protocols do not violate causality.
The Chandy-Lamport asynchronous snapshotting protocol~\cite{chandy1985distributed} (\Cref{fig:snapshots:chandy-lamport}) solves this issue through distributed coordination by means of disseminating markers during its regular execution, separating pre-snapshot and post-snapshot messages.
However, a snapshot captured with the Chandy-Lamport protocol may capture in-flight events: as shown in the example (\Cref{fig:snapshots:chandy-lamport}), the message $m_{2,2}$ was sent (according to $p_2$'s snapshot) but not yet received (according to $p_3$'s snapshot).
The Asynchronous Barrier Snapshotting (ABS) protocol~\cite{carbone2015lightweight,carbone2018thesis}, in contrast, captures complete distributed computations without in-flight events by modification of the marker-based Chandy-Lamport protocol.
As shown in \Cref{fig:snapshots:abs}, the snapshot does not include any in-flight events.

\begin{figure}[t]\centering
\begin{lstlisting}[language=eventhandler, label=lst:abs, caption={Representation of an event handler within a stateful dataflow system implementing failure recovery using the ABS protocol~\cite{carbone2015lightweight,carbone2018thesis}.}]
EventHandler Def $TK \langle f$, $\seq{S_i}{i}{n}, o \rangle$
  Vars state, snapshots
  On Event Receive < $S_j$, epoch, `\texttt{Event}`$\langle$w$\rangle$ > If $\exists \text{v}$: state = < epoch, v > Do
    v', w' = $f$(v, w')
    state = < epoch, v' >
    emit(< $o$, epoch, `\texttt{Event}`$\langle$w'$\rangle$ >)
  On Event Receive $\bigl[$< $S_i$, epoch, $\texttt{Border}$ >$\bigl]_i^n$ If $\exists \text{v}$: state = < epoch, v > Do
    snapshots.update(epoch $\mapsto$ v)
    state = < epoch + 1, v >
    emit(< $o$, epoch, `\texttt{Border}` >)
  On Event Fail Do
    state = `\text{Failed}`
  On Event Recover < recoverEpoch > Do
    state = < recoverEpoch, snapshots(recoverEpoch) >
\end{lstlisting}
\end{figure}

\subparagraph{The ABS Protocol.}
A representation of the ABS protocol~\cite{carbone2015lightweight,carbone2018thesis} corresponding to our formalization in \Cref{sec:implementation-model} is found in \Cref{lst:abs}.
The handler has two mutable states: the processing task's volatile \lstinline{state}, and the persistent \lstinline{snapshots} state.
The \lstinline{state} is a tuple \lstinline{< epoch, v >} consisting of the current \lstinline{epoch}'s sequence number being processed, and the state \lstinline{v} of the processing task.
The first event handler consumes an event \lstinline{Event$\langle$w$\rangle$} from a stream with stream name $S_j$ out of the sequence of stream names $S$ for some \lstinline{epoch} if it is not currently in a failed state.
It processes the event $w$ on its current state $v$, which produces an output event $w'$ and new state $v'$.
It then updates its mutable state, and emits the output on its outgoing stream with stream name $o$.
The second handler processes the \lstinline{Border} markers from the higher-level ABS protocol.
It will consume all border events from all its incoming streams in a single step.
In doing so, it will take a snapshot of the local state and update the epoch number, as well as disseminate the border marker further downstream.
To note is that the first handler does not consume from a stream if that stream has a border marker as its next event, instead it will block such streams until the border step (\ie the second handler) has been taken.
The first and second handlers implement the ABS protocol, whereas the third and fourth handlers implement the failure recovery.
The third handler models the random failures of tasks, a task can randomly fail at any time, in which case it loses its volatile state.
The fourth handler implements the failure recovery, and is triggered by some external coordinating instance once it has detected the failure.
Once a failure has been detected, all tasks are recovered to the same epoch which corresponds to the latest snapshot of the system.
When triggered, the fourth handler recovers the state back to the snapshot of the epoch found in the message.

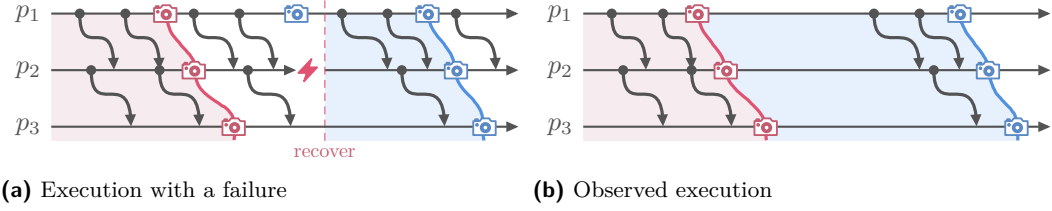
\begin{figure}[t]
\centering
\begin{subfigure}{0.5\textwidth}
\centering
\ifx\hideTikz\undefined
\begin{tikzpicture}[
  scale=0.75, every node/.style={scale=1},
  thick, colorGrayBase,
  timeline/.style={-{Triangle[angle=45:5]}, thick},
  message/.style={-{Triangle[angle=45:5]}, very thick, controls={+(down:0.9) and +(up:1.1)}},
  messageNode/.style={scale=0.4, circle, fill, prefix after command={\pgfextra{\tikzset{every label/.style={label distance=-1.5mm, scale=0.75}}}}},
  inlineImage/.style={scale=0.35, yscale=-1, xshift=-12, yshift=-12},
  snapshotEdge/.style={very thick, handdrawn}
]

\coordinate (s1) at (2,2);
\coordinate (s2) at (2.5,1);
\coordinate (s3) at (3.2,0);

\begin{scope}
\clip (0,-0.25)--(0,2)--(s1) .. controls +(down:0.5) and +(up:0.5) .. (s2) .. controls +(down:0.5) and +(up:0.5) .. (s3)--($(s3)+(0,-0.25)$)--(0,-0.25);
\fill[colorRedBack] (0,-0.25)--(0,2.1)--(5,2.1)--(5,-0.25)--(0,-0.25);
\end{scope}

\draw[colorRedEmph, snapshotEdge](s1) .. controls +(down:0.5) and +(up:0.5) .. (s2) .. controls +(down:0.5) and +(up:0.5) .. (s3)--($(s3)+(0,-0.25)$);

\coordinate (s12) at (7.2-0.6,2);
\coordinate (s22) at (7.7-0.6,1);
\coordinate (s32) at (8.2-0.6,0);

\begin{scope}
\clip (0,-0.25)--(0,2)--(s12) .. controls +(down:0.5) and +(up:0.5) .. (s22) .. controls +(down:0.5) and +(up:0.5) .. (s32)--($(s32)+(0,-0.25)$)--(0,-0.25);
\fill[colorBlueBack] (5.2-0.4,-0.25)--(5.2-0.4,2.1)--(8.2,2.1)--(8.2,-0.25)--(5.2-0.4,-0.25);
\end{scope}

\draw[colorBlueEmph, snapshotEdge](s12) .. controls +(down:0.5) and +(up:0.5) .. (s22) .. controls +(down:0.5) and +(up:0.5) .. (s32)--($(s32)+(0,-0.25)$);

\node[label={left:{$p_1$}},scale=0] at (0,2) {};
\draw[timeline](0,2) -- (8.2,2);
\node[label={left:{$p_2$}},scale=0] at (0,1) {};
\draw[timeline](0,1) -- (4.3,1);
\draw[timeline](5.2-0.4,1) -- (8.2,1);
\draw[densely dashed, thin, colorRedEmph] (5.2-0.4,-0.25) -- (5.2-0.4,2.25);
\node[label={[scale=0.75, colorRedBase]below:{recover}},scale=0] at (5.2-0.4,-0.25) {};
\node[label={left:{$p_3$}},scale=0] at (0,0) {};
\draw[timeline](0,0) -- (8.2,0);

\node[messageNode] (m11) at (0.5,2) {};
\draw[message](m11.center) to (1.1,1);
\node[messageNode] (m21) at (0.7,1) {};
\draw[message](m21.center) to (1.4,0);
\node[messageNode] (m12) at (1.3,2) {};
\draw[message](m12.center) to (1.9,1);
\node[messageNode] (m22p) at (1.9,1) {};
\draw[message](m22p.center) to (2.6,0);

\node[messageNode] (m13) at (2.5,2) {};
\draw[message](m13.center) to (3.2,1);
\node[messageNode] (m22) at (3.45,1) {};
\draw[message](m22.center) to (4.2,0);
\node[messageNode] (m14) at (3.3,2) {};
\draw[message](m14.center) to (3.9,1);
\node[messageNode] (m112) at (5.7-0.6,2) {};
\draw[message](m112.center) to (6.3-0.4,1);
\node[messageNode] (m212) at (6.75-0.6,1) {};
\draw[message](m212.center) to (7.5-0.6,0);
\node[messageNode] (m122) at (6.5-0.6,2) {};
\draw[message](m122.center) to (7.2-0.6,1);
\node[messageNode] (m132) at (7.7-0.6,2) {};
\draw[message](m132.center) to (8.4-0.6,1);

\node[inlineImage] at (s1) {\snapshotImage{colorRedBase}};
\node[inlineImage] at (s2) {\snapshotImage{colorRedBase}};
\node[inlineImage] at (s3) {\snapshotImage{colorRedBase}};
\node[inlineImage] at (4.3,2) {\snapshotImage{colorBlueBase}};
\node[inlineImage] at (s12) {\snapshotImage{colorBlueBase}};
\node[inlineImage] at (s22) {\snapshotImage{colorBlueBase}};
\node[inlineImage] at (s32) {\snapshotImage{colorBlueBase}};

\node[inlineImage] at (4.5,1) {\failureImage};
\end{tikzpicture}
\else
hideTikz is set
\fi
\caption{Execution with a failure}\label{fig:observation:failed}
\end{subfigure}%
\begin{subfigure}{0.5\textwidth}
\centering
\ifx\hideTikz\undefined
\begin{tikzpicture}[
  scale=0.75, every node/.style={scale=1},
  thick, colorGrayBase,
  timeline/.style={-{Triangle[angle=45:5]}, thick},
  message/.style={-{Triangle[angle=45:5]}, very thick, controls={+(down:0.9) and +(up:1.1)}},
  messageNode/.style={scale=0.4, circle, fill, prefix after command={\pgfextra{\tikzset{every label/.style={label distance=-1.5mm, scale=0.75}}}}},
  inlineImage/.style={scale=0.35, yscale=-1, xshift=-12, yshift=-12},
  snapshotEdge/.style={very thick, handdrawn}
]

\coordinate (s12) at (7.2-0.6,2);
\coordinate (s22) at (7.7-0.6,1);
\coordinate (s32) at (8.2-0.6,0);

\begin{scope}
\clip (0,-0.25)--(0,2)--(s12) .. controls +(down:0.5) and +(up:0.5) .. (s22) .. controls +(down:0.5) and +(up:0.5) .. (s32)--($(s32)+(0,-0.25)$)--(0,-0.25);
\fill[colorBlueBack] (0,-0.25)--(0,2.1)--(8.2,2.1)--(8.2,-0.25)--(0,-0.25);
\end{scope}

\draw[colorBlueEmph, snapshotEdge](s12) .. controls +(down:0.5) and +(up:0.5) .. (s22) .. controls +(down:0.5) and +(up:0.5) .. (s32)--($(s32)+(0,-0.25)$);

\coordinate (s1) at (2,2);
\coordinate (s2) at (2.5,1);
\coordinate (s3) at (3.2,0);

\begin{scope}
\clip (0,-0.25)--(0,2)--(s1) .. controls +(down:0.5) and +(up:0.5) .. (s2) .. controls +(down:0.5) and +(up:0.5) .. (s3)--($(s3)+(0,-0.25)$)--(0,-0.25);
\fill[colorRedBack] (0,-0.25)--(0,2.1)--(5,2.1)--(5,-0.25)--(0,-0.25);
\end{scope}

\draw[colorRedEmph, snapshotEdge](s1) .. controls +(down:0.5) and +(up:0.5) .. (s2) .. controls +(down:0.5) and +(up:0.5) .. (s3)--($(s3)+(0,-0.25)$);

\node[label={left:{$p_1$}},scale=0] at (0,2) {};
\draw[timeline](0,2) -- (8.2,2);
\node[label={left:{$p_2$}},scale=0] at (0,1) {};
\draw[timeline](0,1) -- (8.2,1);
\node[label={left:{$p_3$}},scale=0] at (0,0) {};
\draw[timeline](0,0) -- (8.2,0);

\node[messageNode] (m11) at (0.5,2) {};
\draw[message](m11.center) to (1.1,1);
\node[messageNode] (m21) at (0.7,1) {};
\draw[message](m21.center) to (1.4,0);
\node[messageNode] (m12) at (1.4,2) {};
\draw[message](m12.center) to (1.9,1);
\node[messageNode] (m22p) at (1.9,1) {};
\draw[message](m22p.center) to (2.6,0);

\node[messageNode] (m112) at (5.7-0.6,2) {};
\draw[message](m112.center) to (6.3-0.4,1);
\node[messageNode] (m212) at (6.75-0.6,1) {};
\draw[message](m212.center) to (7.5-0.6,0);
\node[messageNode] (m122) at (6.5-0.6,2) {};
\draw[message](m122.center) to (7.2-0.6,1);

\node[inlineImage] at (s1) {\snapshotImage{colorRedBase}};
\node[inlineImage] at (s2) {\snapshotImage{colorRedBase}};
\node[inlineImage] at (s3) {\snapshotImage{colorRedBase}};
\node[inlineImage] at (s12) {\snapshotImage{colorBlueBase}};
\node[inlineImage] at (s22) {\snapshotImage{colorBlueBase}};
\node[inlineImage] at (s32) {\snapshotImage{colorBlueBase}};

\node[label={[scale=0.75]below:{\phantom{recover}}},scale=0] at (5.2-0.4,-0.25) {};
\end{tikzpicture}
\else
hideTikz is set
\fi
\caption{Observed execution}\label{fig:observation:observed}
\end{subfigure}
\caption{An execution with failures and its observed execution.}\label{fig:observation}
\end{figure}

\subparagraph*{Failure Recovery.}
The dataflow system can recover from failures using the ABS protocol.
\Cref{fig:observation:failed} shows an execution using the ABS protocol in which $p_2$ fails.
The coordinator (not displayed) will eventually discover the failure, and trigger a synchronous recovery step in which all processes recover to the latest completed snapshot and continue processing from there.
Even though failures occur in the execution, the observer will be able to construct an idealized execution corresponding to our notion of failure transparency in which there are no failed events or incomplete epochs as shown in \Cref{fig:observation:observed}.
This is, loosely speaking, achieved by ignoring the side effects from the failed epochs, and is explored in detail in \Cref{sec:failure-transparency-of-the-implementation}.

\section{Implementation Model}\label{sec:implementation-model}
We now provide a formal model of the stateful dataflow system described above.
The goal of this formalization is to capture and analyze key aspects of the implementation of the system, with focus on its failure recovery using the Asynchronous Barrier Snapshotting protocol~\cite{carbone2015lightweight,carbone2018thesis}.
The formal model is presented in two parts: the first part presents an explicit evaluation rule for message passing, and the second part presents the evaluation rules for processing and failure recovery.

\subsection{Streaming Model}\label{sec:common-streaming-model}

The streaming model is based on processors (or tasks) that communicate via streams.
A processor is a stateful entity that may consume an event from an incoming stream, process it, and produce events to its outgoing stream.
Streams, in turn, transport the events between processors in a FIFO order.
With this notion of processors and streams, we can execute computational graphs by means of steps.
Note that, in this section, we discuss a general streaming model, leaving the implementation of processors abstract.
Whereas, in the next section, we discuss concrete implementations of processors.

\begin{figure}
\centering
\fcolorbox{colorGrayBase}{white}{\begin{minipage}{\textwidth-2\fboxsep-2\fboxrule}%
\begin{minipage}{\textwidth}%
\centering$\begin{array}[t]{c@{\hspace{10mm}}c@{\hspace{10mm}}c}
p,~q ~{~}~{~}\text{processor ID}
&s,~o ~{~}~{~}\text{stream name}
&n \in \mathbb{N} ~{~}~{~}\text{sequence number}
\end{array}$\end{minipage}\\\begin{minipage}{\textwidth}
\[\begin{array}[t]{c@{\hspace{10mm}}c@{\hspace{10mm}}c@{\hspace{10mm}}c}
\pi ~{~}~{~}\mbox{processor}
&\sigma ~{~}~{~}\mbox{state}
&d ~{~}~{~}\mbox{message data}
&D ~{~}~{~}\mbox{auxiliary data}
\end{array}\]\end{minipage}\\[1mm]\begin{minipage}{\textwidth}
\begin{minipage}{0.59\textwidth}
\[\begin{array}[t]{r@{\hspace{1mm}}r@{\hspace{1mm}}l@{\hspace{2mm}}l}
\Pi &::=& \seq{\pi}{p}{\size{\Pi}} &\mbox{processors}\\[0.1cm]
\Sigma &::=& \seq{\sigma}{p}{\size{\Pi}} &\mbox{states}\\[0.1cm]
M &::=& \seq{m}{i}{\size{M}} &\mbox{messages}\\[0.1cm]
N &::=& \seq{N_p}{p}{\size{\Pi}} &\mbox{sequence numbers}\\[0.1cm]
N_p &::=& \mapbuild{s}{n}{s} &\mbox{sequence numbers of $p$}
\end{array}\]\end{minipage}%
\begin{minipage}{0.41\textwidth}
\[\begin{array}[t]{r@{\hspace{1mm}}r@{\hspace{1mm}}l@{\hspace{2mm}}l}
X &::=& \seq{x}{i}{\size{X}} &\mbox{actions}\\[0.1cm]
x &::=& &\mbox{action}\\
& &+\,s\,d     &\mbox{\small\it production}\\
&\mid &-\,s\,d &\mbox{\small\it consumption}\\[0.1cm]
m &::=& n\,s\,{d}    &\mbox{message}
\end{array}\]\end{minipage}\end{minipage}%
\end{minipage}}
\caption{Streaming syntax.}\label{fig:ssyntax}
\end{figure}

\subparagraph*{Syntax.}
\Cref{fig:ssyntax} shows the syntax of the streaming model.
A configuration $c = \langle\, \Pi,\, \Sigma,\, N,\\ M,\, D \,\rangle$ represents a point in an execution of a streaming program.
The processors $\Pi$ indexed by identifiers $p$ represent processor definitions, for which $\Sigma$ represents the states of the processors.
The messages $M$ are modeled as a sequence of all messages, for which a message $m$ corresponds to a tuple of a sequence number $n$, a stream name $s$, and the message data $d$.
The current sequence number from which a processor $p$ reads from or writes to a stream $s$ is represented by $N_p(s)$.
The sequence numbers for all processors are represented by $N$.
When a processor processes a message, it may produce and consume messages.
This production and consumption is represented by a sequence of actions $X$.
A production action producing message $d$ to stream $s$ has the form $+\, s\, d$, similar to the consumption action $-\, s\, d$.
The auxiliary data $D$ is used to store global and additional execution information which is specific to the models; for example, it can be used to implicitly model the global coordinator.
In the formalization here, the processor $\pi$, state $\sigma$, message data $d$ and auxiliary global data $D$ are seen as atomic values, that is, no information about their internal structure is provided.
These limitations permit reusing the same syntax and rule for different instantiations of $\pi$, $\sigma$, $d$ and $D$.

\begin{figure}\centering
\newcommand*{\tikzscale}{1.5}
\begin{subfigure}{0.33\textwidth}
\centering
\ifx\hideTikz\undefined
\begin{tikzpicture}[
  scale=\tikzscale, every node/.style={scale=\tikzscale},
  pin distance=3m,
  thick, colorGrayBase,
  node distance=0.3cm*\tikzscale and 1cm*\tikzscale,
  task/.style={rectangle, rounded corners=1mm, very thick, draw,text=black,inner sep=2mm,colorRedBase,fill=colorRedBack},
  event/.style={scale=0.45, circle, thick, draw, node on layer=foreground},
  border/.style={scale=0.275, rectangle, very thick, draw, colorRedEmph, fill=colorRedBack, sloped, rotate=90, node on layer=foreground},
  inlineImage/.style={scale=0.35, yscale=-1, xshift=-12, yshift=-12}
]
\begin{scope}[on background layer]
\draw[double distance=1.5mm*\tikzscale] (0,0) -- (2, 0) node[pos=0.25,scale=0] (p) {} node[pos=0.35,event,colorRedBase] (e2) {} node[pos=0.55,event,colorRedBase] (e1) {} node[pos=0.75,event,colorRedBase] (e0) {} node[pos=0.85,scale=0] (c1) {} node[pos=0.95,scale=0] (c2) {};
\draw[line width=1.5mm*\tikzscale,colorRedBack] (0,0) -- (2, 0);
\node[above=of p,scale=0.6] (pl) {$p$};
\node[below=of p,scale=0.5,yshift=4mm] {$3$};
\draw[densely dotted] (pl) -- ($(p.center)+(0,-0.75mm*\tikzscale)$);
\node[below=of e2.center,scale=0.5,yshift=4mm,colorRedBase] {$\mathbf{2}$};
\node[below=of e1.center,scale=0.5,yshift=4mm,colorRedBase] {$\mathbf{1}$};
\node[below=of e0.center,scale=0.5,yshift=4mm,colorRedBase] {$\mathbf{0}$};
\node[above=of c1,scale=0.6] (c1l) {$q_1$};
\node[below=of c1,scale=0.5,yshift=4mm] {$0$};
\draw[densely dotted] (c1l) -- ($(c1.center)+(0,-0.75mm*\tikzscale)$);
\node[above=of c2,scale=0.6] (c2l) {$q_2$};
\node[below=of c2,scale=0.5,yshift=4mm] {$0$};
\draw[densely dotted] (c2l) -- ($(c2.center)+(0,-0.75mm*\tikzscale)$);
\end{scope}
\end{tikzpicture}
\else
hideTikz is set
\fi
\caption{Initial state}\label{fig:streaming:initial}
\end{subfigure}%
\begin{subfigure}{0.33\textwidth}
\centering
\ifx\hideTikz\undefined
\begin{tikzpicture}[
  scale=\tikzscale, every node/.style={scale=\tikzscale},
  pin distance=3m,
  thick, colorGrayBase,
  node distance=0.3cm*\tikzscale and 1cm*\tikzscale,
  task/.style={rectangle, rounded corners=1mm, very thick, draw,text=black,inner sep=2mm,colorRedBase,fill=colorRedBack},
  event/.style={scale=0.45, circle, thick, draw, node on layer=foreground},
  border/.style={scale=0.275, rectangle, very thick, draw, colorRedEmph, fill=colorRedBack, sloped, rotate=90, node on layer=foreground},
  inlineImage/.style={scale=0.35, yscale=-1, xshift=-12, yshift=-12}
]
\begin{scope}[on background layer]
  \draw[double distance=1.5mm*\tikzscale] (0,0) -- (2, 0) node[pos=0.05,scale=0] (p) {} node[pos=0.15,event,colorRedBase] (e3) {} node[pos=0.35,event,colorRedBase] (e2) {} node[pos=0.55,event,colorRedBase] (e1) {} node[pos=0.75,event,colorRedBase] (e0) {} node[pos=0.85,scale=0] (c1) {} node[pos=0.95,scale=0] (c2) {};
\draw[line width=1.5mm*\tikzscale,colorRedBack] (0,0) -- (2, 0);
\node[above=of p,scale=0.6] (pl) {$p$};
\node[below=of p,scale=0.5,yshift=4mm] {$4$};
\draw[densely dotted] (pl) -- ($(p.center)+(0,-0.75mm*\tikzscale)$);
\node[below=of e3.center,scale=0.5,yshift=4mm,colorRedBase] {$\mathbf{3}$};
\node[below=of e2.center,scale=0.5,yshift=4mm,colorRedBase] {$\mathbf{2}$};
\node[below=of e1.center,scale=0.5,yshift=4mm,colorRedBase] {$\mathbf{1}$};
\node[below=of e0.center,scale=0.5,yshift=4mm,colorRedBase] {$\mathbf{0}$};
\node[above=of c1,scale=0.6] (c1l) {$q_1$};
\node[below=of c1,scale=0.5,yshift=4mm] {$0$};
\draw[densely dotted] (c1l) -- ($(c1.center)+(0,-0.75mm*\tikzscale)$);
\node[above=of c2,scale=0.6] (c2l) {$q_2$};
\node[below=of c2,scale=0.5,yshift=4mm] {$0$};
\draw[densely dotted] (c2l) -- ($(c2.center)+(0,-0.75mm*\tikzscale)$);
\end{scope}
\end{tikzpicture}
\else
hideTikz is set
\fi
\caption{Production by $p$}\label{fig:streaming:production}
\end{subfigure}%
\begin{subfigure}{0.33\textwidth}
\centering
\ifx\hideTikz\undefined
\begin{tikzpicture}[
  scale=\tikzscale, every node/.style={scale=\tikzscale},
  pin distance=3m,
  thick, colorGrayBase,
  node distance=0.3cm*\tikzscale and 1cm*\tikzscale,
  task/.style={rectangle, rounded corners=1mm, very thick, draw,text=black,inner sep=2mm,colorRedBase,fill=colorRedBack},
  event/.style={scale=0.45, circle, thick, draw, node on layer=foreground},
  border/.style={scale=0.275, rectangle, very thick, draw, colorRedEmph, fill=colorRedBack, sloped, rotate=90, node on layer=foreground},
  inlineImage/.style={scale=0.35, yscale=-1, xshift=-12, yshift=-12}
]
\begin{scope}[on background layer]
  \draw[double distance=1.5mm*\tikzscale] (0,0) -- (2, 0) node[pos=0.05,scale=0] (p) {} node[pos=0.15,event,colorRedBase] (e3) {} node[pos=0.35,event,colorRedBase] (e2) {} node[pos=0.55,event,colorRedBase] (e1) {} node[pos=0.65,scale=0] (c1) {} node[pos=0.75,event,colorRedBase] (e0) {} node[pos=0.95,scale=0] (c2) {};
\draw[line width=1.5mm*\tikzscale,colorRedBack] (0,0) -- (2, 0);
\node[above=of p,scale=0.6] (pl) {$p$};
\node[below=of p,scale=0.5,yshift=4mm] {$4$};
\draw[densely dotted] (pl) -- ($(p.center)+(0,-0.75mm*\tikzscale)$);
\node[below=of e3.center,scale=0.5,yshift=4mm,colorRedBase] {$\mathbf{3}$};
\node[below=of e2.center,scale=0.5,yshift=4mm,colorRedBase] {$\mathbf{2}$};
\node[below=of e1.center,scale=0.5,yshift=4mm,colorRedBase] {$\mathbf{1}$};
\node[above=of c1,scale=0.6] (c1l) {$q_1$};
\node[below=of c1,scale=0.5,yshift=4mm] {$1$};
\draw[densely dotted] (c1l) -- ($(c1.center)+(0,-0.75mm*\tikzscale)$);
\node[below=of e0.center,scale=0.5,yshift=4mm,colorRedBase] {$0$};
\node[above=of c2,scale=0.6] (c2l) {$q_2$};
\node[below=of c2,scale=0.5,yshift=4mm] {$0$};
\draw[densely dotted] (c2l) -- ($(c2.center)+(0,-0.75mm*\tikzscale)$);
\end{scope}
\end{tikzpicture}
\else
hideTikz is set
\fi
\caption{Consumption by $q_1$}\label{fig:streaming:consumption}
\end{subfigure}%
\caption{Production and consumption to/from a stream with a producer $p$ and consumers $q_1$ and $q_2$.}\label{fig:streaming}
\end{figure}
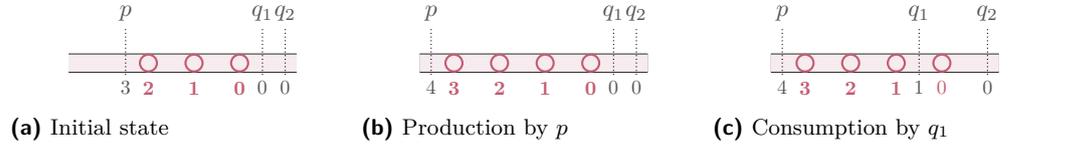

\Cref{fig:streaming} illustrates a stream as a sequence of messages with index numbers.
When producing an event to a stream (\Cref{fig:streaming:production}), the event is appended to the stream with an incremented index number.
This also increments the producer's index number for the stream from 3 to 4.
Similarly, the consumer's index number points to the next event to be consumed.
\Cref{fig:streaming:consumption} shows that the consumer $q_1$ has consumed the event 0, which in turn also increments its index number for the stream, pointing at the next event.
Consumers and producers process the stream independently and asynchronously.
The production of a message is a kind of broadcast, in the sense that all processors will have to consume it before consuming a newer message.

\subparagraph*{Step Rule.}
The streaming model essentially consists of a single rule ($\mathsc{S-Step}$) which describes the processing of messages.
Intuitively, a streaming step from configuration $\configuration{\Pi}{\Sigma}{N}{M}{D}$ can be taken if there is a local step with actions $X$, such that the actions are applicable.
A local step describes how the processor $\Pi_p$ changes its current state $\Sigma_p$ to its next state $\Sigma_p'$ using actions $X$.
The actions $X$ are applicable to $N_p$ and $M$ if all messages consumed by $X$ are available on the input streams of the processor.
The application of the actions $X$ results in $N_p'$ and $M'$, which are the incremented sequence numbers for the processor and the set of messages $M$ extended with the newly produced messages.
In case of taking a streaming step, the configuration transitions to the new configuration $\configuration{\Pi}{\Sigma\map{p}{\Sigma_p'}}{N\map{p}{N_p'}}{M'}{D}$.
In summary, the result of the streaming step is an update of the local state of the processor according to the local step, and an update of the sequence numbers and messages according to the actions $X$.
To simplify the analysis of streaming steps, auxiliary information about the processor ID, its sequence numbers, and the actions of the step is placed on the arrow of the execution step.
This information can be omitted when it is not needed by applying abstraction steps $\mathsc{S-AbsX}$ and $\mathsc{S-AbsP}$.

\[\infer*[Right=S-Step]{
    \Pi_p ~\Vdash~ \Sigma_p \xrightarrow{X} \Sigma_p'
\\
    X(N_p,\, M) = (N'_p,\, M')
}{
    \langle\, \Pi,\, \Sigma,\, N,\, M,\, D \,\rangle \xRightarrow[p]{N_p, X} \langle\, \Pi,\, \Sigma\map{p}{\Sigma_p'},\, N\map{p}{N'_p},\, M',\, D \,\rangle
}\]
\[\infer*[Right=S-AbsX]{
    c \xRightarrow[p]{N_p, X} c'
}{
    c \xRightarrow[p]{} c'
}\quad\quad\quad\quad\quad\quad\infer*[Right=S-AbsP]{
    c \xRightarrow[p]{} c'
}{
    c \Rightarrow c'
}\]

The streaming rule can be applied if there exists a derivation of the form $\Pi_p ~\Vdash~ \Sigma_p \xrightarrow{X} \Sigma_p'$ for a processor $\Pi_p$.
These are called local steps, since they have access only to the local data of
a processor, \ie its definition, state and locally accessible messages.
These rules describe the local step of a processor, in which the processor may produce and consume messages/actions $X$, and update its local state to $\Sigma_p'$.
The produced actions $X$ modify the sequence numbers of the processor $N_p$ and the messages in the system after application.
This is computed by the action application function $X(N_p,\, M)$ and results in the new sequence numbers $N'_p$ and messages $M'$ for the next configuration as defined below.

\subparagraph*{Action Application.}

The action application rule defines how actions modify the sequence numbers and messages.
A production action $+\,s\,d$ increases the sequence number of the stream $s$ for the producer, and adds the message to the sequence of messages.
Each stream has at most one producer; thus, we do not need to specify the producer in the action or message.
A consumption action $-\,s\,d$ increases the sequence number of the stream $s$ for the consumer, but does not remove it from the sequence of messages, as there may be other consumers waiting to consume the message.
To note is that the consumption action application is only defined if the message is present in the sequence of messages.
Due to this, local steps may only be applied in the context of the $\mathsc{S-Step}$ rule if the consumed message is present in the sequence of messages.
The remaining cases of the definition are for the recursive application of actions.

\begin{restatable}[Action Application]{definition}{actionApplication}\label{def:action-application}
\[\begin{array}[t]{r@{~}l l}
(+\,s\,d)(N_p,\,M)      &=  (N_p\map{s}{N_p(s)+1},\, M\cup\set{N_p(s)\,s\,d})\\
(-\,s\,d)(N_p,\,M)      &=  (N_p\map{s}{N_p(s)+1},\, M)~\text{if $N_p(s)\,s\,d \in M$, undefined otherwise}\\[0.1cm]
(\single{x}\concat X)(N_p,\,M) &= X(x(N_p,\,M))\\
\varepsilon(N_p,\,M)      &= (N_p,\,M)
\end{array}\]
\end{restatable}

According to the definition, it is not always possible to apply an action. This may be the case if, for example, a message for some sequence number is not yet available on its stream.
This enables indirectly ``passing'' messages to the local step rules.
Whereas the local step rule is defined for all possible steps for all messages that it may consume, cases in which the message consumption is not applicable by the action application definition are ruled out by the streaming global step rule.
This leaves only messages which are applicable to be applied to the steps, thus passing the message to the rule.

\subsection{Stateful Dataflow Model}

The presented stateful dataflow model consists of processing tasks, sources, and sinks.
A processing task consumes messages from a set of input streams, and produces messages on its output stream.
The task's behavior is defined by a function $f$ which processes the messages.
The function $f$ takes the task's state and an input message, and produces a new state and a sequence of output messages: $f(v, w) = v', \seq{W_i'}{i}{n}$.
The presented formal model does not provide a syntax and semantics for functions; they can be expressed using any suitable formalism.
The sources of the model are emulated by streams which are initialized in the first configuration to contain all the messages which are to be consumed from the source.
That is, each source is represented by its output stream, which in turn becomes an input to one of the tasks of the computational graph.
Sinks are also emulated as streams, however, in contrast to sources, they are initially empty.
The computation of the system, informally, takes inputs from the sources, processes them in the processing graph, and produces outputs to the sinks.

\begin{figure}[t]
\centering
\fcolorbox{colorGrayBase}{white}{\begin{minipage}{\textwidth-2\fboxsep-2\fboxrule}%
\begin{minipage}[t]{0.5\textwidth}%
\centering$\begin{array}[t]{r@{\hspace{1mm}}r@{\hspace{1mm}}l@{\hspace{2mm}}l}
& &v, ~w &\text{value}\\[0.3cm]
\pi &::=&\texttt{TK}\langle\, f,\,\sseq{S_i}{i}{\size{S}},\,o \,\rangle &\text{task}\\[0.1cm]
a &::=& \smapbuild{e}{v}{e} &\text{snapshot archive}\\[0.1cm]
\sigma &::=& \langle\, a,\, \sigma_\mathrm{V} \,\rangle &\text{state}\\[0.1cm]
\sigma_\mathrm{V} &::= &~&\mbox{volatile state}\\
& &\texttt{fl} &\mbox{\small\it failed state}\\
&\mid&\langle\, e,\, v \,\rangle &\mbox{\small\it normal state}
\end{array}$\end{minipage}%
\begin{minipage}[t]{0.5\textwidth}%
\centering$\begin{array}[t]{r@{\hspace{1mm}}r@{\hspace{1mm}}l@{\hspace{2mm}}l}
& &e\in\mathbb{N} &\text{epoch number}\\[0.3cm]
d &::=& \langle\, e,\, d_\mathrm{C} \,\rangle &\text{message}\\[0.1cm]
d_\mathrm{C} &::=& &\mbox{message cases}\\
& &\texttt{EV}\langle\, w \,\rangle &\mbox{\small\it event}\\
&\mid&\texttt{BD} &\mbox{\small\it epoch border}\\[0.1cm]
D &::=& M_0 &\text{initial input messages}
\end{array}$\end{minipage}%
\end{minipage}}
\caption{Stateful dataflow syntax.}\label{fig:low-level-syntax}
\end{figure}

\subparagraph*{Syntax.}
The syntax of the implementation model (\Cref{fig:low-level-syntax}) extends the shared streaming syntax and semantics (\Cref{fig:ssyntax}) by providing concrete instances of processors/tasks, messages, and state definitions.
A task $\texttt{TK}\langle\, f,\, S,\,o \,\rangle$ is a three-tuple of its processing function $f$, sequence of input streams $S$, and its output stream $o$.
Tasks process messages which are tuples of an epoch number $e$ and the message data $d_\mathrm{C}$.
There are two kinds of messages: normal events $\texttt{EV}\langle\, w \,\rangle$ and epoch borders $\texttt{BD}$.
The epoch border messages are markers used for the snapshotting algorithm, whereas the events are the actual data processed by the tasks.
When processing, the tasks manipulate state which consists of a persistent \emph{snapshot archive} $a$, \ie a map from epoch numbers to the corresponding local snapshots, and some \emph{volatile state} $\sigma_\mathrm{V}$.
The snapshot archive is a map from epoch numbers $e$ to the state $v$ of the processor at the end of the epoch.
The volatile state is either a \emph{failed state} $\texttt{fl}$ or a \emph{normal state} $\langle\, e,\, v \,\rangle$, consisting of the current epoch number and the state data value $v$ of the processor.
As with the messages, normal states are tagged by epoch numbers.
A processor is in a failed state if it has crashed and lost its volatile state.
The auxiliary data $D$ used for this model consists of the initial input messages for the system. As we may need to restore the messages which are yet to be consumed, we keep track of all the initial input messages as the global auxiliary data of the system.

\subsubsection{Derivation Rules}
The semantics of the model consists of seven rules.
Three of the rules, $\mathsc{I-Event}$, $\mathsc{I-Border}$, and $\mathsc{F-Fail}$, are local rules which enable deriving a local step of the form $\pi ~\Vdash~ \sigma \xrightarrow{X} \sigma'$.
Whereas the $\mathsc{I-Event}$ and $\mathsc{I-Border}$ rules model the processing of the system, the $\mathsc{F-Fail}$ rule models nondeterministic crash-failures of a processing task within the system.
These rules, together with the streaming rule $\mathsc{S-Step}$ and its abstraction rules $\mathsc{S-AbsX}$ and $\mathsc{S-AbsP}$, are used for deriving global steps.
The fourth rule, $\mathsc{F-Recover}$, is a global rule used for recovering the state of all processors after a failure.

\subparagraph*{Event Rule.}
The first rule, $\mathsc{I-Event}$, models tasks processing events:
\[\infer*[Right=I-Event]{
  f(v, w) = v', \seq{W'_i}{i}{n}
}{
  \texttt{TK}\langle\, f,\, S,\, o \,\rangle ~\Vdash~ \langle\, a,\, \langle\, e,\, v\,\rangle \,\rangle \xrightarrow{\ssingle{-\,S_j\,\langle\,e,\,\texttt{EV}\langle\, w \,\rangle\,\rangle} \concat \sseq{+\,o\,\langle\,e,\,\texttt{EV}\langle\, W'_i \,\rangle\,\rangle}{i}{n}} \langle\, a,\, \langle\, e,\, v'\,\rangle \,\rangle
}\]

The rule can perform a local step for a task $\texttt{TK}\langle\, f,\,\sseq{S_i}{i}{\size{S}},\,o \,\rangle$, if the current state of the task is a normal state $\langle\, e,\, v \,\rangle$, and the task can consume an event $\texttt{EV}\langle\, w \,\rangle$ from one of its inputs $S_j$.
Applying a task's function $f$ to its current state $v$ and the consumed event $w$ results in the task's next state $v'$ and a sequence of output events $\seq{W'_i}{i}{n}$.
The rule updates the state of the task to the new state $\langle\, e,\, v' \,\rangle$ and produces the output events $\sseq{\texttt{EV}\langle\, W'_i \,\rangle}{i}{n}$ on the output stream $o$.
The local step produces the actions which are the concatenation of the consumed and produced events.
For example, $\ssingle{-\,S_j\,\langle\,e,\,\texttt{EV}\langle\, w \,\rangle\,\rangle} \concat \ssingle{+\,o\,\langle\,e,\,\texttt{EV}\langle\, w' \,\rangle\,\rangle}$ is the action of consuming the event $\texttt{EV}\langle\, w \,\rangle$ with epoch number $e$ from the input stream $S_j$ and producing the event $\texttt{EV}\langle\, w' \,\rangle$ with epoch number $e$ on the output stream $o$.

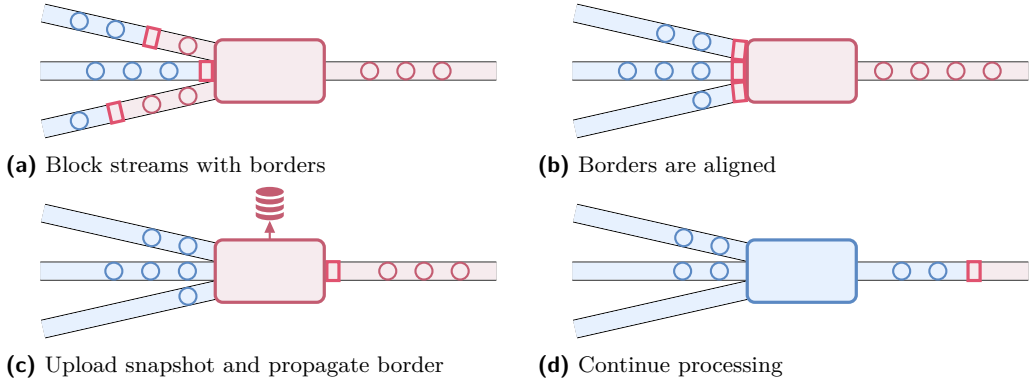
\begin{figure}\centering
  \newcommand*{\tikzscale}{1.5}
  \begin{subfigure}{0.5\textwidth}
  \centering
  \ifx\hideTikz\undefined
\begin{tikzpicture}[
  scale=\tikzscale, every node/.style={scale=\tikzscale},
  thick, colorGrayBase,
  node distance=0.5cm*\tikzscale and 1cm*\tikzscale,
  task/.style={rectangle, rounded corners=1mm, very thick, draw,text=black,inner sep=2mm,colorRedBase,fill=colorRedBack},
  event/.style={scale=0.45, circle, thick, draw, node on layer=foreground},
  border/.style={scale=0.275, rectangle, very thick, draw, colorRedEmph, fill=colorRedBack, sloped, rotate=90, node on layer=foreground},
  inlineImage/.style={scale=0.35, yscale=-1, xshift=-12, yshift=-12}
]
\begin{scope}[on background layer]
\node[scale=0.1] (s2) at (-2, 0) {};
\draw[double distance=1.5mm*\tikzscale] (s2.west) -- (-0.4,0) node[pos=0.3,event,colorBlueBase] {} node[pos=0.5,event,colorBlueBase] {} node[pos=0.7,event,colorBlueBase] {} node[pos=0.903,border] {\phantom{xx}};
\draw[line width=1.5mm*\tikzscale,colorBlueBack] (s2.west) -- (-0.4,0);
\node[scale=0.1,above=of s2] (s1) {};
\draw[double distance=1.5mm*\tikzscale] (s1.south east) -- (-0.4,0.15) node[pos=0.2,event,colorBlueBase] {} node[pos=0.4,event,colorBlueBase] {} node[pos=0.6,border] (bd1) {\phantom{xx}} node[pos=0.8,event,colorRedBase] {};
\draw[line width=1.5mm*\tikzscale,colorBlueBack] (s1.south east) -- (bd1);
\draw[line width=1.5mm*\tikzscale,colorRedBack] (bd1) -- (-0.4,0.15);
\node[scale=0.1,below=of s2] (s3) {};
\draw[double distance=1.5mm*\tikzscale] (s3.north east) -- (-0.4,-0.15) node[pos=0.2,event,colorBlueBase] {} node[pos=0.4,border] (bd2) {\phantom{xx}} node[pos=0.6,event,colorRedBase] {} node[pos=0.8,event,colorRedBase] {};
\draw[line width=1.5mm*\tikzscale,colorBlueBack] (s3.north east) -- (bd2);
\draw[line width=1.5mm*\tikzscale,colorRedBack] (bd2) -- (-0.4,-0.15);

\node[scale=0.1] (o) at (2, 0) {};
\draw[double distance=1.5mm*\tikzscale] (0.4,0) -- (o.west) node[pos=0.3,event,colorRedBase] {} node[pos=0.5,event,colorRedBase] {} node[pos=0.7,event,colorRedBase] {};
\draw[line width=1.5mm*\tikzscale,colorRedBack] (0.4,0) -- (o.west);

\node[task] (task) at (0,0) {\phantom{xxx}};
\end{scope}
\end{tikzpicture}
\else
hideTikz is set
\fi
  \caption{Block streams with borders}\label{fig:alignment:block}
  \end{subfigure}%
  \begin{subfigure}{0.5\textwidth}
  \centering
  \ifx\hideTikz\undefined
\begin{tikzpicture}[
  scale=\tikzscale, every node/.style={scale=\tikzscale},
  thick, colorGrayBase,
  node distance=0.5cm*\tikzscale and 1cm*\tikzscale,
  task/.style={rectangle, rounded corners=1mm, very thick, draw,text=black,inner sep=2mm,colorRedBase,fill=colorRedBack},
  event/.style={scale=0.45, circle, thick, draw, node on layer=foreground},
  border/.style={scale=0.275, rectangle, very thick, draw, colorRedEmph, fill=colorRedBack, sloped, rotate=90, node on layer=foreground},
  inlineImage/.style={scale=0.35, yscale=-1, xshift=-12, yshift=-12}
]
\begin{scope}[on background layer]
\node[scale=0.1] (s2) at (-2, 0) {};
\draw[double distance=1.5mm*\tikzscale] (s2.west) -- (-0.4,0) node[pos=0.3,event,colorBlueBase] {} node[pos=0.5,event,colorBlueBase] {} node[pos=0.7,event,colorBlueBase] {} node[pos=0.903,border] {\phantom{xx}};
\draw[line width=1.5mm*\tikzscale,colorBlueBack] (s2.west) -- (-0.4,0);
\node[scale=0.1,above=of s2] (s1) {};
\draw[double distance=1.5mm*\tikzscale] (s1.south east) -- (-0.4,0.15) node[pos=0.5,event,colorBlueBase] {} node[pos=0.7,event,colorBlueBase] {} node[pos=0.907,border,rotate=7] (bd1) {\phantom{xx}};
\draw[line width=1.5mm*\tikzscale,colorBlueBack] (s1.south east) -- (-0.4,0.15);
\node[scale=0.1,below=of s2] (s3) {};
\draw[double distance=1.5mm*\tikzscale] (s3.north east) -- (-0.4,-0.15) node[pos=0.7,event,colorBlueBase] {} node[pos=0.907,border,rotate=-7] (bd2) {\phantom{xx}};
\draw[line width=1.5mm*\tikzscale,colorBlueBack] (s3.north east) -- (-0.4,-0.15);

\node[scale=0.1] (o) at (2, 0) {};
\draw[double distance=1.5mm*\tikzscale] (0.4,0) -- (o.west) node[pos=0.2,event,colorRedBase] {} node[pos=0.4,event,colorRedBase] {} node[pos=0.6,event,colorRedBase] {} node[pos=0.8,event,colorRedBase] {};
\draw[line width=1.5mm*\tikzscale,colorRedBack] (0.4,0) -- (o.west);

\node[task] (task) at (0,0) {\phantom{xxx}};
\end{scope}
\end{tikzpicture}
\else
hideTikz is set
\fi
  \caption{Borders are aligned}\label{fig:alignment:aligned}
  \end{subfigure}
  \begin{subfigure}{0.5\textwidth}
  \centering
  \ifx\hideTikz\undefined
\begin{tikzpicture}[
  scale=\tikzscale, every node/.style={scale=\tikzscale},
  thick, colorGrayBase,
  node distance=0.5cm*\tikzscale and 1cm*\tikzscale,
  task/.style={rectangle, rounded corners=1mm, very thick, draw,text=black,inner sep=2mm,colorRedBase,fill=colorRedBack},
  event/.style={scale=0.45, circle, thick, draw, node on layer=foreground},
  border/.style={scale=0.275, rectangle, very thick, draw, colorRedEmph, fill=colorRedBack, sloped, rotate=90, node on layer=foreground},
  inlineImage/.style={scale=0.35, yscale=-1, xshift=-12, yshift=-12}
]
\begin{scope}[on background layer]
\node[scale=0.1] (s2) at (-2, 0) {};
\draw[double distance=1.5mm*\tikzscale] (s2.west) -- (-0.4,0) node[pos=0.4,event,colorBlueBase] {} node[pos=0.6,event,colorBlueBase] {} node[pos=0.8,event,colorBlueBase] {};
\draw[line width=1.5mm*\tikzscale,colorBlueBack] (s2.west) -- (-0.4,0);
\node[scale=0.1,above=of s2] (s1) {};
\draw[double distance=1.5mm*\tikzscale] (s1.south east) -- (-0.4,0.15) node[pos=0.6,event,colorBlueBase] {} node[pos=0.8,event,colorBlueBase] {};
\draw[line width=1.5mm*\tikzscale,colorBlueBack] (s1.south east) -- (-0.4,0.15);
\node[scale=0.1,below=of s2] (s3) {};
\draw[double distance=1.5mm*\tikzscale] (s3.north east) -- (-0.4,-0.15) node[pos=0.8,event,colorBlueBase] {};
\draw[line width=1.5mm*\tikzscale,colorBlueBack] (s3.north east) -- (-0.4,-0.15);

\node[scale=0.1] (o) at (2, 0) {};
\draw[double distance=1.5mm*\tikzscale] (0.4,0) -- (o.west) node[pos=0.097,border] (bd) {\phantom{xx}} node[pos=0.4,event,colorRedBase] {} node[pos=0.6,event,colorRedBase] {} node[pos=0.8,event,colorRedBase] {};
\draw[line width=1.5mm*\tikzscale,colorBlueBack] (0.4,0) -- (bd);
\draw[line width=1.5mm*\tikzscale,colorRedBack] (bd) -- (o.west);

\node[task,node on layer=foreground] (task) at (0,0) {\phantom{xxx}};
\node[scale=0] (storage) at ($(task.north)+(0,0.3)$) {};
\node[inlineImage] at (storage) {\storageImage{colorRedBase}};
\draw[-{Triangle[angle=45:5]}, thick, colorRedBase] (task.center) -- ($(storage)+(0,-0.15)$);
\end{scope}
\end{tikzpicture}
\else
hideTikz is set
\fi
  \caption{Upload snapshot and propagate border}\label{fig:alignment:snapshot}
  \end{subfigure}%
  \begin{subfigure}{0.5\textwidth}
  \centering
  \ifx\hideTikz\undefined
\begin{tikzpicture}[
  scale=\tikzscale, every node/.style={scale=\tikzscale},
  thick, colorGrayBase,
  node distance=0.5cm*\tikzscale and 1cm*\tikzscale,
  task/.style={rectangle, rounded corners=1mm, very thick, draw,text=black,inner sep=2mm,colorBlueBase,fill=colorBlueBack},
  event/.style={scale=0.45, circle, thick, draw, node on layer=foreground},
  border/.style={scale=0.275, rectangle, very thick, draw, colorRedEmph, fill=colorRedBack, sloped, rotate=90, node on layer=foreground},
  inlineImage/.style={scale=0.35, yscale=-1, xshift=-12, yshift=-12}
]
\begin{scope}[on background layer]
\node[scale=0.1] (s2) at (-2, 0) {};
\draw[double distance=1.5mm*\tikzscale] (s2.west) -- (-0.4,0) node[pos=0.6,event,colorBlueBase] {} node[pos=0.8,event,colorBlueBase] {};
\draw[line width=1.5mm*\tikzscale,colorBlueBack] (s2.west) -- (-0.4,0);
\node[scale=0.1,above=of s2] (s1) {};
\draw[double distance=1.5mm*\tikzscale] (s1.south east) -- (-0.4,0.15) node[pos=0.6,event,colorBlueBase] {} node[pos=0.8,event,colorBlueBase] {};
\draw[line width=1.5mm*\tikzscale,colorBlueBack] (s1.south east) -- (-0.4,0.15);
\node[scale=0.1,below=of s2] (s3) {};
\draw[double distance=1.5mm*\tikzscale] (s3.north east) -- (-0.4,-0.15);
\draw[line width=1.5mm*\tikzscale,colorBlueBack] (s3.north east) -- (-0.4,-0.15);

\node[scale=0.1] (o) at (2, 0) {};
\draw[double distance=1.5mm*\tikzscale] (0.4,0) -- (o.west) node[pos=0.3,event,colorBlueBase] {} node[pos=0.5,event,colorBlueBase] {} node[pos=0.7,border] (bd) {\phantom{xx}};
\draw[line width=1.5mm*\tikzscale,colorBlueBack] (0.4,0) -- (bd);
\draw[line width=1.5mm*\tikzscale,colorRedBack] (bd) -- (o.west);

\node[task] (task) at (0,0) {\phantom{xxx}};
\node[scale=0] (storage) at ($(task.north)+(0,0.3)$) {};
\node[inlineImage] at (storage) {\storageImage{white}};
\end{scope}
\end{tikzpicture}
\else
hideTikz is set
\fi
  \caption{Continue processing}\label{fig:alignment:continue}
  \end{subfigure}
  \caption{Epoch border alignment protocol (figure adapted from~\cite{carbone2018thesis}).}\label{fig:alignment}
\end{figure}

\subparagraph*{Border Rule.}
Whereas the event rule consumes a single event from a stream, the border rule ($\mathsc{I-Border}$) consumes one border event  $\texttt{BD}$ from \emph{every incoming stream}:
\[\infer*[Right=I-Border]{~}{
    \texttt{TK}\langle\, f,\,\sseq{S_i}{i}{n},\,o \,\rangle ~\Vdash~ \langle\, a,\, \langle\, e,\, v \,\rangle \,\rangle \xrightarrow{\sseq{-\,S_i\,\langle\, e,\,\texttt{BD} \,\rangle}{i}{n} \concat \ssingle{+\,o\,\langle\, e,\,\texttt{BD} \,\rangle}} \langle\, a\map{e}{v},\, \langle\, e+1,\, v \,\rangle \,\rangle
}\]

This consumption is enabled for a task if the next event to be consumed on every one of its incoming streams is a border event.
In other words, the event rule consumes events up until all streams are aligned by the border events, at which point the border rule consumes the border events from all its incoming streams.
The rule is a local step which, in addition to consuming border events from all incoming streams and producing a border event on its outgoing stream, stores the current state $v$ for epoch $e$ to the snapshot storage $a$ (by setting the new snapshot archive to $a[e \mapsto v]$), as well as incrementing the current epoch number.

Epochs are a key concept of Asynchronous Barrier Snapshotting.
Each epoch is a sequence of data-bearing \emph{events}, ending with an \emph{epoch border}, and are used to define the boundaries of state snapshots.
After regular processing for which some streams are blocked by border events  (\Cref{fig:alignment:block}), the rule aligns the streams by the borders (\Cref{fig:alignment:aligned}), takes a copy of the current state of the processor storing it to the snapshot archive (\Cref{fig:alignment:snapshot}), and propagates the epoch border message downstream and increments the epoch number, ready to process events from the next epoch (\Cref{fig:alignment:continue}).
The effect of this is that epochs of events are separated by the border events throughout the whole processing graph.

\subparagraph*{Failure Rule.}
Failures are introduced nondeterministically by the $\mathsc{F-Fail}$ rule:
\[\infer*[Right=F-Fail]{~}{
    \texttt{TK}\langle\, f, S, o \,\rangle ~\Vdash~ \langle\, a,\, \sigma_\mathrm{V} \,\rangle \rightarrow \langle\, a,\, \texttt{fl} \,\rangle
}\]

The failure rule sets the task's state to failed $\langle a, \texttt{fl} \rangle$, thus losing the task's volatile state.
Once a task is failed, it is no longer able to apply the steps $\mathsc{I-Event}$ and $\mathsc{I-Border}$, and will remain idle until the $\mathsc{F-Recover}$ rule has been applied.

\subparagraph*{Failure Recovery Rule.}
The last rule, $\mathsc{F-Recover}$, is a global rule which recovers the state of all failed tasks:
\[\infer*[Right=F-Recover]{
    \langle\, a,\, \texttt{fl} \,\rangle \in \Sigma
}{
    \langle\, \Pi,\, \Sigma,\, N,\, M,\, M_0 \,\rangle \Rightarrow \text{lcs}(\langle\, \Pi,\, \Sigma,\, N,\, M,\, M_0 \,\rangle)
}\]

The rule may be triggered nondeterministically if there exists a task in a failed state, and will reset the state of the system to the latest common snapshot.
The full details of how the latest common snapshot ($\text{lcs}$) is computed is discussed further below, as it depends on additional definitions.

The latest common snapshot is constructed by:
(1) calculating the greatest common epoch for which a snapshot has been taken by all processors in the system;
(2) restoring the state of all processors to their local snapshots at the greatest common epoch;
and (3) restoring sequence numbers and messages to undo any messages that were produced or consumed for epochs greater than the greatest common epoch.
The greatest common epoch is calculated by finding the minimum (\emph{common}) of the maximum (\emph{greatest}) epoch numbers of the local snapshots of all the processors.

\begin{restatable}[Greatest Common Epoch Number]{definition}{gceNumber}\label{def:gce-number}
The greatest common epoch number of a configuration $c = \langle\, \Pi,\, \Sigma,\, N,\, M,\, D \,\rangle$ is:
\[\mathrm{gce}(c) = \mathrm{min}\setbuild{\mathrm{max}(\mathrm{dom}(a))}{\Sigma_p = \langle\, a,\, \sigma_\mathrm{V} \,\rangle}\]
\end{restatable}

The persistent \emph{output messages} of the system consist of all messages produced up to and including the greatest common epoch.
These messages can be identified by comparing their epoch number $e$ to the greatest common epoch number $e \leq \text{gce}(c)$.
The recovery purges any messages which are not part of this set, bar the initial input messages $M_0$, thereby making these output messages (identified by $\mathrm{out}$) persistent.

\begin{restatable}[Output Messages]{definition}{outputMessages}\label{def:output-messages}
For a configuration $c = \langle \Pi, \Sigma, N, M, D \rangle$, its output messages are:
\[
\mathrm{out}(c) = \setbuild{n\,s\,\langle\, e,\, d \,\rangle}{(n\,s\,\langle\, e,\, d \,\rangle) \in M \land e \leq \mathrm{gce}(c)}
\]
\end{restatable}

\begin{definition}[Messages on a Stream]
The subset $M\downarrow s$ of messages on a particular stream is defined as:
\[M\downarrow s = \setbuild{n'\,s'\,d'}{(n'\,s'\,d') \in M \land s' = s}\]
\end{definition}

The $lcs$ function computes the latest common snapshot of a configuration for use as a recovery point in the $\mathsc{F-Recover}$ rule.
Its computation makes use of the greatest common epoch number (gce), and the output messages (out).
The states $\Sigma'$ are restored by removing any stored snapshots with an epoch number larger than the gce, and the volatile states are restored to the states captured by the snapshot of the gce.
The messages are updated to only keep the stable output messages $\text{out}(c)$ and the messages which are yet to be consumed $M_\mathrm{in}$.
The sequence numbers $N'$ are updated accordingly, setting the sequence number of a processor $p$ for a stream $s$ to the number of messages that the processor has either produced or consumed on the stream: $|\mathrm{out}(c)\downarrow s|$. Its complete definition is given below.

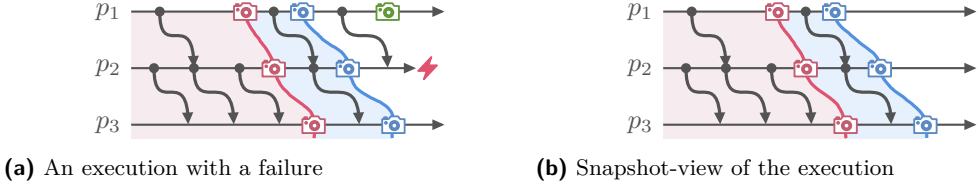
\begin{figure}[t]\centering
\begin{subfigure}{0.5\textwidth}
\centering
\ifx\hideTikz\undefined
\begin{tikzpicture}[
  scale=0.75, every node/.style={scale=1},
  thick, colorGrayBase,
  timeline/.style={-{Triangle[angle=45:5]}, thick},
  message/.style={-{Triangle[angle=45:5]}, very thick, controls={+(down:0.9) and +(up:1.1)}},
  messageNode/.style={scale=0.4, circle, fill, prefix after command={\pgfextra{\tikzset{every label/.style={label distance=-1.5mm, scale=0.75}}}}},
  inlineImage/.style={scale=0.35, yscale=-1, xshift=-12, yshift=-12},
  snapshotEdge/.style={very thick, handdrawn}
]
\coordinate (s12) at (3,2);
\coordinate (s22) at (3.8,1);
\coordinate (s32) at (4.6,0);

\begin{scope}
\clip (0,-0.25)--(0,2)--(s12) .. controls +(down:0.5) and +(up:0.5) .. (s22) .. controls +(down:0.5) and +(up:0.5) .. (s32)--($(s32)+(0,-0.25)$)--(0,-0.25);
\fill[colorBlueBack] (0,-0.25)--(0,2.1)--(8.2,2.1)--(8.2,-0.25)--(0,-0.25);
\end{scope}

\draw[colorBlueEmph, snapshotEdge](s12) .. controls +(down:0.5) and +(up:0.5) .. (s22) .. controls +(down:0.5) and +(up:0.5) .. (s32)--($(s32)+(0,-0.25)$);

\coordinate (s1) at (2,2);
\coordinate (s2) at (2.5,1);
\coordinate (s3) at (3.2,0);

\begin{scope}
\clip (0,-0.25)--(0,2)--(s1) .. controls +(down:0.5) and +(up:0.5) .. (s2) .. controls +(down:0.5) and +(up:0.5) .. (s3)--($(s3)+(0,-0.25)$)--(0,-0.25);
\fill[colorRedBack] (0,-0.25)--(0,2.1)--(5,2.1)--(5,-0.25)--(0,-0.25);
\end{scope}

\draw[colorRedEmph, snapshotEdge](s1) .. controls +(down:0.5) and +(up:0.5) .. (s2) .. controls +(down:0.5) and +(up:0.5) .. (s3)--($(s3)+(0,-0.25)$);

\node[label={left:{$p_1$}},scale=0] at (0,2) {};
\draw[timeline](0,2) -- (5.5,2);
\node[label={left:{$p_2$}},scale=0] at (0,1) {};
\draw[timeline](0,1) -- (5,1);
\node[label={left:{$p_3$}},scale=0] at (0,0) {};
\draw[timeline](0,0) -- (5.5,0);

\node[messageNode] (m11) at (0.5,2) {};
\draw[message](m11.center) to (1.1,1);
\node[messageNode] (m20) at (0.4,1) {};
\draw[message](m20.center) to (1,0);
\node[messageNode] (m21) at (1.1,1) {};
\draw[message](m21.center) to (1.8,0);
\node[messageNode] (m22p) at (1.9,1) {};
\draw[message](m22p.center) to (2.6,0);

\node[messageNode] (m13) at (2.5,2) {};
\draw[message](m13.center) to (3.2,1);
\node[messageNode] (m22) at (3.2,1) {};
\draw[message](m22.center) to (4,0);

\node[messageNode] (m3) at (3.7,2) {};
\draw[message](m3.center) to (4.5,1);

\node[inlineImage] at (s1) {\snapshotImage{colorRedBase}};
\node[inlineImage] at (s2) {\snapshotImage{colorRedBase}};
\node[inlineImage] at (s3) {\snapshotImage{colorRedBase}};
\node[inlineImage] at (s12) {\snapshotImage{colorBlueBase}};
\node[inlineImage] at (s22) {\snapshotImage{colorBlueBase}};
\node[inlineImage] at (s32) {\snapshotImage{colorBlueBase}};

\node[inlineImage] at (4.5,2) {\snapshotImage{colorGreenBase}};

\node[inlineImage] at (5.2,1) {\failureImage};
\end{tikzpicture}
\else
hideTikz is set
\fi
\caption{An execution with a failure}\label{fig:snapshot-view:original}
\end{subfigure}%
\begin{subfigure}{0.5\textwidth}
\centering
\ifx\hideTikz\undefined
\begin{tikzpicture}[
  scale=0.75, every node/.style={scale=1},
  thick, colorGrayBase,
  timeline/.style={-{Triangle[angle=45:5]}, thick},
  message/.style={-{Triangle[angle=45:5]}, very thick, controls={+(down:0.9) and +(up:1.1)}},
  messageNode/.style={scale=0.4, circle, fill, prefix after command={\pgfextra{\tikzset{every label/.style={label distance=-1.5mm, scale=0.75}}}}},
  inlineImage/.style={scale=0.35, yscale=-1, xshift=-12, yshift=-12},
  snapshotEdge/.style={very thick, handdrawn}
]
\coordinate (s12) at (3,2);
\coordinate (s22) at (3.8,1);
\coordinate (s32) at (4.6,0);

\begin{scope}
\clip (0,-0.25)--(0,2)--(s12) .. controls +(down:0.5) and +(up:0.5) .. (s22) .. controls +(down:0.5) and +(up:0.5) .. (s32)--($(s32)+(0,-0.25)$)--(0,-0.25);
\fill[colorBlueBack] (0,-0.25)--(0,2.1)--(8.2,2.1)--(8.2,-0.25)--(0,-0.25);
\end{scope}

\draw[colorBlueEmph, snapshotEdge](s12) .. controls +(down:0.5) and +(up:0.5) .. (s22) .. controls +(down:0.5) and +(up:0.5) .. (s32)--($(s32)+(0,-0.25)$);

\coordinate (s1) at (2,2);
\coordinate (s2) at (2.5,1);
\coordinate (s3) at (3.2,0);

\begin{scope}
\clip (0,-0.25)--(0,2)--(s1) .. controls +(down:0.5) and +(up:0.5) .. (s2) .. controls +(down:0.5) and +(up:0.5) .. (s3)--($(s3)+(0,-0.25)$)--(0,-0.25);
\fill[colorRedBack] (0,-0.25)--(0,2.1)--(5,2.1)--(5,-0.25)--(0,-0.25);
\end{scope}

\draw[colorRedEmph, snapshotEdge](s1) .. controls +(down:0.5) and +(up:0.5) .. (s2) .. controls +(down:0.5) and +(up:0.5) .. (s3)--($(s3)+(0,-0.25)$);

\node[label={left:{$p_1$}},scale=0] at (0,2) {};
\draw[timeline](0,2) -- (5.5,2);
\node[label={left:{$p_2$}},scale=0] at (0,1) {};
\draw[timeline](0,1) -- (5.5,1);
\node[label={left:{$p_3$}},scale=0] at (0,0) {};
\draw[timeline](0,0) -- (5.5,0);

\node[messageNode] (m11) at (0.5,2) {};
\draw[message](m11.center) to (1.1,1);
\node[messageNode] (m20) at (0.4,1) {};
\draw[message](m20.center) to (1,0);
\node[messageNode] (m21) at (1.1,1) {};
\draw[message](m21.center) to (1.8,0);
\node[messageNode] (m22p) at (1.9,1) {};
\draw[message](m22p.center) to (2.6,0);

\node[messageNode] (m13) at (2.5,2) {};
\draw[message](m13.center) to (3.2,1);
\node[messageNode] (m22) at (3.2,1) {};
\draw[message](m22.center) to (4,0);

\node[inlineImage] at (s1) {\snapshotImage{colorRedBase}};
\node[inlineImage] at (s2) {\snapshotImage{colorRedBase}};
\node[inlineImage] at (s3) {\snapshotImage{colorRedBase}};
\node[inlineImage] at (s12) {\snapshotImage{colorBlueBase}};
\node[inlineImage] at (s22) {\snapshotImage{colorBlueBase}};
\node[inlineImage] at (s32) {\snapshotImage{colorBlueBase}};
\end{tikzpicture}
\else
hideTikz is set
\fi
\caption{Snapshot-view of the execution}\label{fig:snapshot-view:snapshot}
\end{subfigure}%
\caption{Executions viewed through the latest common snapshot.}\label{fig:snapshot-view}
\end{figure}

\begin{restatable}[Latest Common Snapshot]{definition}{latestCommonSnapshot}\label{def:latest-common-snapshot}
The latest common snapshot of a configuration $c = \langle \Pi, \Sigma, N, M, M_0 \rangle$ is a configuration described by $\mathrm{lcs}(c)$:
\[
\mathrm{lcs}(c) = \langle~ \Pi,\, \Sigma',\, N',\, M_0\cup \mathrm{out}(c),\, M_0 ~\rangle, \text{where}\\
\]\[\begin{array}[t]{r@{~}l}
\Sigma' &= \mapbuild{p}{\langle\, \mathrm{A}(a),\, \langle\, \gce{c} + 1,\, a(\gce{c})\,\rangle\,\rangle}{\Sigma_p = \langle\, a,\, \sigma_\mathrm{V} \,\rangle} \\
\mathrm{A}(a) &= \mapbuild{e}{a(e)}{e\in\dom{a}\land e \leq \gce{c}}\\
N' &= \mapbuild{p}{\mapbuild{s}{|\mathrm{out}(c)\downarrow s|}{s \in \mathrm{dom}(N_p)}}{p \in \mathrm{dom}(N)}
\end{array}\]
\end{restatable}

Viewing computations through the lens of the latest common snapshot shows configurations which are caused by failure-free executions.
\Cref{fig:snapshot-view:original} shows an execution with a failed processor $p_2$ and an incompletely processed epoch (green).
In contrast, the latest common snapshot view of the same execution (\Cref{fig:snapshot-view:snapshot}) shows only the two completed epochs (red, blue), masking the failed epoch.
The snapshot is emulating an execution such that all the steps on epochs after the greatest common epoch are not taken, and all failed steps of incompletely processed epochs are ignored.
This reasoning is further elaborated for the proof of failure transparency in the next section, where we show that the implementation model is failure transparent when viewed through the lens of the output messages function.

\subsection{Assumptions}\label{sec:modeling-assumptions}
We make the following assumptions as a means to distill the essential mechanism of the failure recovery protocol.
We assume that the message channels are FIFO ordered, a common assumption for snapshotting protocols~\cite{chandy1985distributed}.
With regard to failures, we make common assumptions to asynchronous distributed systems~\cite{cachin2011introduction}.
Failures are assumed to be crash-recovery failures, in which a node looses its volatile state from crashing.
Further, we assume the existence of an eventually perfect failure detector, which is used for (eventually) triggering the recovery.
With regard to system components, we assume the following components which can be found in production dataflow systems.
The implicit coordinator instance is assumed to be failure free; in practice it is implemented using a distributed consensus protocol such as Paxos~\cite{DBLP:journals/tocs/Lamport98}.
The snapshot storage is assumed to be persistent and durable; a system such as HDFS~\cite{shvachko2010hadoop} would provide this.
Further, the input to the dataflow graph is assumed to be logged such that it can be replayed upon failure.
In practice, a durable log system such as Kafka~\cite{kreps2011kafka} would be used for this.
For our model, we make the following assumptions.
The recovery is assumed to be an atomic, synchronous system-wide step. In practice, it may be implemented as an asynchronous atomic step, which allows tasks to start processing before all have been recovered.
Further, the task's processing functions are assumed to be pure, \ie free from side effects.
A function $f$ may be re-executed multiple times due to failures; a common assumption in related work~\cite{burckhardt2021durable,kallas2023executing}.


\section{Failure Transparency}\label{sec:failure-transparency}

In this section, we define failure transparency such that it can be applied to systems described in small-step operational semantics with distinct failure-related rules.
We first provide a rationale behind failure transparency, followed by its formalization.

\subsection{Rationale}\label{sec:intuition}
The purpose of failure transparency is to provide an abstraction of a system which hides the internals of failures and failure recovery.
In particular, we would like to be able to show that the implementation model presented in the previous section is failure transparent.
In concrete terms, this entails showing that executions in the implementation model can be ``explained'' by failure-free executions, something which we explore in this section.

\begin{figure}
\centering
\newcommand*{\tikzscale}{0.79}
\ifx\hideTikz\undefined

\begin{tikzpicture}[thick, colorGrayBase, node distance=7mm, align=center, every node/.style={transform shape}]

\begin{scope}[scale=\tikzscale]

\newcommand*{\ftyshift}{-3.5mm}
\newcommand*{\fttextshift}{0.5mm}

\tikzstyle{config}    = [rectangle, draw, rounded corners=1mm, very thick,text=black,inner sep=1mm,fill=colorGrayBack, scale=0.875]
\tikzstyle{execstep}  = [-{Triangle[angle=45:4]}]
\tikzstyle{maps}      = [|->, shorten <=1pt, shorten >=1pt, color=colorBlueBase]
\tikzstyle{desc}      = [above, yshift=\fttextshift, font=\itshape\small]

\node[config]              (s0) {$av = 0$ \\ $a = a_0$};
\node[config, right=of s0] (s1) {$av = 1$ \\ $a = a_0$};
\node[config, right=of s1, xshift=-1mm] (s2) {$av = 1$ \\ $a = a_1$};
\node[config, right=of s2] (s3) {$av = 2$ \\ $a = a_1$};
\node[config, right=of s3, xshift=-1.25mm] (s4) {$av = \texttt{fl}$ \\ $a = a_1$};
\node[config, right=of s4, xshift=1.5mm] (s5) {$av = 1$ \\ $a = a_1$};
\node[config, right=of s5, xshift=3mm] (s6) {$av = 0$ \\ $a = a_1$};
\node[config, right=of s6] (s7) {$av = 3$ \\ $a = a_1$};
\node[config, right=of s7] (s8) {$av = 4$ \\ $a = a_1$};
\node[config, right=of s8, xshift=-1mm] (s9) {$av = 4$ \\ $a = a_2$};

\node[config, below=of s0, yshift=\ftyshift] (s0b) {$av = 0$ \\ $a = a_0$};
\node[config, below=of s1, yshift=\ftyshift] (s1b) {$av = 1$ \\ $a = a_0$};
\node[config, below=of s2, yshift=\ftyshift] (s2b) {$av = 1$ \\ $a = a_1$};
\node[config, below=of s6, yshift=\ftyshift] (s5b) {$av = 0$ \\ $a = a_1$};
\node[config, below=of s7, yshift=\ftyshift] (s7b) {$av = 3$ \\ $a = a_1$};
\node[config, below=of s8, yshift=\ftyshift] (s8b) {$av = 4$ \\ $a = a_1$};
\node[config, below=of s9, yshift=\ftyshift] (s9b) {$av = 4$ \\ $a = a_2$};


\draw[execstep] (s0) -- node[desc] {receive \\ $\smash{\integerEvent{1}}$} (s1);
\draw[execstep,colorBlueBase] (s1) -- node[desc,colorBlueBase] {process \\ $\smash{\mathtt{BD}}$s} (s2);
\draw[execstep] (s2) -- node[desc] {receive \\ $\smash{\integerEvent{3}}$} (s3);
\draw[execstep,colorRedBase] (s3) -- node[desc,colorRedBase] {fail} (s4);
\draw[execstep,colorGreenBase] (s4) -- node[desc,colorGreenBase] {recover to \\ $\smash{a_1(1)}$} (s5);
\draw[execstep] (s5) -- node[desc] {receive \\ $\smash{\resetEvent}$} (s6);
\draw[execstep] (s6) -- node[desc] {receive \\ $\smash{\integerEvent{3}}$} (s7);
\draw[execstep] (s7) -- node[desc] {receive \\ $\smash{\integerEvent{5}}$} (s8);
\draw[execstep,colorBlueBase] (s8) -- node[desc,colorBlueBase] {process \\ $\smash{\mathtt{BD}}$s} (s9);

\draw[execstep] (s0b) -- node[desc] {receive \\ $\smash{\integerEvent{1}}$} (s1b);
\draw[execstep,colorBlueBase] (s1b) -- node[desc,colorBlueBase] {process \\ $\smash{\mathtt{BD}}$s} (s2b);
\draw[execstep] (s2b) -- node[desc] {receive \\ $\smash{\resetEvent}$} (s5b);
\draw[execstep] (s5b) -- node[desc] {receive \\ $\smash{\integerEvent{3}}$} (s7b);
\draw[execstep] (s7b) -- node[desc] {receive \\ $\smash{\integerEvent{5}}$} (s8b);
\draw[execstep,colorBlueBase] (s8b) -- node[desc,colorBlueBase] {process \\ $\smash{\mathtt{BD}}$s} (s9b);

\draw[maps] (s0) -- (s0b);
\draw[maps] (s2) -- (s2b);
\draw[maps] (s9) -- (s9b);

\end{scope}
\end{tikzpicture}

\else
hideTikz is set
\fi
\caption{Execution of the incremental average task (\Cref{fig:incremental-average}). Top: execution with a failure and subsequent recovery.
Bottom: corresponding failure-free execution.
Snapshot archives: $a_0 = \map{0}{0}$, $a_1 = a_0 \map{1}{1}$, $a_2 = a_1 \map{2}{4}$.
}
\label{fig:ft-rationale}
\end{figure}
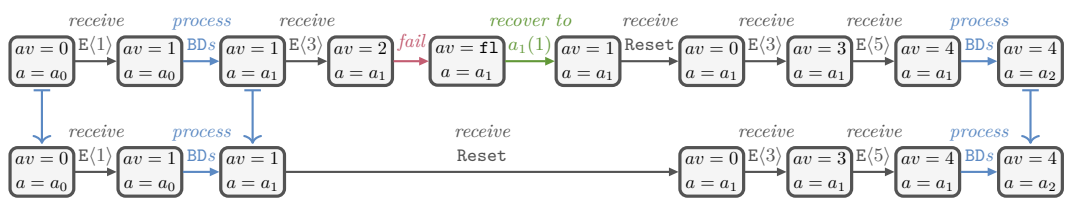

Consider the task of computing the incremental average from the previous example (\Cref{sec:stateful-dataflow}, \Cref{fig:incremental-average}).
The task consumes regular events \integerEvent{i}, reset events, and border events \texttt{BD}.
For this example, we consider a partial execution of the task in which it processes the events:
$[\integerEvent{1}, \texttt{BD}, \integerEvent{3}, \text{fail}, \text{recover}, \resetEvent, \integerEvent{3}, \integerEvent{5}, \texttt{BD}, \ldots]$.
The task's configurations consist of the task's current average value $av$, and its snapshot archive, $a$.
\Cref{fig:ft-rationale} shows at the top an execution of the task with a failure and subsequent failure recovery as the fourth and fifth events. After the recovery step, in its sixth configuration, the task's state is reset to its state for the snapshot $a_1(1)$, at which point it had the average value 1.

The question we ask is whether we can rely on the behavior of the task?
More specifically, can we use the average value $av = 2$ in the fourth configuration (after receiving the event \integerEvent{3})?
The problem is that the task will fail in its next step, and recover to a state in which the receiving of the event has been undone.
Moreover, the task continues its execution after recovery by processing the reset event first, and does never reach a state again in which its average value is $2$.
For this reason, we cannot blindly rely on the observed behaviors of the task as we may observe things which are later undone.
In more complex systems, failures may further result in duplications and reorderings of events, further complicating the reasoning about the system.

Dealing with these issues requires the observer of the system to reason about which events are effectful and which are to be discarded.
In some sense, the observer should be able to reason about the observed execution as if it was an ideal, failure-free execution, \ie an execution in which all events are effectful.
Put in another way, the solution is to find a corresponding failure-free execution, and reason about that one instead.
Intuitively, the observer should find some failure-free execution which ``explains'' the execution.
Considering the above example, a failure-free execution thereof would correspond to the bottom execution in \Cref{fig:ft-rationale}.
Note that there are no failure or recovery steps in the failure-free execution, yet its state progresses in a similar way to the original execution.

Even though the failure-free execution on an intuitive level correspond to the original execution, we would like to have a formal notion for this.
The idea is to lift the observed executions by means of ``observability functions'', to a level where failure-related events and states are hidden.
For example, for the executions above, we could define an observability function which takes the configuration of the task and keeps only the snapshot storage.
After this transformation, applying this function to every configuration in the executions, we will not be able to distinguish the two executions by observing the system at any point in time.
That is, common to both executions, we will first observe $a_0$, then $a_1$, and finally $a_2$.
On a technical level, for every configuration of the original execution, we can find a configuration in the failure-free execution which, after application of the observability functions, is equal to it (\eg the mapping from top to bottom configurations in \Cref{fig:ft-rationale}); this is what we mean by ``observable explainability''.
Thus, we can explain the original execution by the failure-free execution using the provided observability function.

The essence of our definition of failure transparency is derived from the notion of explaining the original executions by failure-free executions using observability functions.
Instead of reasoning about executions, we can reason about the observable output of executions at any given moment.
Using observability functions effectively hides the internals of the model and enables the user to focus on the output of the system.
That is, the user can reason about failure-free executions instead of faulty executions.

This informal introduction highlights three essential parts of failure transparency: the execution system, failures within the system, and the observability of the system.
The goal of the rest of this section is to define these terms and to provide a formal definition of \emph{failure transparency}.

\subsection{Executions}\label{sec:executions}
The execution system for the failure transparency analysis is modelled as a transition system for which the transition relation is provided as a set of inference rules.
In particular, we provide a formal definition for executions as a means to discuss the execution of systems.
With this notion, distributed programs can be formally modelled in small-step operational semantics, and consequently formally verified.
Although it may seem unintuitive to model distributed systems as transition systems for which the transition relation is defined over the global state, this is in fact commonly done in other formal frameworks such as TLA\textsuperscript{+}~\cite{lamport2002specifying}.

\begin{definition}[Execution Step]
A statement $c \Rightarrow c'$ is called an execution step from $c$ to $c'$.
We denote the derivability of an execution step in the set of rules $R$ by $R \vdash c \Rightarrow c'$.
\end{definition}

We reason about systems in terms of their executions.
An execution is a sequence of configurations $C$, connected by execution steps derivable in a set of rules $R$, starting from some initial configuration $C_0$.

\begin{definition}[Executions]
A sequence of configurations $\seq{C_i}{i}{n}$ is called an \emph{execution} in a set of rules $R$, if $\forall i < n .~  R \vdash C_{i-1} \Rightarrow C_{i}$.
The set of all possible executions starting from $C_0$ in $R$ is denoted as $\mathbb{E}_{C_0}^R$.
\end{definition}

The set of rules $R$ of an execution specifies its reducibility relation by providing $c \Rightarrow c'$ as a conclusion of some of its rules.
This approach is commonly known as \emph{small-step operational semantics}.
In our representation, the set of rules is explicit, whereas commonly it is implicit.
This is due to our need to explicitly distinguish between separate execution systems.
This allows us, for example, to separate an execution system into two parts: one with failures $R$ \suchthat the failure-related rules are a subset thereof $F \subseteq R$, and one without failures $(R \setminus F)$.

\subsection{Observational Explainability}\label{sec:observational-explainability}
The observability function represents the observer's view of the system.
It notably differs from the plain configurations in the following two ways: the observer may not observe all internal details of configurations, \ie some parts of the configuration are \emph{hidden} from the observer (\eg hiding commit messages~\cite{burckhardt2021durable}); and the observer may observe some derived views of the configuration.

\begin{definition}[Observability Function]
An \emph{observability function} $O$ of an execution system is a function which maps configurations to their observable outputs.
It is required to be monotonic with respect to execution steps possible in the set of rules $R$ for some partial order $\sqsubseteq_O$, that is: $\forall c, c'.~ (R \vdash c \Rightarrow c') \implies O(c) \sqsubseteq_O O(c')$.
\end{definition}

We say that an implementation's execution is observably explained by a specification's execution, if the observer cannot distinguish the two executions.
This is the case when, for every configuration in the implementation's execution, there is a corresponding configuration in the specification's execution, such that their observed values are equal after application of the respective observability functions.

\begin{definition}[Observational Explanation]\label{def:observational-explaination}
A sequence of configurations $C$ of length $n$ is explained by a sequence of configurations $C'$ of length $n'$ with respect to observability functions $O$ and $O'$, denoted as $C ~\,^{O}\negmedspace\rightleftharpoons\!^{O'}\, C'$, if:
\[
\forall m < n .~ \exists m' < n' .~ O(C_{m}) = O'(C'_{m'})
\]
\end{definition}

An implementation's system, in turn, is observably explainable by the specification's system, if for each execution of the implementation there exists an explaining execution in the specification.
We call this property \emph{observational explainability}.

\begin{definition}[Observational Explainability]\label{def:observational-explainability}
The set of rules $R$ is \emph{observationally explainable} by $R'$ with respect to their observability functions $O$ and $O'$ and the translation relation $T$, denoted as $R ~\,^{O}\negmedspace\xrightleftharpoons{T}\!^{O'}\, R'$, if:
\[
\forall\, c' \in \dom{T}.~ \forall c.~ c' T c \implies
\forall C \in \mathbb{E}_{c}^{R}.~
\exists\, C' \in \mathbb{E}_{c'}^{R'}.~
C ~\,^{O}\negmedspace\rightleftharpoons\!^{O'}\, C'
\]
\end{definition}

\subparagraph*{Properties of Observational Explainability.}
Observability functions are required to be monotonic, since observations should be regarded as stable.
That is, once a value has been observed, then it should remain observable in the future.
The system should not be able to undo something that has been observed, otherwise the observer would not be able to rely on the output.
The reason for this is twofold.
First, an observer may observe the system multiple times, and newer observations should provide more up-to-date views.
Second, the sequence of observations should correspond to a valid explanation with respect to the higher-level specification, this is explored next.

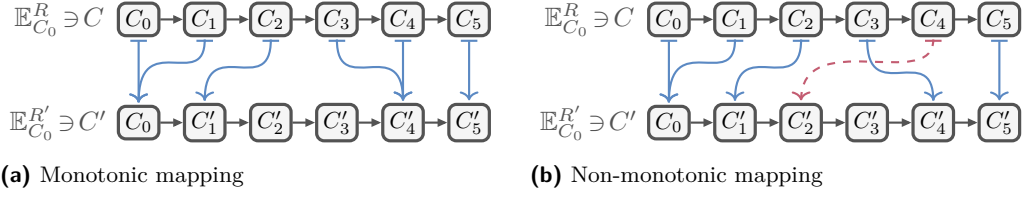
\begin{figure}[t]\centering
\newcommand*{\tikzscale}{1.0}
\begin{subfigure}[t]{0.49\textwidth}
  \centering
  \ifx\hideTikz\undefined

\begin{tikzpicture}[thick, colorGrayBase, node distance=3mm, align=center, every node/.style={transform shape}]

\begin{scope}[scale=\tikzscale]

\newcommand*{\ftyshift}{-3.5mm}
\newcommand*{\fttextshift}{0.5mm}

\tikzstyle{config}    = [rectangle, draw, rounded corners=1mm, very thick,text=black,inner sep=1mm,fill=colorGrayBack, scale=0.875]
\tikzstyle{execstep}  = [-{Triangle[angle=45:4]}]
\tikzstyle{maps}      = [|->, shorten <=1pt, shorten >=1pt, color=colorBlueBase]
\tikzstyle{desc}      = [above, yshift=\fttextshift, font=\itshape\small]

\node (ecr) {$\mathbb{E}^{R}_{C_0}\!\ni\! C$};
\node[config, right=of ecr, xshift=-3mm] (c0) {$C_0$};
\node[config, right=of c0] (c1) {$C_1$};
\node[config, right=of c1] (c2) {$C_2$};
\node[config, right=of c2] (c3) {$C_3$};
\node[config, right=of c3] (c4) {$C_4$};
\node[config, right=of c4] (c5) {$C_5$};

\node[config, below=0.9cm of c0] (cp0) {$C_0$};
\node[config, right=of cp0] (cp1) {$C'_1$};
\node[config, right=of cp1] (cp2) {$C'_2$};
\node[config, right=of cp2] (cp3) {$C'_3$};
\node[config, right=of cp3] (cp4) {$C'_4$};
\node[config, right=of cp4] (cp5) {$C'_5$};
\node (ecrp) [left= of cp0,xshift=+3mm] {$\mathbb{E}^{R'}_{C_0}\!\ni\! C'$};

\draw[execstep] (c0) -- (c1);
\draw[execstep] (c1) -- (c2);
\draw[execstep] (c2) -- (c3);
\draw[execstep] (c3) -- (c4);
\draw[execstep] (c4) -- (c5);

\draw[execstep] (cp0) -- (cp1);
\draw[execstep] (cp1) -- (cp2);
\draw[execstep] (cp2) -- (cp3);
\draw[execstep] (cp3) -- (cp4);
\draw[execstep] (cp4) -- (cp5);

\draw[maps] (c0) -- (cp0);
\draw[maps] (c4) -- (cp4);
\draw[maps] (c5) -- (cp5);
\draw[maps] (c1.south) -- ($(c1.south)+(0,-1pt)$) .. controls +(down:0.7) and +(up:1.0) .. ($(cp0.north)+(0,+1pt)$) -- (cp0.north);
\draw[maps] (c2.south) -- ($(c2.south)+(0,-1pt)$) .. controls +(down:0.8) and +(up:0.9) .. ($(cp1.north)+(0,+1pt)$) -- (cp1.north);
\draw[maps] (c3.south) -- ($(c3.south)+(0,-1pt)$) .. controls +(down:0.8) and +(up:0.9) .. ($(cp4.north)+(0,+1pt)$) -- (cp4.north);

\end{scope}
\end{tikzpicture}

\else
hideTikz is set
\fi
  \caption{Monotonic mapping}\label{fig:history:monotonic}
\end{subfigure}
\begin{subfigure}[t]{0.49\textwidth}
  \centering
  \ifx\hideTikz\undefined

\begin{tikzpicture}[thick, colorGrayBase, node distance=3mm, align=center, every node/.style={transform shape}]

\begin{scope}[scale=\tikzscale]

\newcommand*{\ftyshift}{-3.5mm}
\newcommand*{\fttextshift}{0.5mm}

\tikzstyle{config}    = [rectangle, draw, rounded corners=1mm, very thick,text=black,inner sep=1mm,fill=colorGrayBack, scale=0.875]
\tikzstyle{execstep}  = [-{Triangle[angle=45:4]}]
\tikzstyle{maps}      = [|->, shorten <=1pt, shorten >=1pt, color=colorBlueBase]
\tikzstyle{nonmaps}   = [|->, shorten <=1pt, shorten >=1pt, color=colorRedBase, dashed]
\tikzstyle{desc}      = [above, yshift=\fttextshift, font=\itshape\small]

\node (ecr) {$\mathbb{E}^{R}_{C_0}\!\ni\! C$};
\node[config, right=of ecr, xshift=-3mm] (c0) {$C_0$};
\node[config, right=of c0] (c1) {$C_1$};
\node[config, right=of c1] (c2) {$C_2$};
\node[config, right=of c2] (c3) {$C_3$};
\node[config, right=of c3] (c4) {$C_4$};
\node[config, right=of c4] (c5) {$C_5$};

\node[config, below=0.9cm of c0] (cp0) {$C_0$};
\node[config, right=of cp0] (cp1) {$C'_1$};
\node[config, right=of cp1] (cp2) {$C'_2$};
\node[config, right=of cp2] (cp3) {$C'_3$};
\node[config, right=of cp3] (cp4) {$C'_4$};
\node[config, right=of cp4] (cp5) {$C'_5$};
\node (ecrp) [left= of cp0,xshift=+3mm] {$\mathbb{E}^{R'}_{C_0}\!\ni\! C'$};

\draw[execstep] (c0) -- (c1);
\draw[execstep] (c1) -- (c2);
\draw[execstep] (c2) -- (c3);
\draw[execstep] (c3) -- (c4);
\draw[execstep] (c4) -- (c5);

\draw[execstep] (cp0) -- (cp1);
\draw[execstep] (cp1) -- (cp2);
\draw[execstep] (cp2) -- (cp3);
\draw[execstep] (cp3) -- (cp4);
\draw[execstep] (cp4) -- (cp5);

\draw[maps] (c0) -- (cp0);
\draw[maps] (c5) -- (cp5);
\draw[maps] (c1.south) -- ($(c1.south)+(0,-1pt)$) .. controls +(down:0.7) and +(up:1.0) .. ($(cp0.north)+(0,+1pt)$) -- (cp0.north);
\draw[maps] (c2.south) -- ($(c2.south)+(0,-1pt)$) .. controls +(down:0.8) and +(up:0.9) .. ($(cp1.north)+(0,+1pt)$) -- (cp1.north);
\draw[maps] (c3.south) -- ($(c3.south)+(0,-1pt)$) .. controls +(down:0.9) and +(up:0.8) .. ($(cp4.north)+(0,+1pt)$) -- (cp4.north);
\draw[nonmaps] (c4.south) -- ($(c4.south)+(0,-1pt)$) .. controls +(down:0.7) and +(up:1.0) .. ($(cp2.north)+(0,+1pt)$) -- (cp2.north);

\end{scope}
\end{tikzpicture}

\else
hideTikz is set
\fi
  \caption{Non-monotonic mapping}\label{fig:history:nonmonotonic}
\end{subfigure}
\caption{Monotonic and non-monotonic mapping of configurations.}\label{fig:history}
\end{figure}

In the general case, it is desirable to have a monotonic mapping of configurations between the abstract-level and implementation-level executions.
\Cref{fig:history:monotonic} shows a monotonic mapping of configurations between an implementation (top) and a specification (bottom).
What makes the mapping monotonic is that each subsequently mapped configuration of the implementation is mapped to a configuration with a monotonically growing index.
\Cref{fig:history:nonmonotonic}, on the other hand, shows a non-monotonic mapping, as indicated by the red dashed line.
Non-monotonic mappings, however, are not considered valid explanations.
For example, if the specification consists of the sequence $a$ followed by $b$, then an implementation which produces $b$ followed by $a$ is not considered a valid implementation thereof.
Thus, we should not use non-monotonic mappings for the explainability of executions.
We capture this notion in the definition of monotonic observational explanation.

\begin{definition}[Monotonic Observational Explanation]\label{def:monotonic-observational-explaination}
An observational explanation is \emph{monotonic} if it is a monotonic mapping of configurations.
That is, $[C_i]^n_i$ is \emph{monotonically} explained by $[C'_j]^{n'}_j$ w.r.t.\ $O$ and $O'$ if:
\[
  \exists \seq{h_k}{k}{n}.~ (\forall k < n.~ \forall k' \leq k.~ h_{k'} \leq h_k) ~ \land ~
  (\forall m < n.~ \exists m'  = h_{m} < n'.~ O(C_{m}) = O'(C'_{m'}))
\]
\end{definition}

The following lemma explicitly shows that our definition of \emph{observational explainability} is equivalent to the definition of \emph{monotonic observational explainability}.
That is, our definition does not have the problem with non-monotonic mappings of configurations since the observability functions are required to be monotonic.
For this reason, we do not distinguish between the two definitions in the following sections.

\begin{restatable}{lemma}{monotonicity}\label{lem:monotonicity}
If $R$ is observationally explainable by $R'$ w.r.t.\ $O$, $O'$, $T$, then it is also \emph{monotonically} observationally explainable:
\begin{align*}
&\forall\, c' \in \dom{T}.~ \forall c.~ c' T c \implies
\forall C \in \mathbb{E}_{c}^{R}.~
\exists\, C' \in \mathbb{E}_{c'}^{R'}.~ \\
&\quad C \text{ is monotonically explained by } C' \text{ w.r.t. } O \text{ and } O'
\end{align*}
\end{restatable}
\begin{proof}
\appRef{prf:monotonicity}.
\end{proof}

To further aid the use of these definitions within proofs, we also show that the definition of observational explainability is transitive, as well as a compositionality lemma on the observability functions.
The parametrization of the observable explainability enables reasoning about models which differ in their initial states, and for which we want to apply different observability functions at the different levels.
That is, it can be used for reasoning about sets of rules which differ in their initial states, and for which we want to apply different observability functions at the different levels.

\begin{restatable}[Transitivity]{lemma}{transitivity}\label{lem:transitivity}
$
\rotor{R}{O}{T}{O'}{R'}
\land
\rotor{R'}{O'}{T'}{O''}{R''}
~{~}{~}\implies~{~}{~}
\rotor{R}{O}{T \circ T'}{O''}{R''}
$
\end{restatable}
\begin{proof}
\appRef{prf:transitivity}.
\end{proof}

\begin{restatable}[Composition]{lemma}{composition}\label{lem:composition}
$
  \forall O''.~
  \rotor{R}{O}{T}{O'}{R'}
  {~}{~}\implies{~}{~}
  \rotor{R}{O'' \circ O}{T}{O'' \circ O'}{R'}
$
\end{restatable}
\begin{proof}
\appRef{prf:composition}.
\end{proof}

\subsection{Defining Failure Transparency}\label{sec:defining-failure-transparency}
The general goal of failure transparency is to provide an abstraction of a system which masks failures from the users.
We express this notion using observational explainability between the implementation and its failure-free part.
That is, the implementation should be observationally explainable by the implementation without failures.
By explicitly separating the set of failure-related rules $F$, it is easy to define the two systems: namely, the implementation system with all rules, \ie $R$; and another system with all rules except the failure-related rules, \ie $R \setminus F$.
To fully instantiate the observational equivalence, we further use the same observability function $O$ on both the low and high levels, and as a translation relation we use the identity relation on the set of initial configurations.

\begin{definition}[Failure Transparency]\label{def:failure-transparency}
A set of rules $R$ is \emph{failure-transparent} with respect to failure rules $F \subseteq R$ for a monotonic observability function $O$ and a set of initial configurations $K$, this is denoted as $R \bbslash^O_K F$, iff:
\[{R} ~\,^{O}\negmedspace\xrightleftharpoons{\ssetbuild{(c,\, c)}{c\in K}}\!^{O}\, {(R\setminus F)}\]
\end{definition}

\section{Failure Transparency of Stateful Dataflow}\label{sec:failure-transparency-of-the-implementation}
In this section, we show that the presented implementation model (\Cref{sec:implementation-model}) is failure transparent (\Cref{def:failure-transparency}) for the observability function $\mathrm{out}$ (\Cref{def:output-messages}).
In order to prove this, instead of reasoning about executions directly, we reason about the traces of steps which are performed to obtain these executions.
This simplifies the proof, enabling us to reorder and remove specific steps in and from a trace; in contrast, doing the same with a configuration from an execution affects all following configurations.
In this section, we first define traces and a causal order relation on traces, and then prove the failure transparency of the implementation model by manipulating traces.
Finally, we complete our analysis of the model by formulating and proving its liveness, showing that the implementation model eventually produces outputs for all epochs in its input.

\subsection{Traces and Causality}\label{sec:steps-traces-and-causality}

A trace is a sequence of steps, for which each step is a compact representation of the derivation of a transition from one configuration to another.

\begin{restatable}[Trace]{definition}{traceDefinition}\label{def:trace}
A trace $Z$ is a sequence of trace steps.
A trace step $z$ is one of:
$\langle \mathsc{I-Event}, p, N_p, X \rangle$;
$\langle \mathsc{I-Border}, p, N_p, X \rangle$;
$\langle \mathsc{F-Fail}, p \rangle$;
$\langle \mathsc{F-Recover}\rangle$.
Here $\mathsc{I-Event}$, $\mathsc{I-Border}$, $\mathsc{F-Fail}$, and $\mathsc{F-Recover}$ play the role of the discriminant, where the trace step is a tagged union.
\end{restatable}

For example, if in the derivation tree of an execution step from the $i$th to the $i + 1$th configuration, \ie of $R \vdash C_{i} \Rightarrow C_{i+1}$, $\mathsc{F-Recover}$ was the root rule, then this execution step corresponds to the step $\langle \mathsc{F-Recover} \rangle$ in the trace.
To link traces with executions, we use the following definition of trace application.

\begin{definition}[Trace Application]\label{def:trace-from-execution}
A trace $Z$ of length $n$ applied to a configuration $c$ results in a sequence of configurations $C$ of length $n+1$, \ie $Z(c) = C$, if, for all steps $Z_i$, the represented derivation of an execution step can be applied to the $i$th configuration producing the $i+1$th configuration.
\end{definition}

Traces can be generated from executions; however, not every trace corresponds to an execution.
This may be the case if a trace has been constructed incorrectly, or reordered in some way.
For this reason, we define valid traces, which are traces that correspond to executions.

\begin{restatable}[Valid Trace]{definition}{validTrace}\label{def:valid-trace}
A trace $Z$ is valid from configuration $c$ if it is applicable to it, \ie if there exists an execution $C\in\executions{c}{I}$ such that $Z(c) = C$.
\end{restatable}

As the proof reasons about the reordering of steps in a trace, it is important to formulate which reorderings of steps preserve the validity of the trace.
To handle this, we define a causal order relation on trace steps similar to the happens-before relation~\cite{DBLP:journals/cacm/Lamport78}, and show how it can be used to reason about traces.

\begin{definition}[Causal Order]\label{def:causal-order}
\techReport{(See \Cref{def:causal-order-full} for the formal definition)}
\mainReport{(See technical report~\cite{veresov2024technicalreport} for the formal definition)}
A step $Z_i$ happens before $Z_j$ with $i < j$ if:
\begin{enumerate}
  \item One of them is an $\mathsc{F-Recover}$ step (global recovery)
  \item They both occur on the same processor (intraprocessor order)
  \item If $Z_i$ produced a message which is consumed by $Z_j$ (interprocessor order)
  \item If there exists some step $Z_k$ such that $Z_i$ happens before $Z_k$ and $Z_k$ happens before $Z_j$ (transitivity)
\end{enumerate}
\end{definition}

Finally, we state a lemma that causality-preserving permutations, \ie permutations that preserve the causal order relation~\appCite{def:permutationpreserving}{Definition B.5}, also preserve the validity and the end result of their application.
Intuitively, it follows from the fact that causally unrelated steps should not influence each other.

\begin{restatable}[Application of Causality-Preserving Permutations]{lemma}{validityCausal}\label{lem:validity-similarity-causal}
For a trace $Z$ valid from $c$ with size $\size{Z}=n$, if $Z'$ is a causality-preserving permutation of $Z$, then: $Z'$ is valid from $c$; $Z$ and $Z'$ end in the same configuration after application to $c$, \ie $Z(c)_n = Z'(c)_n$.
\end{restatable}
\begin{proof}
\appRef{prf:validity-similarity-causal}.
\end{proof}

\subsection{Proving Failure Transparency}
As it is required by the definition of failure transparency, we first define the sets of rules, namely $\rulesI$, $\rulesF$, and $(\rulesI \setminus \rulesF)$; and the set of valid initial configurations $K$.

The semantics of the model consist of seven rules, defining two separate sets of rules.
The set of rules with failures $\rulesI$ consists of all seven rules that have been defined for the stateful dataflow implementation model; it corresponds to the implementation model presented in \Cref{sec:implementation-model}.
The set of failure-related rules $\rulesF$ within the implementation model consists of the two rules $\mathsc{F-Fail}$ and $\mathsc{F-Recover}$.
This way, the rules without failures are defined as the set $(\rulesI \setminus \rulesF)$.

\begin{definition}[Implementation Model Rules]\label{def:low-level}
$\rulesI = \{ \mathsc{S-Step}, \mathsc{S-AbsX}, \mathsc{S-AbsP}, \mathsc{I-Event}, \allowbreak \mathsc{I-Border} \} \cup \mathrm{F}$
\end{definition}

\begin{definition}[Failure-Related Rules]\label{def:failures}
$\rulesF = \{\mathsc{F-Fail}, \mathsc{F-Recover}\}$
\end{definition}

The sets of initial configurations which are considered are any acyclic graph structures which are properly initialized.

\begin{definition}[Valid Initial Configurations]\label{def:valid-initial-configurations}
$K = \configuration{\Pi}{\Sigma}{N}{M}{M_0}$
such that:
the graph defined by $\Pi$ is acyclic, and the tasks' functions $f$ do not output infinite sequences;
$\Sigma$ are the initial well-formed states;
$N$ are sequence numbers initialized to 0 for the streams;
$M$ consists of the well-formed inputs to the streams;
$M_0 = M$.
\end{definition}

\begin{restatable}[Failure Transparency of the Implementation Model]{theorem}{failureTransparency}\label{thm:failure-transparency}
$I \bbslash^{\mathrm{out}}_{K} F$, \ie the set of rules $\mathrm{I} = \{\mathsc{S-Step}, \mathsc{S-AbsX}, \mathsc{S-AbsP}, \mathsc{I-Event}, \mathsc{I-Border}\}\cup \mathrm{F}$ is failure transparent with respect to the failure rules $\mathrm{F} = \{\mathsc{F-Fail}, \mathsc{F-Recover}\}$ for the observability function $\mathrm{out}$ and the set of initial configurations $K$.
\end{restatable}

Before proceeding with the proof itself, we provide a sketch of it.
The proof idea is to construct a failure-free observational explanation of an arbitrary execution in the implementation model.

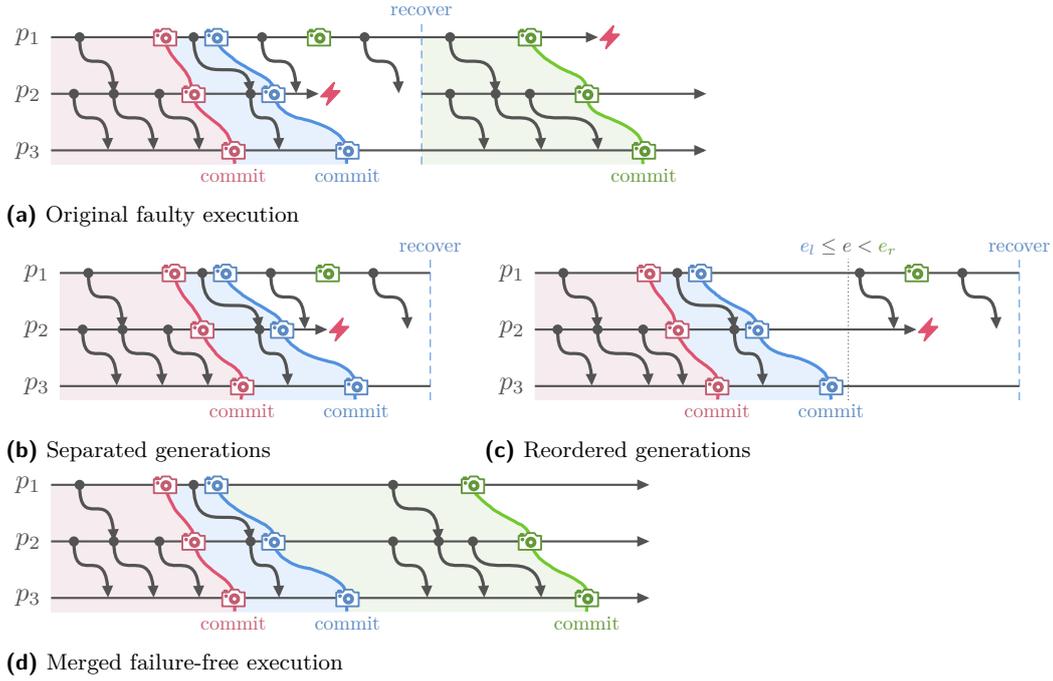
\begin{figure}[t]\centering
\begin{subfigure}{1.0\textwidth}
\centering
\begin{minipage}{1.0\textwidth}
\ifx\hideTikz\undefined
\begin{tikzpicture}[
  scale=0.75, every node/.style={scale=1},
  thick, colorGrayBase,
  timeline/.style={-{Triangle[angle=45:5]}, thick},
  message/.style={-{Triangle[angle=45:5]}, very thick, controls={+(down:0.9) and +(up:1.1)}},
  messageNode/.style={scale=0.4, circle, fill, prefix after command={\pgfextra{\tikzset{every label/.style={label distance=-1.5mm, scale=0.75}}}}},
  inlineImage/.style={scale=0.35, yscale=-1, xshift=-12, yshift=-12},
  snapshotEdge/.style={very thick, handdrawn}
]
\coordinate (s12) at (2.9,2);
\coordinate (s22) at (3.9,1);
\coordinate (s32) at (5.2,0);

\begin{scope}
\clip (0,-0.25)--(0,2)--(s12) .. controls +(down:0.5) and +(up:0.5) .. (s22) .. controls +(down:0.5) and +(up:0.5) .. (s32)--($(s32)+(0,-0.25)$)--(0,-0.25);
\fill[colorBlueBack] (0,-0.25)--(0,2.1)--(8.2,2.1)--(8.2,-0.25)--(0,-0.25);
\end{scope}

\draw[colorBlueEmph, snapshotEdge](s12) .. controls +(down:0.5) and +(up:0.5) .. (s22) .. controls +(down:0.5) and +(up:0.5) .. (s32)--($(s32)+(0,-0.25)$);

\coordinate (s1) at (2,2);
\coordinate (s2) at (2.5,1);
\coordinate (s3) at (3.2,0);

\begin{scope}
\clip (0,-0.25)--(0,2)--(s1) .. controls +(down:0.5) and +(up:0.5) .. (s2) .. controls +(down:0.5) and +(up:0.5) .. (s3)--($(s3)+(0,-0.25)$)--(0,-0.25);
\fill[colorRedBack] (0,-0.25)--(0,2.1)--(5,2.1)--(5,-0.25)--(0,-0.25);
\end{scope}

\draw[colorRedEmph, snapshotEdge](s1) .. controls +(down:0.5) and +(up:0.5) .. (s2) .. controls +(down:0.5) and +(up:0.5) .. (s3)--($(s3)+(0,-0.25)$);

\coordinate (s13) at (8-0.6+1,2);
\coordinate (s23) at (9-0.6+1,1);
\coordinate (s33) at (10-0.6+1,0);

\begin{scope}
\clip (0,-0.25)--(0,2)--(s13) .. controls +(down:0.5) and +(up:0.5) .. (s23) .. controls +(down:0.5) and +(up:0.5) .. (s33)--($(s33)+(0,-0.25)$)--(0,-0.25);
\fill[colorGreenBack] (6.5,-0.25)--(6.5,2.1)--(10.5,2.1)--(10.5,-0.25)--(6.5,-0.25);
\end{scope}

\draw[colorGreenEmph, snapshotEdge](s13) .. controls +(down:0.5) and +(up:0.5) .. (s23) .. controls +(down:0.5) and +(up:0.5) .. (s33)--($(s33)+(0,-0.25)$);

\node[label={left:{$p_1$}},scale=0] at (0,2) {};
\draw[timeline](0,2) -- (8.6+1,2);
\node[label={left:{$p_2$}},scale=0] at (0,1) {};
\draw[timeline](0,1) -- (4.7,1);
\draw[densely dashed, thin, colorBlueEmph] (6.5,-0.25) -- (6.5,2.25);
\node[label={[scale=0.75, colorBlueBase]above:{recover}},scale=0] at (6.5,2.25) {};
\draw[timeline](6.5,1) -- (10.5+1,1);
\node[label={left:{$p_3$}},scale=0] at (0,0) {};
\draw[timeline](0,0) -- (10.5+1,0);

\node[label={[scale=0.75, colorRedBase]below:{commit}},scale=0] at ($(s3.center)+(0,-0.2)$) {};
\node[label={[scale=0.75, colorBlueBase]below:{commit}},scale=0] at ($(s32.center)+(0,-0.2)$) {};
\node[label={[scale=0.75, colorGreenBase]below:{commit}},scale=0] at ($(s33.center)+(0,-0.2)$) {};

\node[messageNode] (m11) at (0.5,2) {};
\draw[message](m11.center) to (1.1,1);
\node[messageNode] (m20) at (0.4,1) {};
\draw[message](m20.center) to (1,0);
\node[messageNode] (m21) at (1.1,1) {};
\draw[message](m21.center) to (1.8,0);
\node[messageNode] (m22p) at (1.9,1) {};
\draw[message](m22p.center) to (2.6,0);

\node[messageNode] (m13) at (2.5,2) {};
\draw[message,controls={+(down:1.1) and +(up:0.9)}](m13.center) to (3.5,1);
\node[messageNode] (m22) at (3.5,1) {};
\draw[message](m22.center) to (4,0);

\node[messageNode] (m3) at (3.7,2) {};
\draw[message](m3.center) to (4.3,1);
\node[messageNode] (m4) at (5.5,2) {};
\draw[message](m4.center) to (6.1,1);

\node[messageNode] (m00) at (6+1,1) {};
\draw[message](m00.center) to (6.6+1,0);
\node[messageNode] (m01) at (6+1,2) {};
\draw[message](m01.center) to (6.8+1,1);
\node[messageNode] (m02) at (6.8+1,1) {};
\draw[message](m02.center) to (7.6+1,0);
\node[messageNode] (m03) at (7.4+1,1) {};
\draw[message](m03.center) to (8.6+1,0);

\node[inlineImage] at (s1) {\snapshotImage{colorRedBase}};
\node[inlineImage] at (s2) {\snapshotImage{colorRedBase}};
\node[inlineImage] at (s3) {\snapshotImage{colorRedBase}};
\node[inlineImage] at (s12) {\snapshotImage{colorBlueBase}};
\node[inlineImage] at (s22) {\snapshotImage{colorBlueBase}};
\node[inlineImage] at (s32) {\snapshotImage{colorBlueBase}};
\node[inlineImage] at (s13) {\snapshotImage{colorGreenBase}};
\node[inlineImage] at (s23) {\snapshotImage{colorGreenBase}};
\node[inlineImage] at (s33) {\snapshotImage{colorGreenBase}};

\node[inlineImage] at (4.7,2) {\snapshotImage{colorGreenBase}};

\node[inlineImage] at (4.9,1) {\failureImage};
\node[inlineImage] at (8.8+1,2) {\failureImage};
\end{tikzpicture}
\else
hideTikz is set
\fi
\end{minipage}
\caption{Original faulty execution}\label{fig:failure-transparency:original}
\end{subfigure}\\%
\begin{subfigure}{0.45\textwidth}
\centering
\begin{minipage}{1.0\textwidth}
\ifx\hideTikz\undefined
\begin{tikzpicture}[
  scale=0.75, every node/.style={scale=1},
  thick, colorGrayBase,
  timeline/.style={-{Triangle[angle=45:5]}, thick},
  message/.style={-{Triangle[angle=45:5]}, very thick, controls={+(down:0.9) and +(up:1.1)}},
  messageNode/.style={scale=0.4, circle, fill, prefix after command={\pgfextra{\tikzset{every label/.style={label distance=-1.5mm, scale=0.75}}}}},
  inlineImage/.style={scale=0.35, yscale=-1, xshift=-12, yshift=-12},
  snapshotEdge/.style={very thick, handdrawn}
]
\coordinate (s12) at (2.9,2);
\coordinate (s22) at (3.9,1);
\coordinate (s32) at (5.2,0);

\begin{scope}
\clip (0,-0.25)--(0,2)--(s12) .. controls +(down:0.5) and +(up:0.5) .. (s22) .. controls +(down:0.5) and +(up:0.5) .. (s32)--($(s32)+(0,-0.25)$)--(0,-0.25);
\fill[colorBlueBack] (0,-0.25)--(0,2.1)--(8.2,2.1)--(8.2,-0.25)--(0,-0.25);
\end{scope}

\draw[colorBlueEmph, snapshotEdge](s12) .. controls +(down:0.5) and +(up:0.5) .. (s22) .. controls +(down:0.5) and +(up:0.5) .. (s32)--($(s32)+(0,-0.25)$);

\coordinate (s1) at (2,2);
\coordinate (s2) at (2.5,1);
\coordinate (s3) at (3.2,0);

\begin{scope}
\clip (0,-0.25)--(0,2)--(s1) .. controls +(down:0.5) and +(up:0.5) .. (s2) .. controls +(down:0.5) and +(up:0.5) .. (s3)--($(s3)+(0,-0.25)$)--(0,-0.25);
\fill[colorRedBack] (0,-0.25)--(0,2.1)--(5,2.1)--(5,-0.25)--(0,-0.25);
\end{scope}

\draw[colorRedEmph, snapshotEdge](s1) .. controls +(down:0.5) and +(up:0.5) .. (s2) .. controls +(down:0.5) and +(up:0.5) .. (s3)--($(s3)+(0,-0.25)$);

\node[label={left:{$p_1$}},scale=0] at (0,2) {};
\draw[timeline,-](0,2) -- (6.5,2);
\node[label={left:{$p_2$}},scale=0] at (0,1) {};
\draw[timeline](0,1) -- (4.7,1);
\draw[densely dashed, thin, colorBlueEmph] (6.5,-0.25) -- (6.5,2.25);
\node[label={[transparent,scale=0.75]above:{${\color{colorBlueBase}e_l}\leq e < {\color{colorGreenBase}e_r}$}},scale=0] at (5.5,2.2) {};
\node[label={[scale=0.75, colorBlueBase]above:{recover}},scale=0] at (6.5,2.25) {};
\node[label={left:{$p_3$}},scale=0] at (0,0) {};
\draw[timeline,-](0,0) -- (6.5,0);

\node[label={[scale=0.75, colorRedBase]below:{commit}},scale=0] at ($(s3.center)+(0,-0.2)$) {};
\node[label={[scale=0.75, colorBlueBase]below:{commit}},scale=0] at ($(s32.center)+(0,-0.2)$) {};

\node[messageNode] (m11) at (0.5,2) {};
\draw[message](m11.center) to (1.1,1);
\node[messageNode] (m20) at (0.4,1) {};
\draw[message](m20.center) to (1,0);
\node[messageNode] (m21) at (1.1,1) {};
\draw[message](m21.center) to (1.8,0);
\node[messageNode] (m22p) at (1.9,1) {};
\draw[message](m22p.center) to (2.6,0);

\node[messageNode] (m13) at (2.5,2) {};
\draw[message,controls={+(down:1.1) and +(up:0.9)}](m13.center) to (3.5,1);
\node[messageNode] (m22) at (3.5,1) {};
\draw[message](m22.center) to (4,0);

\node[messageNode] (m3) at (3.7,2) {};
\draw[message](m3.center) to (4.3,1);
\node[messageNode] (m4) at (5.5,2) {};
\draw[message](m4.center) to (6.1,1);

\node[inlineImage] at (s1) {\snapshotImage{colorRedBase}};
\node[inlineImage] at (s2) {\snapshotImage{colorRedBase}};
\node[inlineImage] at (s3) {\snapshotImage{colorRedBase}};
\node[inlineImage] at (s12) {\snapshotImage{colorBlueBase}};
\node[inlineImage] at (s22) {\snapshotImage{colorBlueBase}};
\node[inlineImage] at (s32) {\snapshotImage{colorBlueBase}};

\node[inlineImage] at (4.7,2) {\snapshotImage{colorGreenBase}};

\node[inlineImage] at (4.9,1) {\failureImage};
\end{tikzpicture}
\else
hideTikz is set
\fi
\end{minipage}
\caption{Separated generations}\label{fig:failure-transparency:generation}
\end{subfigure}%
\begin{subfigure}{0.55\textwidth}
\centering
\begin{minipage}{1.0\textwidth}
\ifx\hideTikz\undefined
\begin{tikzpicture}[
  scale=0.75, every node/.style={scale=1},
  thick, colorGrayBase,
  timeline/.style={-{Triangle[angle=45:5]}, thick},
  message/.style={-{Triangle[angle=45:5]}, very thick, controls={+(down:0.9) and +(up:1.1)}},
  messageNode/.style={scale=0.4, circle, fill, prefix after command={\pgfextra{\tikzset{every label/.style={label distance=-1.5mm, scale=0.75}}}}},
  inlineImage/.style={scale=0.35, yscale=-1, xshift=-12, yshift=-12},
  snapshotEdge/.style={very thick, handdrawn}
]
\coordinate (s12) at (2.9,2);
\coordinate (s22) at (3.9,1);
\coordinate (s32) at (5.2,0);

\begin{scope}
\clip (0,-0.25)--(0,2)--(s12) .. controls +(down:0.5) and +(up:0.5) .. (s22) .. controls +(down:0.5) and +(up:0.5) .. (s32)--($(s32)+(0,-0.25)$)--(0,-0.25);
\fill[colorBlueBack] (0,-0.25)--(0,2.1)--(8.2,2.1)--(8.2,-0.25)--(0,-0.25);
\end{scope}

\draw[colorBlueEmph, snapshotEdge](s12) .. controls +(down:0.5) and +(up:0.5) .. (s22) .. controls +(down:0.5) and +(up:0.5) .. (s32)--($(s32)+(0,-0.25)$);

\coordinate (s1) at (2,2);
\coordinate (s2) at (2.5,1);
\coordinate (s3) at (3.2,0);

\begin{scope}
\clip (0,-0.25)--(0,2)--(s1) .. controls +(down:0.5) and +(up:0.5) .. (s2) .. controls +(down:0.5) and +(up:0.5) .. (s3)--($(s3)+(0,-0.25)$)--(0,-0.25);
\fill[colorRedBack] (0,-0.25)--(0,2.1)--(5,2.1)--(5,-0.25)--(0,-0.25);
\end{scope}

\draw[colorRedEmph, snapshotEdge](s1) .. controls +(down:0.5) and +(up:0.5) .. (s2) .. controls +(down:0.5) and +(up:0.5) .. (s3)--($(s3)+(0,-0.25)$);

\node[label={left:{$p_1$}},scale=0] at (0,2) {};
\draw[timeline,-](0,2) -- (6.5+2,2);
\node[label={left:{$p_2$}},scale=0] at (0,1) {};
\draw[timeline](0,1) -- (4.7+2,1);
\draw[densely dashed, thin, colorBlueEmph] (6.5+2,-0.25) -- (6.5+2,2.25);
\node[label={[scale=0.75, colorBlueBase]above:{recover}},scale=0] at (6.5+2,2.25) {};
\node[label={left:{$p_3$}},scale=0] at (0,0) {};
\draw[timeline,-](0,0) -- (6.5+2,0);
\draw[densely dotted, thin] (5.5,-0.25) -- (5.5,2.25);
\node[label={[scale=0.75]above:{${\color{colorBlueBase}e_l}\leq e < {\color{colorGreenBase}e_r}$}},scale=0] at (5.5,2.2) {};

\node[label={[scale=0.75, colorRedBase]below:{commit}},scale=0] at ($(s3.center)+(0,-0.2)$) {};
\node[label={[scale=0.75, colorBlueBase]below:{commit}},scale=0] at ($(s32.center)+(0,-0.2)$) {};

\node[messageNode] (m11) at (0.5,2) {};
\draw[message](m11.center) to (1.1,1);
\node[messageNode] (m20) at (0.4,1) {};
\draw[message](m20.center) to (1,0);
\node[messageNode] (m21) at (1.1,1) {};
\draw[message](m21.center) to (2,0);
\node[messageNode] (m22p) at (1.9,1) {};
\draw[message](m22p.center) to (2.6,0);

\node[messageNode] (m13) at (2.5,2) {};
\draw[message,controls={+(down:1.1) and +(up:0.9)}](m13.center) to (3.5,1);
\node[messageNode] (m22) at (3.5,1) {};
\draw[message](m22.center) to (4,0);

\node[messageNode] (m3) at (3.7+2,2) {};
\draw[message](m3.center) to (4.3+2,1);
\node[messageNode] (m4) at (5.5+2,2) {};
\draw[message](m4.center) to (6.1+2,1);

\node[inlineImage] at (s1) {\snapshotImage{colorRedBase}};
\node[inlineImage] at (s2) {\snapshotImage{colorRedBase}};
\node[inlineImage] at (s3) {\snapshotImage{colorRedBase}};
\node[inlineImage] at (s12) {\snapshotImage{colorBlueBase}};
\node[inlineImage] at (s22) {\snapshotImage{colorBlueBase}};
\node[inlineImage] at (s32) {\snapshotImage{colorBlueBase}};

\node[inlineImage] at (4.7+2,2) {\snapshotImage{colorGreenBase}};

\node[inlineImage] at (4.9+2,1) {\failureImage};
\end{tikzpicture}
\else
hideTikz is set
\fi
\end{minipage}
\caption{Reordered generations}\label{fig:failure-transparency:reordering}
\end{subfigure}\\%
\begin{subfigure}{1.0\textwidth}
\centering
\begin{minipage}{1.0\textwidth}
\ifx\hideTikz\undefined
\begin{tikzpicture}[
  scale=0.75, every node/.style={scale=1},
  thick, colorGrayBase,
  timeline/.style={-{Triangle[angle=45:5]}, thick},
  message/.style={-{Triangle[angle=45:5]}, very thick, controls={+(down:0.9) and +(up:1.1)}},
  messageNode/.style={scale=0.4, circle, fill, prefix after command={\pgfextra{\tikzset{every label/.style={label distance=-1.5mm, scale=0.75}}}}},
  inlineImage/.style={scale=0.35, yscale=-1, xshift=-12, yshift=-12},
  snapshotEdge/.style={very thick, handdrawn}
]
\coordinate (s13) at (8-0.6,2);
\coordinate (s23) at (9-0.6,1);
\coordinate (s33) at (10-0.6,0);

\begin{scope}
\clip (0,-0.25)--(0,2)--(s13) .. controls +(down:0.5) and +(up:0.5) .. (s23) .. controls +(down:0.5) and +(up:0.5) .. (s33)--($(s33)+(0,-0.25)$)--(0,-0.25);
\fill[colorGreenBack] (0,-0.25)--(0,2.1)--(10.5,2.1)--(10.5,-0.25)--(0,-0.25);
\end{scope}

\draw[colorGreenEmph, snapshotEdge](s13) .. controls +(down:0.5) and +(up:0.5) .. (s23) .. controls +(down:0.5) and +(up:0.5) .. (s33)--($(s33)+(0,-0.25)$);

\coordinate (s12) at (2.9,2);
\coordinate (s22) at (3.9,1);
\coordinate (s32) at (5.2,0);

\begin{scope}
\clip (0,-0.25)--(0,2)--(s12) .. controls +(down:0.5) and +(up:0.5) .. (s22) .. controls +(down:0.5) and +(up:0.5) .. (s32)--($(s32)+(0,-0.25)$)--(0,-0.25);
\fill[colorBlueBack] (0,-0.25)--(0,2.1)--(8.2,2.1)--(8.2,-0.25)--(0,-0.25);
\end{scope}

\draw[colorBlueEmph, snapshotEdge](s12) .. controls +(down:0.5) and +(up:0.5) .. (s22) .. controls +(down:0.5) and +(up:0.5) .. (s32)--($(s32)+(0,-0.25)$);

\coordinate (s1) at (2,2);
\coordinate (s2) at (2.5,1);
\coordinate (s3) at (3.2,0);

\begin{scope}
\clip (0,-0.25)--(0,2)--(s1) .. controls +(down:0.5) and +(up:0.5) .. (s2) .. controls +(down:0.5) and +(up:0.5) .. (s3)--($(s3)+(0,-0.25)$)--(0,-0.25);
\fill[colorRedBack] (0,-0.25)--(0,2.1)--(5,2.1)--(5,-0.25)--(0,-0.25);
\end{scope}

\draw[colorRedEmph, snapshotEdge](s1) .. controls +(down:0.5) and +(up:0.5) .. (s2) .. controls +(down:0.5) and +(up:0.5) .. (s3)--($(s3)+(0,-0.25)$);

\node[label={left:{$p_1$}},scale=0] at (0,2) {};
\draw[timeline](0,2) -- (10.5,2);
\node[label={left:{$p_2$}},scale=0] at (0,1) {};
\draw[timeline](0,1) -- (10.5,1);
\node[label={left:{$p_3$}},scale=0] at (0,0) {};
\draw[timeline](0,0) -- (10.5,0);

\node[label={[scale=0.75, colorRedBase]below:{commit}},scale=0] at ($(s3.center)+(0,-0.2)$) {};
\node[label={[scale=0.75, colorBlueBase]below:{commit}},scale=0] at ($(s32.center)+(0,-0.2)$) {};
\node[label={[scale=0.75, colorGreenBase]below:{commit}},scale=0] at ($(s33.center)+(0,-0.2)$) {};

\node[messageNode] (m11) at (0.5,2) {};
\draw[message](m11.center) to (1.1,1);
\node[messageNode] (m20) at (0.4,1) {};
\draw[message](m20.center) to (1,0);
\node[messageNode] (m21) at (1.1,1) {};
\draw[message](m21.center) to (1.8,0);
\node[messageNode] (m22p) at (1.9,1) {};
\draw[message](m22p.center) to (2.6,0);

\node[messageNode] (m13) at (2.5,2) {};
\draw[message,controls={+(down:1.1) and +(up:0.9)}](m13.center) to (3.5,1);
\node[messageNode] (m22) at (3.5,1) {};
\draw[message](m22.center) to (4,0);

\node[messageNode] (m00) at (6,1) {};
\draw[message](m00.center) to (6.6,0);
\node[messageNode] (m01) at (6,2) {};
\draw[message](m01.center) to (6.8,1);
\node[messageNode] (m02) at (6.8,1) {};
\draw[message](m02.center) to (7.6,0);
\node[messageNode] (m03) at (7.4,1) {};
\draw[message](m03.center) to (8.6,0);

\node[inlineImage] at (s1) {\snapshotImage{colorRedBase}};
\node[inlineImage] at (s2) {\snapshotImage{colorRedBase}};
\node[inlineImage] at (s3) {\snapshotImage{colorRedBase}};
\node[inlineImage] at (s12) {\snapshotImage{colorBlueBase}};
\node[inlineImage] at (s22) {\snapshotImage{colorBlueBase}};
\node[inlineImage] at (s32) {\snapshotImage{colorBlueBase}};
\node[inlineImage] at (s13) {\snapshotImage{colorGreenBase}};
\node[inlineImage] at (s23) {\snapshotImage{colorGreenBase}};
\node[inlineImage] at (s33) {\snapshotImage{colorGreenBase}};
\end{tikzpicture}
\else
hideTikz is set
\fi
\end{minipage}
\caption{Merged failure-free execution}\label{fig:failure-transparency:merged}
\end{subfigure}
\caption{The step-wise construction of a failure-free execution trace from an execution with failures.}\label{fig:failure-transparency}
\end{figure}

The construction is done using traces; we reorder and manipulate the original trace so that failures, recoveries, and discarded trace steps are removed from it.
\Cref{fig:failure-transparency} illustrates the construction:
(1) first, we split the trace by the recovery steps into generations;
(2) next, the trace steps are reordered such that all discarded steps are moved to the end of the generation;
(3) then, these steps are safely discarded;
(4) finally, we concatenate the generations to get the final trace.

Next, we have to show that:
(i) the constructed trace is valid, \ie it corresponds to a failure-free execution; and
(ii) that the execution is an observational explanation of the original execution.
We do so by reasoning about the preservation of validity and observable outputs in each step of the construction.
For trace validity, the most complicated step is the reordering (step 2 of the construction).
We show that the reordering is causality-preserving and thus, by \Cref{lem:validity-similarity-causal}, it produces a valid trace.
For observational explanation, throughout the construction we maintain a mapping of observations from the steps of the original trace to the steps of the constructed trace.
The challenge lies in the reordering of steps (step 2 of the construction) and the fusion of generations (step 4 of the construction).
For the reordering, we show a lemma that the observable output is not changed by the discarded steps; and, for the fusion, we show that the latest common snapshot of a generation is exactly the configuration obtained by the reordering and removal of the discarded steps.
This, accompanied by an analysis of the rules, lets us show that the sequence of observable outputs is the same for the original and the failure-free traces.

\begin{proof}
Expanding the definitions, we need to prove that,
for all executions in $\rulesI$ with potential failures, there is an observational explanation in the failure-free model $(\rulesI \setminus \rulesF)$.
Given an arbitrary execution $C$ of length $n$ in $\rulesI$ from initial configuration $c \in K$, \ie $\seq{C_i}{i}{n} \in \mathbb{E}_{c}^{\mathrm{I}}$, the goal is to construct a failure-free execution $\seq{C'_j}{j}{n'}$ such that:
\[
\seq{C'_j}{j}{n'} \in \mathbb{E}_{c}^{\mathrm{I}\setminus\mathrm{F}} \land
\forall m < n .~ \exists m' < n' .~
\mathrm{out}(C_{m}) = \mathrm{out}(C'_{m'})
\]

This execution is constructed indirectly, by first constructing a trace $Z'$ which then generates it.
First, we need to prove that the constructed failure-free trace $Z'$ is valid from $c$, \ie $Z(c)$;
next, we need to show that the corresponding execution $C' = Z'(c)$ is an observational explanation of the original execution, that is, for each configuration in the original execution, we have to provide an observationally equal configuration in the constructed execution.
From the original trace $Z$, for which $Z(c) = C$, we construct the failure-free trace $Z'$ in four steps as outlined in the proof sketch and illustrated in \Cref{fig:failure-transparency}.

(1) First, the trace is split by the recovery steps into generations, giving us a sequence of generations $G$ (\Cref{fig:failure-transparency:generation}).
Each generation is a sequence of $\mathsc{S-Step}$s ending with an $\mathsc{F-Recover}$ step; in the case of the last generation it may not necessarily end with an $\mathsc{F-Recover}$ step.
By construction, each generation is a valid trace as each of them is a contiguous part of a valid trace.
We construct the observability mapping by mapping the configurations of the original trace to their closest preceding committing border steps.
A \emph{committing border step} is an $\mathsc{I-Border}$ step which changes the greatest common epoch number, $\mathrm{gce}$, and thus also the observed output, $\mathrm{out}$; such steps are labeled with ``commit'' in \Cref{fig:failure-transparency}.
The equality of observations holds, since, by inspection of the rules, only a committing border step can change the observable output~\appCite{lem:committing-border-output}{Lemma B.6}.

(2) Next, from each generation $g = G_i$, we construct a new reordered trace $g' = G'_i$ so that all the steps of epochs above the greatest common epoch of the generation are placed after the steps of epochs below it (\Cref{fig:failure-transparency:reordering}).
In effect, this moves all the discarded steps to the end of the generation, since they are discarded by the recovery, which in turn is done to the greatest common epoch of the generation.
In other words, $g' = \mathrm{filter}(x \in g.~ \mathrm{epoch}(x) \leq e) \concat \mathrm{filter}(x \in g.~ \mathrm{epoch}(x) > e)$, where $e = \gce{g_{\size{g}-1}}$ is the epoch number to which the recovery is done.
The new traces are still valid, as the reordering is causality preserving~\appCite{lem:epoch-order}{Lemma B.7}, and thus the validity follows from \Cref{lem:validity-similarity-causal}.
The mapping of observations is kept intact, since the outputs of the committing border steps are not changed by the reordering~\appCite{lem:reordering-output}{Lemma B.8}.
This follows from the fact that the observable output is only changed by committing border steps~\appCite{lem:committing-border-output}{Lemma B.6}, that causality-preserving permutations result in the same configuration~(\Cref{lem:validity-similarity-causal}), and that the reordering is preserving causality.

(3) Then, from each $G'_i$ we construct a new trace $G''_i$ by removing the discarded steps and the recovery step.
That is, the suffix consisting of the failure steps, recovery steps, and any steps of epochs greater than the greatest common epoch of the generation are removed.
The new trace is a prefix of $G'_i$, and is thus still a valid trace~\appCite{lem:validity-prefix}{Lemma B.2}.
We keep the same mapping of observations for the steps that were not removed.
As, within a generation, only the suffix is removed, it does not affect the observed outputs of the remaining steps, and thus the mapping of observations is kept unchanged.

(4) Finally, we concatenate all stripped generations $G''_i$ to get the merged trace $Z'$ (\Cref{fig:failure-transparency:merged}).
We show that the last configuration of each of the generations $G''_i$ is exactly the latest common snapshot of the original generation $G_i$~\appCite{lem:snapshot,lem:snapshot-small}{Lemmas B.10-11}, in other words, the latest common snapshot is a view of a configuration as if only the committed steps occurred.
Since the recovery is done to the latest common snapshot, it is also the same configuration as the first configuration of the following generation $G''_{i+1}$.
For this reason, the concatenation of all generations forms a trace $Z'$ valid from $c$.
The observed outputs are not changed by the merge, and we maintain the same mapping.

By these four steps we have constructed a failure-free observational explanation of the faulty execution, which means that the implementation model is observationally explainable (\Cref{def:observational-explainability}) by its failure-free version, or, in other words, it is failure transparent (\Cref{def:failure-transparency}).
\end{proof}

\subsection{Liveness}
The proposed definition of failure transparency is a safety property~\cite{alpern1985defining,lamport1977proving}, \ie it prohibits the implementation from reaching invalid states.
Being as such, failure transparency does not require the implementation to take any observable execution steps; an implementation that never takes a step would trivially satisfy the property.
In contrast, ensuring that the implementation eventually does something is a \emph{liveness} property~\cite{alpern1985defining,lamport1977proving}.
To complete our analysis, we would like to show that the implementation model eventually produces outputs for all epochs in its input.
This is a liveness property which, consequently, does not concern itself with the correctness of the outputs.
However, in combination with the failure transparency property, the properties ensure that the presented implementation model \emph{eventually} produces the \emph{correct} outputs.
For this reason, we prove the following theorem about the liveness of the implementation model.

\begin{restatable}[Liveness of the Implementation Model]{theorem}{livenessOfImplementation}\label{thm:liveness}
For every input epoch present in the initial configuration, eventually a corresponding epoch appears in the output of a fair execution. That is:
\begin{gather*}
\forall k = \langle \Pi, \Sigma, M, N, D \rangle \in K.~ \forall C \in \mathbb{E}_k^\rulesI.~ \mathrm{fair}(C) \implies\\
\quad \forall (n\,s\,\langle e, d \rangle) \in M.~ \exists c \in C.~ \exists (n'\,s'\,\langle e', d' \rangle) \in \mathrm{out}(c).~ e = e'
\end{gather*}
where a $\mathrm{fair}$ execution is maximal (\ie it is not a prefix of another execution), has a finite amount of failures, and eventually executes any step which is {}eventually always enabled \techReport{(see \Cref{def:fair-execution} for the formal definition of $\mathrm{fair}$ execution)}\mainReport{(see the technical report~\cite{veresov2024technicalreport} for the formal definition of $\mathrm{fair}$ execution)}.
\lipicsEnd
\end{restatable}

The liveness theorem states that for any fair execution $C$ starting from a valid initial configuration $k$, and for all input epochs $e$, eventually there is a configuration $c$ in the execution for which the output $\mathrm{out}(c)$ contains the epoch $e$.

\begin{proof}
\mainReport{The complete proof is available in the companion technical report~\cite{veresov2024technicalreport}, it is summarized as follows. }%
\techReport{The complete proof is available in~\ref{prf:liveness}, it is summarized as follows. }%
First, we show that it suffices to demonstrate that, continuing from any configuration $c$ reachable from the valid initial configuration $k$, one or both of the following are true: eventually there is a failure; or eventually the epoch is visible in the output.
As the considered executions have only finite amounts of failures, we further simplify the proof goal: it suffices to show that eventually the epoch appears in the output under the assumption that there are no more failures.
We handle this simplified case by inductive reasoning on the acyclic dataflow graph of processors.
The induction's base case is the graph consisting of the source input streams but with no processors.
The induction hypothesis states that all streams are well-formed and that the border message of all input epochs eventually appear on all streams; this is satisfied for the base case by validity of the initial configuration.
Then, in the induction step, we construct the graph by adding one processor at a time, given that all of its input streams are already handled, as either source inputs or as outputs of other processors in the previous step's graph.
The assumption of fair scheduling allows us to reason about the processor locally, since, by definition of fairness, if a message has arrived to the processor, it will eventually be consumed.
As a conclusion of the induction, each processor will eventually have processed a border of each epoch present in the initial configuration; thus, eventually all processors will process a border of each initial epoch.
This, in turn, by analysis of $\mathsc{I-Border}$, $\mathrm{gce}$, and $\mathrm{out}$, shows that the border messages of the epoch will eventually be in the output.
\end{proof}

\section{Related Work}\label{sec:related-work}

\subparagraph*{Failure Transparency, Observational Explainability.}
There has been a significant body of research on failure transparency~\cite{Lee1990}.
To our knowledge, the earliest work on failure transparency was by von Neumann in 1956~\cite{neumann1956probabilistic} on creating reliable systems from unreliable components.
Later work by Wensley in 1972~\cite{DBLP:conf/afips/Wensley72} discussed software techniques for failure transparent computing.
Lowell and Chen discussed failure transparency in the context of consistent failure recovery protocols~\cite{lowell99theory}.
In their work, they introduced ``equivalence functions'' for comparing executions, a concept which inspired the observability functions in this paper.
Our work, in contrast, restricts these functions to be monotonic, and discusses their application to both levels (low and high) of the system, which facilitates the presented transitivity lemma (\Cref{lem:transitivity}).
Around the same time as Lowell and Chen, Gärtner discussed general models for fault-tolerant computing~\cite{gartner1999fundamentals}.
Similar to our work, Gärtner separated fault-tolerant programs into two separate sets of rules (actions): the rules for normal behavior; and the rules for failure behavior.
With this separation, Gärtner discussed various properties and forms of fault-tolerant programs.
In the context of Gärtner's work, our definition of failure transparency would be considered ``failure masking'', in the sense that the system can recover from failures and continue its normal operation.
Whereas these works defined failure transparency as a conjunct of \emph{safety and liveness}~\cite{lamport1977proving,alpern1985defining}, we have only considered its safety property for our definition.

The presented definition for observational explainability is closely related to previous definitions of refinement (\eg TLA~\cite{lamport2002specifying,lamport1994temporal}, Compiler Correctness~\cite{patterson2019next}), implementation (\eg I/O Automata~\cite{lynch1988introduction}), and simulation.
In simplified terms, one set of executions implements another if it is a subset thereof (modulo stuttering and multistep executions).
Our definition of observational explainability, in some sense, extends the notion of refinement to directly include a refinement mapping~\cite{abadi1991existence} on both sides via observability functions.
It resembles notions from related work such as observational equivalence~\cite{burckhardt2021durable} and observational refinement~\cite{kallas2023executing}; in contrast to these works, we provide a formal definition thereof.
Different from inductive proof approaches as typical for TLA~\cite{lamport2002specifying} and simulation proof strategies, our proof approach reasons about the whole sequence.
This makes it not necessary to include notions for ghost variables ~\cite{marcus1996using} (also known as auxiliary variables~\cite{lamport2017auxiliary}) for the purpose of reasoning about past or future events.

\subparagraph*{Failure Transparency Proofs.}
Failure transparency and observational explainability can be proven in various ways.
For example, Burckhardt et al.~\cite{burckhardt2021durable} prove ``observational equivalence'' for their serverless programming model.
Mukherjee et al.~\cite{mukherjee2019reliable} propose a failure transparency theorem for their system of reliable state machines: an execution of the implementation is a refinement of an execution without failures ``with respect to its observable behavior'', reminiscent of our definition of failure transparency.
Other works include models for distributed reliable actor communication~\cite{tardieu2023reliable}, serverless microservices and observational refinement~\cite{kallas2023executing}, and reliable state machines~\cite{mukherjee2019reliable}.
Their specific approaches may differ, some use simulation~\cite{burckhardt2021durable,kallas2023executing}, others model failures explicitly~\cite{kallas2023executing,tardieu2023reliable,mukherjee2019reliable}, and others use notions similar to observability functions~\cite{burckhardt2021durable}.
Another approach is to prove the proper restoration of applications to the exact configuration as before the crash~\cite{mogk2019fault}.
Our presented failure transparency proof shares similarities to the proof of the Asynchronous Barrier Snapshotting protocol~\cite{carbone2018thesis}, such as reasoning about causal orderings; however, our proof relies to a greater degree on abstraction in terms of refinement of models.

\subparagraph*{Distributed, Resilient Programming Models.}
Stateful dataflow has had a high impact~\cite{fragkoulis2023survey} through systems such as: MapReduce~\cite{dean2008mapreduce}, Apache Spark~\cite{zaharia2010spark,ZahariaCDDMMFSS12}, Apache Flink~\cite{carbone2015lightweight}, Google Dataflow~\cite{akidau2015dataflow}, IBM Streams~\cite{jacques2016consistent}, Portals~\cite{spenger2022portals}, and others~\cite{balazinska2005fault,shah2004highly}.
However, there are other notable resilient programming models and systems, including: Pregel, a graph-based system~\cite{malewicz2010pregel}, Resilient X10~\cite{cunningham2014resilient}, virtually resilient immortals~\cite{goldstein2020ambrosia}, fault-tolerant reactives~\cite{mogk2019fault}, thread-safe reactive programming~\cite{DBLP:journals/pacmpl/DrechslerMSM18}, Durable Functions~\cite{burckhardt2021durable}, stateful entities~\cite{DBLP:conf/edbt/PsarakisZFSK24}, the eXchange Calculus~\cite{DBLP:journals/jss/AudritoCDSV24}, and others~\cite{tardieu2023reliable,kallas2023executing,mukherjee2019reliable,DBLP:conf/cidr/CheungCHM21}.
In general, these resilient programming models provide system means to recover from failures, the user does not need to implement the failure recovery mechanisms themselves.
Actor models, in contrast, provide the users with manual failure-handling constructs.
For example, the failure-handling constructs in Erlang, such as actor monitors and supervision~\cite{Armstrong93}, have been used successfully for building reliable services within the telecom industry~\cite{armstrong1996erlang}.
Moreover, other programming models such as Argus~\cite{liskov1988distributed} and transactors~\cite{field2005transactors} provide constructs for transactions, which in turn can be used for building reliable services.

The formalization of distributed systems has been a long-standing research topic.
Notably, formalization frameworks such as TLA~\cite{lamport2002specifying} and I/O Automata~\cite{lynch1988introduction}, have been used to reason about distributed systems.
Examples of this include a dataflow system that was formalized using I/O Automata~\cite{lynch1989proof}.
The ABS protocol for stateful dataflow has been formalized with transition systems~\cite{carbone2018thesis}.
Recently, operational semantics have been used to model and reason about such systems~\cite{burckhardt2021durable,tardieu2023reliable,kallas2023executing,mukherjee2019reliable,haller2018programming}.

\subparagraph*{Failure Recovery.}
A general overview of rollback-recovery protocols was given by Elnozahy et al.~\cite{elnozahy2002survey}, comparing between checkpointing-based and logging-based protocols.
Stateful dataflow systems use either checkpointing, or a combination of the two~\cite{balazinska2005fault,shah2004highly,akidau2015dataflow,carbone2015lightweight,wang2019lineage,ZahariaCDDMMFSS12,dean2008mapreduce,jacques2016consistent}.
The MapReduce system performs failure recovery by detecting failed nodes, and replaying the computation from sources or from persisted intermediate results~\cite{dean2008mapreduce}.
Apache Spark, in contrast, improves the recovery by replaying from the sources through what is called lineage recovery~\cite{ZahariaCDDMMFSS12}.
A similar idea is used in a dynamic dataflow system within Ray~\cite{wang2019lineage}.
This paper focused on the ABS protocol used in Apache Flink, which, in contrast to previous works, uses an asynchronous checkpointing technique~\cite{carbone2015lightweight}.
It has been proven to provide high performance and has since been widely adopted~\cite{siachamis2024checkmate}.
The current version of Apache Flink's runtime offers an opt-in feature for ``unaligned checkpoints'', which allow the checkpoint markers to be treated at a higher priority, decreasing the end-to-end latency at the cost of some overhead as buffered events may become part of the snapshots~\cite{flinkunalignedcheckpoints}.
Other adaptations of the Flink protocol include Clonos~\cite{silvestre2021clonos}, which logs the nondeterminism to facilitate faster partial recovery after failures.
Failure recovery remains an open research topic, as it has great impact on the performance characteristics of fault-tolerant systems~\cite{siachamis2024checkmate}.

\section{Conclusions and Future Work}\label{sec:conclusions-and-future-work}
This paper studies failure transparency of stateful dataflow systems. We propose a novel definition of failure transparency for programming models expressed in small-step operational semantics.
For the definition of failure transparency we introduce observational explainability, a notion which resembles refinement but on the level of observations of executions.
We provide an implementation model of a stateful dataflow system using the Asynchronous Barrier Snapshotting protocol in a small-step operational semantics, and prove that the model is failure transparent and guarantees liveness.

In future work, we plan to implement a fully verified implementation of a stateful dataflow system based on the semantics presented in this paper, starting from our Coq mechanization.
Furthermore, we would like to apply our definitions to existing related work.

\bibliography{base}

\begin{thebibliography}{10}

\bibitem{abadi1991existence}
Mart{\'{\i}}n Abadi and Leslie Lamport.
\newblock The existence of refinement mappings.
\newblock {\em Theor. Comput. Sci.}, 82(2):253--284, 1991.
\newblock \href {https://doi.org/10.1016/0304-3975(91)90224-P} {\path{doi:10.1016/0304-3975(91)90224-P}}.

\bibitem{akidau2015dataflow}
Tyler Akidau, Robert Bradshaw, Craig Chambers, Slava Chernyak, Rafael Fern{\'{a}}ndez{-}Moctezuma, Reuven Lax, Sam McVeety, Daniel Mills, Frances Perry, Eric Schmidt, and Sam Whittle.
\newblock The dataflow model: {A} practical approach to balancing correctness, latency, and cost in massive-scale, unbounded, out-of-order data processing.
\newblock {\em Proc. {VLDB} Endow.}, 8(12):1792--1803, 2015.
\newblock URL: \url{http://www.vldb.org/pvldb/vol8/p1792-Akidau.pdf}, \href {https://doi.org/10.14778/2824032.2824076} {\path{doi:10.14778/2824032.2824076}}.

\bibitem{alpern1985defining}
Bowen Alpern and Fred~B. Schneider.
\newblock Defining liveness.
\newblock {\em Inf. Process. Lett.}, 21(4):181--185, 1985.
\newblock \href {https://doi.org/10.1016/0020-0190(85)90056-0} {\path{doi:10.1016/0020-0190(85)90056-0}}.

\bibitem{armstrong1996erlang}
Joe Armstrong.
\newblock Erlang--a survey of the language and its industrial applications.
\newblock In {\em Proc. INAP}, volume~96, pages 16--18, 1996.

\bibitem{Armstrong93}
Joe Armstrong, Robert Virding, and Mike Williams.
\newblock {\em Concurrent programming in {ERLANG}}.
\newblock Prentice Hall, 1993.

\bibitem{DBLP:journals/jss/AudritoCDSV24}
Giorgio Audrito, Roberto Casadei, Ferruccio Damiani, Guido Salvaneschi, and Mirko Viroli.
\newblock The exchange calculus {(XC):} {A} functional programming language design for distributed collective systems.
\newblock {\em J. Syst. Softw.}, 210:111976, 2024.
\newblock \href {https://doi.org/10.1016/J.JSS.2024.111976} {\path{doi:10.1016/J.JSS.2024.111976}}.

\bibitem{balazinska2005fault}
Magdalena Balazinska, Hari Balakrishnan, Samuel Madden, and Michael Stonebraker.
\newblock Fault-tolerance in the borealis distributed stream processing system.
\newblock In Fatma {\"{O}}zcan, editor, {\em Proceedings of the {ACM} {SIGMOD} International Conference on Management of Data, Baltimore, Maryland, USA, June 14-16, 2005}, pages 13--24. {ACM}, 2005.
\newblock \href {https://doi.org/10.1145/1066157.1066160} {\path{doi:10.1145/1066157.1066160}}.

\bibitem{burckhardt2021durable}
Sebastian Burckhardt, Chris Gillum, David Justo, Konstantinos Kallas, Connor McMahon, and Christopher~S. Meiklejohn.
\newblock Durable functions: semantics for stateful serverless.
\newblock {\em Proc. {ACM} Program. Lang.}, 5({OOPSLA}):1--27, 2021.
\newblock \href {https://doi.org/10.1145/3485510} {\path{doi:10.1145/3485510}}.

\bibitem{cachin2011introduction}
Christian Cachin, Rachid Guerraoui, and Lu{\'{\i}}s E.~T. Rodrigues.
\newblock {\em Introduction to Reliable and Secure Distributed Programming {(2.} ed.)}.
\newblock Springer, 2011.
\newblock \href {https://doi.org/10.1007/978-3-642-15260-3} {\path{doi:10.1007/978-3-642-15260-3}}.

\bibitem{carbone2018thesis}
Paris Carbone.
\newblock {\em Scalable and Reliable Data Stream Processing}.
\newblock PhD thesis, Royal Institute of Technology, Stockholm, Sweden, 2018.
\newblock URL: \url{https://nbn-resolving.org/urn:nbn:se:kth:diva-233527}.

\bibitem{carbone2015lightweight}
Paris Carbone, Gyula F{\'{o}}ra, Stephan Ewen, Seif Haridi, and Kostas Tzoumas.
\newblock Lightweight asynchronous snapshots for distributed dataflows.
\newblock {\em CoRR}, abs/1506.08603, 2015.
\newblock URL: \url{http://arxiv.org/abs/1506.08603}, \href {https://arxiv.org/abs/1506.08603} {\path{arXiv:1506.08603}}.

\bibitem{carbone2015apache}
Paris Carbone, Asterios Katsifodimos, Stephan Ewen, Volker Markl, Seif Haridi, and Kostas Tzoumas.
\newblock Apache flink{\texttrademark}: Stream and batch processing in a single engine.
\newblock {\em {IEEE} Data Eng. Bull.}, 38(4):28--38, 2015.
\newblock URL: \url{http://sites.computer.org/debull/A15dec/p28.pdf}.

\bibitem{chandy1985distributed}
K.~Mani Chandy and Leslie Lamport.
\newblock Distributed snapshots: Determining global states of distributed systems.
\newblock {\em {ACM} Trans. Comput. Syst.}, 3(1):63--75, 1985.
\newblock \href {https://doi.org/10.1145/214451.214456} {\path{doi:10.1145/214451.214456}}.

\bibitem{DBLP:conf/cidr/CheungCHM21}
Alvin Cheung, Natacha Crooks, Joseph~M. Hellerstein, and Mae Milano.
\newblock New directions in cloud programming.
\newblock In {\em 11th Conference on Innovative Data Systems Research, {CIDR} 2021, Virtual Event, January 11-15, 2021, Online Proceedings}. www.cidrdb.org, 2021.
\newblock URL: \url{http://cidrdb.org/cidr2021/papers/cidr2021\_paper16.pdf}.

\bibitem{DBLP:journals/pacmpl/ChoiVSCA17}
Joonwon Choi, Muralidaran Vijayaraghavan, Benjamin Sherman, Adam Chlipala, and Arvind.
\newblock Kami: a platform for high-level parametric hardware specification and its modular verification.
\newblock {\em Proc. {ACM} Program. Lang.}, 1({ICFP}):24:1--24:30, 2017.
\newblock \href {https://doi.org/10.1145/3110268} {\path{doi:10.1145/3110268}}.

\bibitem{cunningham2014resilient}
David Cunningham, David Grove, Benjamin Herta, Arun Iyengar, Kiyokuni Kawachiya, Hiroki Murata, Vijay~A. Saraswat, Mikio Takeuchi, and Olivier Tardieu.
\newblock Resilient {X10:} efficient failure-aware programming.
\newblock In Jos{\'{e}}~E. Moreira and James~R. Larus, editors, {\em {ACM} {SIGPLAN} Symposium on Principles and Practice of Parallel Programming, PPoPP '14, Orlando, FL, USA, February 15-19, 2014}, pages 67--80. {ACM}, 2014.
\newblock \href {https://doi.org/10.1145/2555243.2555248} {\path{doi:10.1145/2555243.2555248}}.

\bibitem{jacques2016consistent}
Gabriela~Jacques da~Silva, Fang Zheng, Daniel Debrunner, Kun{-}Lung Wu, Victor Dogaru, Eric Johnson, Michael Spicer, and Ahmet~Erdem Sariy{\"{u}}ce.
\newblock Consistent regions: Guaranteed tuple processing in {IBM} streams.
\newblock {\em Proc. {VLDB} Endow.}, 9(13):1341--1352, 2016.
\newblock URL: \url{http://www.vldb.org/pvldb/vol9/p1341-jacquesSilva.pdf}, \href {https://doi.org/10.14778/3007263.3007272} {\path{doi:10.14778/3007263.3007272}}.

\bibitem{dean2008mapreduce}
Jeffrey Dean and Sanjay Ghemawat.
\newblock Mapreduce: Simplified data processing on large clusters.
\newblock In Eric~A. Brewer and Peter Chen, editors, {\em 6th Symposium on Operating System Design and Implementation {(OSDI} 2004), San Francisco, California, USA, December 6-8, 2004}, pages 137--150. {USENIX} Association, 2004.
\newblock URL: \url{http://www.usenix.org/events/osdi04/tech/dean.html}.

\bibitem{DBLP:journals/pacmpl/DrechslerMSM18}
Joscha Drechsler, Ragnar Mogk, Guido Salvaneschi, and Mira Mezini.
\newblock Thread-safe reactive programming.
\newblock {\em Proc. {ACM} Program. Lang.}, 2({OOPSLA}):107:1--107:30, 2018.
\newblock \href {https://doi.org/10.1145/3276477} {\path{doi:10.1145/3276477}}.

\bibitem{elnozahy2002survey}
E.~N. Elnozahy, Lorenzo Alvisi, Yi{-}Min Wang, and David~B. Johnson.
\newblock A survey of rollback-recovery protocols in message-passing systems.
\newblock {\em {ACM} Comput. Surv.}, 34(3):375--408, 2002.
\newblock \href {https://doi.org/10.1145/568522.568525} {\path{doi:10.1145/568522.568525}}.

\bibitem{field2005transactors}
John Field and Carlos~A. Varela.
\newblock Transactors: a programming model for maintaining globally consistent distributed state in unreliable environments.
\newblock In Jens Palsberg and Mart{\'{\i}}n Abadi, editors, {\em Proceedings of the 32nd {ACM} {SIGPLAN-SIGACT} Symposium on Principles of Programming Languages, {POPL} 2005, Long Beach, California, USA, January 12-14, 2005}, pages 195--208. {ACM}, 2005.
\newblock \href {https://doi.org/10.1145/1040305.1040322} {\path{doi:10.1145/1040305.1040322}}.

\bibitem{flinkunalignedcheckpoints}
The Apache~Software Foundation.
\newblock Unaligned checkpoints flip-76.
\newblock \url{https://issues.apache.org/jira/browse/FLINK-14551}, 2020.
\newblock Accessed on 2024-03-28.

\bibitem{fragkoulis2023survey}
Marios Fragkoulis, Paris Carbone, Vasiliki Kalavri, and Asterios Katsifodimos.
\newblock A survey on the evolution of stream processing systems.
\newblock {\em {VLDB} J.}, 33(2):507--541, 2024.
\newblock \href {https://doi.org/10.1007/S00778-023-00819-8} {\path{doi:10.1007/S00778-023-00819-8}}.

\bibitem{fu2021real}
Yupeng Fu and Chinmay Soman.
\newblock Real-time data infrastructure at uber.
\newblock In Guoliang Li, Zhanhuai Li, Stratos Idreos, and Divesh Srivastava, editors, {\em {SIGMOD} '21: International Conference on Management of Data, Virtual Event, China, June 20-25, 2021}, pages 2503--2516. {ACM}, 2021.
\newblock \href {https://doi.org/10.1145/3448016.3457552} {\path{doi:10.1145/3448016.3457552}}.

\bibitem{gartner1999fundamentals}
Felix~C. G{\"{a}}rtner.
\newblock Fundamentals of fault-tolerant distributed computing in asynchronous environments.
\newblock {\em {ACM} Comput. Surv.}, 31(1):1--26, 1999.
\newblock \href {https://doi.org/10.1145/311531.311532} {\path{doi:10.1145/311531.311532}}.

\bibitem{goldstein2020ambrosia}
Jonathan Goldstein, Ahmed~S. Abdelhamid, Mike Barnett, Sebastian Burckhardt, Badrish Chandramouli, Darren Gehring, Niel Lebeck, Christopher Meiklejohn, Umar~Farooq Minhas, Ryan Newton, Rahee Peshawaria, Tal Zaccai, and Irene Zhang.
\newblock {A.M.B.R.O.S.I.A:} providing performant virtual resiliency for distributed applications.
\newblock {\em Proc. {VLDB} Endow.}, 13(5):588--601, 2020.
\newblock URL: \url{http://www.vldb.org/pvldb/vol13/p588-goldstein.pdf}, \href {https://doi.org/10.14778/3377369.3377370} {\path{doi:10.14778/3377369.3377370}}.

\bibitem{haller2018programming}
Philipp Haller, Heather Miller, and Normen M{\"{u}}ller.
\newblock A programming model and foundation for lineage-based distributed computation.
\newblock {\em J. Funct. Program.}, 28:e7, 2018.
\newblock \href {https://doi.org/10.1017/S0956796818000035} {\path{doi:10.1017/S0956796818000035}}.

\bibitem{DBLP:conf/cav/KaivolaGNTWPSTFRN09}
Roope Kaivola, Rajnish Ghughal, Naren Narasimhan, Amber Telfer, Jesse Whittemore, Sudhindra Pandav, Anna Slobodov{\'{a}}, Christopher Taylor, Vladimir~A. Frolov, Erik Reeber, and Armaghan Naik.
\newblock Replacing testing with formal verification in intel coretm i7 processor execution engine validation.
\newblock In Ahmed Bouajjani and Oded Maler, editors, {\em Computer Aided Verification, 21st International Conference, {CAV} 2009, Grenoble, France, June 26 - July 2, 2009. Proceedings}, volume 5643 of {\em Lecture Notes in Computer Science}, pages 414--429. Springer, 2009.
\newblock \href {https://doi.org/10.1007/978-3-642-02658-4\_32} {\path{doi:10.1007/978-3-642-02658-4\_32}}.

\bibitem{kallas2023executing}
Konstantinos Kallas, Haoran Zhang, Rajeev Alur, Sebastian Angel, and Vincent Liu.
\newblock Executing microservice applications on serverless, correctly.
\newblock {\em Proc. {ACM} Program. Lang.}, 7({POPL}):367--395, 2023.
\newblock \href {https://doi.org/10.1145/3571206} {\path{doi:10.1145/3571206}}.

\bibitem{DBLP:conf/sosp/KleinEHACDEEKNSTW09}
Gerwin Klein, Kevin Elphinstone, Gernot Heiser, June Andronick, David~A. Cock, Philip Derrin, Dhammika Elkaduwe, Kai Engelhardt, Rafal Kolanski, Michael Norrish, Thomas Sewell, Harvey Tuch, and Simon Winwood.
\newblock se{L}4: formal verification of an os kernel.
\newblock In Jeanna~Neefe Matthews and Thomas~E. Anderson, editors, {\em Proceedings of the 22nd {ACM} Symposium on Operating Systems Principles 2009, {SOSP} 2009, Big Sky, Montana, USA, October 11-14, 2009}, pages 207--220. {ACM}, 2009.
\newblock \href {https://doi.org/10.1145/1629575.1629596} {\path{doi:10.1145/1629575.1629596}}.

\bibitem{kreps2011kafka}
Jay Kreps, Neha Narkhede, Jun Rao, et~al.
\newblock Kafka: A distributed messaging system for log processing.
\newblock In {\em Proceedings of the NetDB}, volume~11, pages 1--7. Athens, Greece, 2011.

\bibitem{DBLP:conf/popl/KumarMNO14}
Ramana Kumar, Magnus~O. Myreen, Michael Norrish, and Scott Owens.
\newblock Cakeml: a verified implementation of {ML}.
\newblock In Suresh Jagannathan and Peter Sewell, editors, {\em The 41st Annual {ACM} {SIGPLAN-SIGACT} Symposium on Principles of Programming Languages, {POPL} '14, San Diego, CA, USA, January 20-21, 2014}, pages 179--192. {ACM}, 2014.
\newblock \href {https://doi.org/10.1145/2535838.2535841} {\path{doi:10.1145/2535838.2535841}}.

\bibitem{lamport1977proving}
Leslie Lamport.
\newblock Proving the correctness of multiprocess programs.
\newblock {\em {IEEE} Trans. Software Eng.}, 3(2):125--143, 1977.
\newblock \href {https://doi.org/10.1109/TSE.1977.229904} {\path{doi:10.1109/TSE.1977.229904}}.

\bibitem{DBLP:journals/cacm/Lamport78}
Leslie Lamport.
\newblock Time, clocks, and the ordering of events in a distributed system.
\newblock {\em Commun. {ACM}}, 21(7):558--565, 1978.
\newblock \href {https://doi.org/10.1145/359545.359563} {\path{doi:10.1145/359545.359563}}.

\bibitem{lamport1994temporal}
Leslie Lamport.
\newblock The temporal logic of actions.
\newblock {\em {ACM} Trans. Program. Lang. Syst.}, 16(3):872--923, 1994.
\newblock \href {https://doi.org/10.1145/177492.177726} {\path{doi:10.1145/177492.177726}}.

\bibitem{lamport1995write}
Leslie Lamport.
\newblock How to write a proof.
\newblock {\em The American mathematical monthly}, 102(7):600--608, 1995.

\bibitem{DBLP:journals/tocs/Lamport98}
Leslie Lamport.
\newblock The part-time parliament.
\newblock {\em {ACM} Trans. Comput. Syst.}, 16(2):133--169, 1998.
\newblock \href {https://doi.org/10.1145/279227.279229} {\path{doi:10.1145/279227.279229}}.

\bibitem{lamport2002specifying}
Leslie Lamport.
\newblock {\em Specifying Systems, The {TLA+} Language and Tools for Hardware and Software Engineers}.
\newblock Addison-Wesley, 2002.
\newblock URL: \url{http://research.microsoft.com/users/lamport/tla/book.html}.

\bibitem{lamport2012write}
Leslie Lamport.
\newblock How to write a 21 st century proof.
\newblock {\em Journal of fixed point theory and applications}, 11:43--63, 2012.

\bibitem{lamport2017auxiliary}
Leslie Lamport and Stephan Merz.
\newblock Auxiliary variables in {TLA+}.
\newblock {\em CoRR}, abs/1703.05121, 2017.
\newblock URL: \url{http://arxiv.org/abs/1703.05121}, \href {https://arxiv.org/abs/1703.05121} {\path{arXiv:1703.05121}}.

\bibitem{Lee1990}
Peter~Alan Lee and Thomas Anderson.
\newblock {\em Fault Tolerance}, pages 51--77.
\newblock Springer Vienna, Vienna, 1990.
\newblock \href {https://doi.org/10.1007/978-3-7091-8990-0_3} {\path{doi:10.1007/978-3-7091-8990-0_3}}.

\bibitem{DBLP:journals/cacm/Leroy09}
Xavier Leroy.
\newblock Formal verification of a realistic compiler.
\newblock {\em Commun. {ACM}}, 52(7):107--115, 2009.
\newblock \href {https://doi.org/10.1145/1538788.1538814} {\path{doi:10.1145/1538788.1538814}}.

\bibitem{liskov1988distributed}
Barbara Liskov.
\newblock Distributed programming in argus.
\newblock {\em Commun. {ACM}}, 31(3):300--312, 1988.
\newblock \href {https://doi.org/10.1145/42392.42399} {\path{doi:10.1145/42392.42399}}.

\bibitem{DBLP:phd/us/Lowell99}
David~E. Lowell.
\newblock {\em Theory and practice of failure transparency}.
\newblock PhD thesis, University of Michigan, {USA}, 1999.
\newblock URL: \url{https://hdl.handle.net/2027.42/132190}.

\bibitem{lowell99theory}
David~E. Lowell and Peter~M. Chen.
\newblock The theory and practice of failure transparency.
\newblock Technical report, University of Michigan, 1999.

\bibitem{lynch1988introduction}
Nancy Lynch and Mark Tuttle.
\newblock An introduction to input/output automata.
\newblock {\em CWI-Quarterly}, 2(3):219--246, 1989.
\newblock Also available as MIT Technical Memo MIT/LCS/TM-373, Laboratory for Computer Science, Massachusetts Institute of Technology.

\bibitem{lynch1989proof}
Nancy~A. Lynch and Eugene~W. Stark.
\newblock A proof of the kahn principle for input/output automata.
\newblock {\em Inf. Comput.}, 82(1):81--92, 1989.
\newblock \href {https://doi.org/10.1016/0890-5401(89)90066-7} {\path{doi:10.1016/0890-5401(89)90066-7}}.

\bibitem{malewicz2010pregel}
Grzegorz Malewicz, Matthew~H. Austern, Aart J.~C. Bik, James~C. Dehnert, Ilan Horn, Naty Leiser, and Grzegorz Czajkowski.
\newblock Pregel: a system for large-scale graph processing.
\newblock In Ahmed~K. Elmagarmid and Divyakant Agrawal, editors, {\em Proceedings of the {ACM} {SIGMOD} International Conference on Management of Data, {SIGMOD} 2010, Indianapolis, Indiana, USA, June 6-10, 2010}, pages 135--146. {ACM}, 2010.
\newblock \href {https://doi.org/10.1145/1807167.1807184} {\path{doi:10.1145/1807167.1807184}}.

\bibitem{mao2023streamops}
Yancan Mao, Zhanghao Chen, Yifan Zhang, Meng Wang, Yong Fang, Guanghui Zhang, Rui Shi, and Richard T.~B. Ma.
\newblock Streamops: Cloud-native runtime management for streaming services in bytedance.
\newblock {\em Proc. {VLDB} Endow.}, 16(12):3501--3514, 2023.
\newblock URL: \url{https://www.vldb.org/pvldb/vol16/p3501-mao.pdf}, \href {https://doi.org/10.14778/3611540.3611543} {\path{doi:10.14778/3611540.3611543}}.

\bibitem{marcus1996using}
Monica Marcus and Amir Pnueli.
\newblock Using ghost variables to prove refinement.
\newblock In Martin Wirsing and Maurice Nivat, editors, {\em Algebraic Methodology and Software Technology, 5th International Conference, {AMAST} '96, Munich, Germany, July 1-5, 1996, Proceedings}, volume 1101 of {\em Lecture Notes in Computer Science}, pages 226--240. Springer, 1996.
\newblock \href {https://doi.org/10.1007/BFB0014319} {\path{doi:10.1007/BFB0014319}}.

\bibitem{mogk2019fault}
Ragnar Mogk, Joscha Drechsler, Guido Salvaneschi, and Mira Mezini.
\newblock A fault-tolerant programming model for distributed interactive applications.
\newblock {\em Proc. {ACM} Program. Lang.}, 3({OOPSLA}):144:1--144:29, 2019.
\newblock \href {https://doi.org/10.1145/3360570} {\path{doi:10.1145/3360570}}.

\bibitem{mukherjee2019reliable}
Suvam Mukherjee, Nitin~John Raj, Krishnan Govindraj, Pantazis Deligiannis, Chandramouleswaran Ravichandran, Akash Lal, Aseem Rastogi, and Raja Krishnaswamy.
\newblock Reliable state machines: {A} framework for programming reliable cloud services.
\newblock In Alastair~F. Donaldson, editor, {\em 33rd European Conference on Object-Oriented Programming, {ECOOP} 2019, July 15-19, 2019, London, United Kingdom}, volume 134 of {\em LIPIcs}, pages 18:1--18:29. Schloss Dagstuhl - Leibniz-Zentrum f{\"{u}}r Informatik, 2019.
\newblock \href {https://doi.org/10.4230/LIPICS.ECOOP.2019.18} {\path{doi:10.4230/LIPICS.ECOOP.2019.18}}.

\bibitem{murray2013naiad}
Derek~Gordon Murray, Frank McSherry, Rebecca Isaacs, Michael Isard, Paul Barham, and Mart{\'{\i}}n Abadi.
\newblock Naiad: a timely dataflow system.
\newblock In Michael Kaminsky and Mike Dahlin, editors, {\em {ACM} {SIGOPS} 24th Symposium on Operating Systems Principles, {SOSP} '13, Farmington, PA, USA, November 3-6, 2013}, pages 439--455. {ACM}, 2013.
\newblock \href {https://doi.org/10.1145/2517349.2522738} {\path{doi:10.1145/2517349.2522738}}.

\bibitem{patterson2019next}
Daniel Patterson and Amal Ahmed.
\newblock The next 700 compiler correctness theorems (functional pearl).
\newblock {\em Proc. {ACM} Program. Lang.}, 3({ICFP}):85:1--85:29, 2019.
\newblock \href {https://doi.org/10.1145/3341689} {\path{doi:10.1145/3341689}}.

\bibitem{DBLP:conf/edbt/PsarakisZFSK24}
Kyriakos Psarakis, Wouter Zorgdrager, Marios Fragkoulis, Guido Salvaneschi, and Asterios Katsifodimos.
\newblock Stateful entities: Object-oriented cloud applications as distributed dataflows.
\newblock In Letizia Tanca, Qiong Luo, Giuseppe Polese, Loredana Caruccio, Xavier Oriol, and Donatella Firmani, editors, {\em Proceedings 27th International Conference on Extending Database Technology, {EDBT} 2024, Paestum, Italy, March 25 - March 28}, pages 15--21. OpenProceedings.org, 2024.
\newblock \href {https://doi.org/10.48786/EDBT.2024.02} {\path{doi:10.48786/EDBT.2024.02}}.

\bibitem{DBLP:conf/cav/ReidCDGHKPSVZ16}
Alastair Reid, Rick Chen, Anastasios Deligiannis, David Gilday, David Hoyes, Will Keen, Ashan Pathirane, Owen Shepherd, Peter Vrabel, and Ali Zaidi.
\newblock End-to-end verification of processors with isa-formal.
\newblock In Swarat Chaudhuri and Azadeh Farzan, editors, {\em Computer Aided Verification - 28th International Conference, {CAV} 2016, Toronto, ON, Canada, July 17-23, 2016, Proceedings, Part {II}}, volume 9780 of {\em Lecture Notes in Computer Science}, pages 42--58. Springer, 2016.
\newblock \href {https://doi.org/10.1007/978-3-319-41540-6\_3} {\path{doi:10.1007/978-3-319-41540-6\_3}}.

\bibitem{shah2004highly}
Mehul~A. Shah, Joseph~M. Hellerstein, and Eric~A. Brewer.
\newblock Highly-available, fault-tolerant, parallel dataflows.
\newblock In Gerhard Weikum, Arnd~Christian K{\"{o}}nig, and Stefan De{\ss}loch, editors, {\em Proceedings of the {ACM} {SIGMOD} International Conference on Management of Data, Paris, France, June 13-18, 2004}, pages 827--838. {ACM}, 2004.
\newblock \href {https://doi.org/10.1145/1007568.1007662} {\path{doi:10.1145/1007568.1007662}}.

\bibitem{shvachko2010hadoop}
Konstantin Shvachko, Hairong Kuang, Sanjay Radia, and Robert Chansler.
\newblock The hadoop distributed file system.
\newblock In Mohammed~G. Khatib, Xubin He, and Michael Factor, editors, {\em {IEEE} 26th Symposium on Mass Storage Systems and Technologies, {MSST} 2012, Lake Tahoe, Nevada, USA, May 3-7, 2010}, pages 1--10. {IEEE} Computer Society, 2010.
\newblock \href {https://doi.org/10.1109/MSST.2010.5496972} {\path{doi:10.1109/MSST.2010.5496972}}.

\bibitem{siachamis2024checkmate}
George Siachamis, Kyriakos Psarakis, Marios Fragkoulis, Arie van Deursen, Paris Carbone, and Asterios Katsifodimos.
\newblock Checkmate: Evaluating checkpointing protocols for streaming dataflows.
\newblock {\em CoRR}, abs/2403.13629, 2024.
\newblock \href {https://arxiv.org/abs/2403.13629} {\path{arXiv:2403.13629}}, \href {https://doi.org/10.48550/ARXIV.2403.13629} {\path{doi:10.48550/ARXIV.2403.13629}}.

\bibitem{silvestre2021clonos}
Pedro~F. Silvestre, Marios Fragkoulis, Diomidis Spinellis, and Asterios Katsifodimos.
\newblock Clonos: Consistent causal recovery for highly-available streaming dataflows.
\newblock In Guoliang Li, Zhanhuai Li, Stratos Idreos, and Divesh Srivastava, editors, {\em {SIGMOD} '21: International Conference on Management of Data, Virtual Event, China, June 20-25, 2021}, pages 1637--1650. {ACM}, 2021.
\newblock \href {https://doi.org/10.1145/3448016.3457320} {\path{doi:10.1145/3448016.3457320}}.

\bibitem{spenger2022portals}
Jonas Spenger, Paris Carbone, and Philipp Haller.
\newblock Portals: An extension of dataflow streaming for stateful serverless.
\newblock In Christophe Scholliers and Jeremy Singer, editors, {\em Proceedings of the 2022 {ACM} {SIGPLAN} International Symposium on New Ideas, New Paradigms, and Reflections on Programming and Software, Onward! 2022, Auckland, New Zealand, December 8-10, 2022}, pages 153--171. {ACM}, 2022.
\newblock \href {https://doi.org/10.1145/3563835.3567664} {\path{doi:10.1145/3563835.3567664}}.

\bibitem{tardieu2023reliable}
Olivier Tardieu, David Grove, Gheorghe{-}Teodor Bercea, Paul Castro, Jaroslaw Cwiklik, and Edward~A. Epstein.
\newblock Reliable actors with retry orchestration.
\newblock {\em Proc. {ACM} Program. Lang.}, 7({PLDI}):1293--1316, 2023.
\newblock \href {https://doi.org/10.1145/3591273} {\path{doi:10.1145/3591273}}.

\bibitem{neumann1956probabilistic}
John von Neumann.
\newblock Probabilistic logics and the synthesis of reliable organisms from unreliable components.
\newblock {\em Automata studies}, 34(34):43--98, 1956.

\bibitem{wang2019lineage}
Stephanie Wang, John Liagouris, Robert Nishihara, Philipp Moritz, Ujval Misra, Alexey Tumanov, and Ion Stoica.
\newblock Lineage stash: fault tolerance off the critical path.
\newblock In Tim Brecht and Carey Williamson, editors, {\em Proceedings of the 27th {ACM} Symposium on Operating Systems Principles, {SOSP} 2019, Huntsville, ON, Canada, October 27-30, 2019}, pages 338--352. {ACM}, 2019.
\newblock \href {https://doi.org/10.1145/3341301.3359653} {\path{doi:10.1145/3341301.3359653}}.

\bibitem{DBLP:conf/afips/Wensley72}
John~H. Wensley.
\newblock {SIFT:} software implemented fault tolerance.
\newblock In {\em American Federation of Information Processing Societies: Proceedings of the {AFIPS} '72 Fall Joint Computer Conference, December 5-7, 1972, Anaheim, California, {USA} - Part {I}}, volume~41 of {\em {AFIPS} Conference Proceedings}, pages 243--253. {AFIPS} / {ACM} / Thomson Book Company, Washington {D.C.}, 1972.
\newblock \href {https://doi.org/10.1145/1479992.1480025} {\path{doi:10.1145/1479992.1480025}}.

\bibitem{ZahariaCDDMMFSS12}
Matei Zaharia, Mosharaf Chowdhury, Tathagata Das, Ankur Dave, Justin Ma, Murphy McCauly, Michael~J. Franklin, Scott Shenker, and Ion Stoica.
\newblock Resilient distributed datasets: {A} fault-tolerant abstraction for in-memory cluster computing.
\newblock In Steven~D. Gribble and Dina Katabi, editors, {\em Proceedings of the 9th {USENIX} Symposium on Networked Systems Design and Implementation, {NSDI} 2012, San Jose, CA, USA, April 25-27, 2012}, pages 15--28. {USENIX} Association, 2012.
\newblock URL: \url{https://www.usenix.org/conference/nsdi12/technical-sessions/presentation/zaharia}.

\bibitem{zaharia2010spark}
Matei Zaharia, Mosharaf Chowdhury, Michael~J. Franklin, Scott Shenker, and Ion Stoica.
\newblock Spark: Cluster computing with working sets.
\newblock In Erich~M. Nahum and Dongyan Xu, editors, {\em 2nd {USENIX} Workshop on Hot Topics in Cloud Computing, HotCloud'10, Boston, MA, USA, June 22, 2010}. {USENIX} Association, 2010.
\newblock URL: \url{https://www.usenix.org/conference/hotcloud-10/spark-cluster-computing-working-sets}.

\bibitem{zaharia2013discretized}
Matei Zaharia, Tathagata Das, Haoyuan Li, Timothy Hunter, Scott Shenker, and Ion Stoica.
\newblock Discretized streams: fault-tolerant streaming computation at scale.
\newblock In Michael Kaminsky and Mike Dahlin, editors, {\em {ACM} {SIGOPS} 24th Symposium on Operating Systems Principles, {SOSP} '13, Farmington, PA, USA, November 3-6, 2013}, pages 423--438. {ACM}, 2013.
\newblock \href {https://doi.org/10.1145/2517349.2522737} {\path{doi:10.1145/2517349.2522737}}.

\bibitem{zeuch2019nebulastream}
Steffen Zeuch, Ankit Chaudhary, Bonaventura~Del Monte, Haralampos Gavriilidis, Dimitrios Giouroukis, Philipp~M. Grulich, Sebastian Bre{\ss}, Jonas Traub, and Volker Markl.
\newblock The nebulastream platform for data and application management in the internet of things.
\newblock In {\em 10th Conference on Innovative Data Systems Research, {CIDR} 2020, Amsterdam, The Netherlands, January 12-15, 2020, Online Proceedings}. www.cidrdb.org, 2020.
\newblock URL: \url{http://cidrdb.org/cidr2020/papers/p7-zeuch-cidr20.pdf}.

\end{thebibliography}

\ifx\removeAppendix\undefined
\newpage

\appendix
\section{Syntax and Semantics}\label{appendix:sec:syntax-and-semantics}

\subsection{Syntax}\label{appendix:sec:syntax}

\begin{Figure}
\centering
\fcolorbox{colorGrayBase}{white}{\begin{minipage}{\textwidth-2\fboxsep-2\fboxrule}
\begin{minipage}{\textwidth}
\centering$\begin{array}[t]{c@{\hspace{14mm}}c@{\hspace{14mm}}c}
p,~q ~{~}~{~}\text{processor ID}
&s,~o ~{~}~{~}\text{stream name}
&n \in \mathbb{N} ~{~}~{~}\text{sequence number}
\end{array}$\end{minipage}\\\begin{minipage}{\textwidth}
\[\begin{array}[t]{c@{\hspace{10mm}}c@{\hspace{10mm}}c@{\hspace{10mm}}c}
\pi ~{~}~{~}\mbox{processor}
&\sigma ~{~}~{~}\mbox{state}
&d ~{~}~{~}\mbox{message data}
&D ~{~}~{~}\mbox{auxiliary data}
\end{array}\]\end{minipage}\\[1mm]\begin{minipage}{\textwidth}
\begin{minipage}{0.59\textwidth}
\[\begin{array}[t]{r@{\hspace{1mm}}r@{\hspace{1mm}}l@{\hspace{2mm}}l}
\Pi &::=& \seq{\pi}{p}{\size{\Pi}} &\mbox{processors}\\[0.1cm]
\Sigma &::=& \seq{\sigma}{p}{\size{\Pi}} &\mbox{states}\\[0.1cm]
M &::=& \seq{m}{i}{\size{M}} &\mbox{messages}\\[0.1cm]
N &::=& \seq{N_p}{p}{\size{\Pi}} &\mbox{sequence numbers}\\[0.1cm]
N_p &::=& \mapbuild{s}{n}{s} &\mbox{sequence numbers of $p$}\\[0.1cm]
\end{array}\]\end{minipage}%
\begin{minipage}{0.41\textwidth}
\[\begin{array}[t]{r@{\hspace{1mm}}r@{\hspace{1mm}}l@{\hspace{2mm}}l}
X &::=& \seq{x}{i}{\size{X}} &\mbox{actions}\\[0.1cm]
x &::=& &\mbox{action}\\
& &+\,s\,d     &\mbox{\small\it production}\\
&\mid &-\,s\,d &\mbox{\small\it consumption}\\[0.1cm]
m &::=& n\,s\,{d}    &\mbox{message}\\[0.1cm]
\end{array}\]\end{minipage}\end{minipage}
\end{minipage}}
\caption{Streaming syntax.}\label{appendix:fig:ssyntax}
\end{Figure}

\actionApplication*

\begin{Figure}
\centering
\fcolorbox{colorGrayBase}{white}{\begin{minipage}{\textwidth-2\fboxsep-2\fboxrule}
\begin{minipage}[t]{0.5\textwidth}
\centering$\begin{array}[t]{r@{\hspace{1mm}}r@{\hspace{1mm}}l@{\hspace{2mm}}l}
& &v, ~w &\text{value}\\[0.3cm]
\pi &::=&\texttt{TK}\langle\, f,\,\sseq{S_i}{i}{\size{S}},\,o \,\rangle &\text{task}\\[0.1cm]
a &::=& \smapbuild{e}{v}{e} &\text{snapshot archive}\\[0.1cm]
\sigma &::=& \langle\, a,\, \sigma_\mathrm{V} \,\rangle &\text{state}\\[0.1cm]
\sigma_\mathrm{V} &::= &~&\mbox{volatile state}\\
& &\texttt{fl} &\mbox{\small\it failed state}\\
&\mid&\langle\, e,\, v \,\rangle &\mbox{\small\it normal state}\\[0.1cm]
\end{array}$\end{minipage}%
\begin{minipage}[t]{0.5\textwidth}
\centering$\begin{array}[t]{r@{\hspace{1mm}}r@{\hspace{1mm}}l@{\hspace{2mm}}l}
& &e\in\mathbb{N} &\text{epoch number}\\[0.3cm]
d &::=& \langle\, e,\, d_\mathrm{C} \,\rangle &\text{message}\\[0.1cm]
d_\mathrm{C} &::=& &\mbox{message cases}\\
& &\texttt{EV}\langle\, w \,\rangle &\mbox{\small\it event}\\
&\mid&\texttt{BD} &\mbox{\small\it epoch border}\\[0.1cm]
D &::=& M_0 &\text{initial input messages}\\[0.1cm]
\end{array}$\end{minipage}
\end{minipage}}
\caption{Stateful dataflow syntax.}\label{appendix:fig:low-level-syntax}
\end{Figure}
\gceNumber*
\outputMessages*
\latestCommonSnapshot*

\subsection{Derivation Rules}\label{appendix:sec:rules}

\[\infer*[Right=S-Step]{
    \Pi_p ~\Vdash~ \Sigma_p \xrightarrow{X} \Sigma_p'
\\
    X(N_p,\, M) = (N'_p,\, M')
}{
    \langle\, \Pi,\, \Sigma,\, N,\, M,\, D \,\rangle \xRightarrow[p]{N_p, X} \langle\, \Pi,\, \Sigma\map{p}{\Sigma_p'},\, N\map{p}{N'_p},\, M',\, D \,\rangle
}\]
\[\infer*[Right=S-AbsX]{
    c \xRightarrow[p]{N_p, X} c'
}{
    c \xRightarrow[p]{} c'
}\quad\quad\quad\quad\quad\quad\infer*[Right=S-AbsP]{
    c \xRightarrow[p]{} c'
}{
    c \Rightarrow c'
}\]
\[\infer*[Right=I-Event]{
  f(v, w) = v', \seq{W'_i}{i}{n}
}{
  \texttt{TK}\langle\, f,\, S,\, o \,\rangle ~\Vdash~ \langle\, a,\, \langle\, e,\, v\,\rangle \,\rangle \xrightarrow{\ssingle{-\,S_j\,\langle\,e,\,\texttt{EV}\langle\, w \,\rangle\,\rangle} \concat \sseq{+\,o\,\langle\,e,\,\texttt{EV}\langle\, W'_i \,\rangle\,\rangle}{i}{n}} \langle\, a,\, \langle\, e,\, v'\,\rangle \,\rangle
}\]
\[\infer*[Right=I-Border]{~}{
    \texttt{TK}\langle\, f,\,\sseq{S_i}{i}{n},\,o \,\rangle ~\Vdash~ \langle\, a,\, \langle\, e,\, v \,\rangle \,\rangle \xrightarrow{\sseq{-\,S_i\,\langle\, e,\,\texttt{BD} \,\rangle}{i}{n} \concat \ssingle{+\,o\,\langle\, e,\,\texttt{BD} \,\rangle}} \langle\, a\map{e}{v},\, \langle\, e+1,\, v \,\rangle \,\rangle
}\]
\[\infer*[Right=F-Fail]{~}{
    \texttt{TK}\langle\, f, S, o \,\rangle ~\Vdash~ \langle\, a,\, \sigma_\mathrm{V} \,\rangle \rightarrow \langle\, a,\, \texttt{fl} \,\rangle
}\]
\[\infer*[Right=F-Recover]{
    \langle\, a,\, \texttt{fl} \,\rangle \in \Sigma
}{
    \langle\, \Pi,\, \Sigma,\, N,\, M,\, M_0 \,\rangle \Rightarrow \text{lcs}(\langle\, \Pi,\, \Sigma,\, N,\, M,\, M_0 \,\rangle)
}\]

\section{Definitions and Proofs}\label{appendix:sec:proofs}

\subsection{Properties of Observational Explainability}\label{appendix:sec:properties-of-observational-explainability}

We follow a structured style of proofs based on~\cite{lamport1995write,lamport2012write}. In this style, the goal is to \emph{Prove} a statement under given \emph{Assumptions}. The goal to be proven can be updated through a \emph{Suffices} statement. To eliminate quantifiers, we either introduce a \emph{New} variable ($\forall$), or \emph{Pick} a variable ($\exists$).

\monotonicity*\prflabel{lem:monotonicity}{prf:monotonicity}
\begin{proof}~

\assume
\singleitem{ $
      \forall\, c' \in \dom{T}.~ \forall c.~ c' T c \implies
      \forall [C_i]^n_i \in \mathbb{E}_{c}^{R}.~
      \exists\, [C'_j]^m_j \in \mathbb{E}_{c'}^{R'}.~\\
      \forall n' < n.~ \exists m' < m.~
      O(C_{n'}) = O'(C'_{m'})
      $ \label{asn:assumption_validhistory}
}

\prove
\singleitem{
    $
    \forall\, c' \in \dom{T}.~ \forall c.~ c' T c \implies
    \forall [C_i]^n_i \in \mathbb{E}_{c}^{R}.~
    \exists\, [C'_j]^m_j \in \mathbb{E}_{c'}^{R'}.~\\
    \exists \seq{h_k}{k}{n}.~ (\forall k < n.~ \forall k' \leq k.~ h_{k'} \leq h_k) ~ \land \\
    \forall n' < n.~ \exists m' = h_{n'} < m.~
    O(C_{n'}) = O'(C'_{m'})
    $
}

\begin{enumerate}
    \item \new $c' \in \dom{T}$.
    \item \new $c$ \suchthat $c' T c$.
    \item \new $\seq{C_i}{i}{n} \in \mathbb{E}_{c}^{R}$.
    \item \pick $\seq{C'_j}{j}{m} \in \mathbb{E}_{c'}^{R'}$ \suchthat
    \begin{enumerate}
        \item $\forall n' < n.~ \exists m' < m.~ O(C_{n'}) = O'(C'_{m'})$ \label{stp:validhistory_a}
        \proofby{Exists by the assumption.}
    \end{enumerate}
    \item \pick $\seq{h_k}{k}{n}$ \suchthat
    \begin{enumerate}
        \item $\forall k' < n.~ \exists m' < m.~ O(C_{k'}) = O'(C'_{m'}) \land h_{k'} = m'$ \label{stp:validhistory_b}
        \proofby{Such $m'$ exists by step \ref{stp:validhistory_a}, and thus we can construct the sequence with $h_{k'} = m'$.}
        \item $\forall k' < n. \forall k'' \leq k'.~ O(C_{k'}) = O(C_{k''}) \implies h_{k''} = h_{k'}$  \label{stp:validhistory_d}
        \proofby{If two $k'$ and $k''$ have the same observation, then we can pick the same $m'$ for both.}
    \end{enumerate}
    \item \suffices
        \begin{enumerate}
          \item $\forall k < n.~ \forall k' \leq k.~ h_{k'} \leq h_k$
          \item $\forall n' < n.~ \exists m' = h_{n'} < m.~
          O(C_{n'}) = O'(C'_{m'})$
        \end{enumerate}
    \item $\forall k < n.~ \forall k' \leq k.~ h_{k'} \leq h_k$
    \begin{enumerate}
        \item $\forall k < n.~ \forall k' \leq k.~ O(C_{k'}) \leq_{O} O(C_{k})$ \label{stp:validhistory_c}
        \item $\forall l < m.~ \forall l' \leq l.~ O'(C_{l'}) \leq_{O'} O'(C_{l})$
            \proofby{By the definition of the monotonicity of $O$ and $O'$.}
        \item \new $k < n$.
        \item \new $k' \leq k$.
        \item $O(C_{k'}) \leq_{O} O(C_{k})$
            \proofby{By step \ref{stp:validhistory_c}}
        \item $O'(C_{h_{k'}}) \leq_{O'} O'(C_{h_k})$
            \proofby{Follows from previous step by the equality of observations $O$ and $O'$ in step \ref{stp:validhistory_b}.}
        \item \case $O'(C_{h_{k'}}) =_{O'} O'(C_{h_k})$
        \begin{enumerate}
          \item $O(C_{k'}) =_O O(C_k)$
          \item $h_{k'} = h_k$
          \proofby{Follows step \ref{stp:validhistory_d} applied to previous step.}
          \item $h_{k'} \leq h_k$
        \end{enumerate}
        \item \case $O'(C_{h_{k'}}) <_{O'} O'(C_{h_k})$
        \begin{enumerate}
          \item $h_{k'} < h_k$
          \proofby{Follows from the definition of monotonicity of $O'$.}
          \item $h_{k'} \leq h_k$
        \end{enumerate}
        \qeditem
    \end{enumerate}
    \item $\forall n' < n.~ \exists m' = h_{n'} < m.~ O(C_{n'}) = O'(C'_{m'})$
    \begin{enumerate}
        \item \new $n' < n$.
        \item \pick $m' = h_{n'} < m$.
        \item $O(C_{n'}) = O'(C'_{m'})$
              \proofby{By construction of the sequence in step \ref{stp:validhistory_b} with $n'$ in place of $k'$.}
        \qeditem
    \end{enumerate}
    \qeditem
\end{enumerate}
\end{proof}

\transitivity*\prflabel{lem:transitivity}{prf:transitivity}
\begin{proof}~

\assume
\begin{enumerate}
    \item $R ~\,^{O}\negmedspace\xrightleftharpoons{T}\!^{O'}\, R'$ \label{asn:assumption_trans_1}
    \item $R' ~\,^{O'}\negmedspace\xrightleftharpoons{T'}\!^{O''}\, R''$ \label{asn:assumption_trans_2}
\end{enumerate}

\prove
\singleitem{
  $R ~\,^{O}\negmedspace\xrightleftharpoons{T \circ T'}\!^{O''}\, R''$
}

\suffices
\singleitem{
  $
  \forall\, c'' \in \dom{T'}.~ \forall c.~ c'' (T \circ T') c \implies
  \forall [C_i]^n_i \in \mathbb{E}_{c}^{R}.~
  \exists\, [C''_j]^m_j \in \mathbb{E}_{c''}^{R''}.~\\
  \forall n' < n .~ \exists m' < m .~
  O(C_{n'}) = O'(C''_{m'})
  $
}

\begin{enumerate}
    \item \new $c'' \in dom(T')$.
    \item \new $c'$ \suchthat $c'' T' c'$.
    \item \new $c$ \suchthat $c' T c$.
          \smallnote{$c'' T' c' \land c' T c$ therefore $c'' (T \circ T') c$.}
    \item \new $\seq{C_k'}{k}{l} \in \mathbb{E}_{c'}^{R'}$.
    \item \new $\seq{C_i}{i}{n} \in \mathbb{E}_{c}^{R}$.
    \item \pick $\seq{C_j''}{j}{m} \in \mathbb{E}_{c''}^{R''}$ \suchthat
          $
          \forall k' < l.~ \exists m' < m.~
          O'(C_{k'}) = O''(C''_{m'})
          $
          \label{stp:transitivity_a}
        \proofby{Exists by expanded definition of assumption \ref{asn:assumption_trans_2}.}
    \item \suffices
          $
          \forall n' < n .~ \exists m' < m .~
          O(C_{n'}) = O'(C''_{m'})
          $
          \proofby{By definition after elimination of quantifiers.}
    \item \new $n' < n$.
    \item \pick $k' < l$ \suchthat
    \begin{enumerate}
      \item $O(C_{n'}) = O'(C'_{k'})$ \label{stp:transitivity_b}
    \end{enumerate}
    \proofbyitem{Exists by assumption \ref{asn:assumption_trans_1}.}
    \item \pick $m'$ \suchthat
    \begin{enumerate}
      \item $m' < m$
      \item $O'(C_{k'}) = O''(C''_{m'})$ \label{stp:transitivity_c}
    \end{enumerate}
    \proofbyitem{Exists by step \ref{stp:transitivity_a}.}
    \item \suffices $O(C_{n'}) = O'(C''_{m'})$.
          \proofby{By definition after elimination of quantifiers.}
    \item $O(C_{n'}) = O'(C'_{k'}) = O''(C''_{m'})$.
          \proofby{Follows from steps \ref{stp:transitivity_b} and \ref{stp:transitivity_c}.}
    \qeditem
\end{enumerate}
\end{proof}

\composition*\prflabel{lem:composition}{prf:composition}
\begin{proof}~

\begin{enumerate}
  \item \new $O''$
  \item \new $R, O, T, O', R'$ \suchthat $\rotor{R}{O}{T}{O'}{R'}$
  \begin{enumerate}
    \item \expand $\forall c \in \dom{T}.~ \forall c'. c' T c \implies \forall [C_i]^n_i \in \mathbb{E}_{c}^{R}.~ \exists [C_j']^m_j \in \mathbb{E}_{c'}^{R'}.~ \forall n' < n.~ \exists m' < m.~ \highlight{O(C_{n'}) = O'(C_{m'})}$\label{stp:composition_observability_b}
  \end{enumerate}
  \item \suffices $\rotor{R}{O'' \circ O}{T}{O'' \circ O'}{R'}$
  \begin{enumerate}
    \item \expand $\forall c \in \dom{T}.~ \forall c'. c' T c \implies \forall [C_i]^n_i \in \mathbb{E}_{c}^{R}.~ \exists [C_j']^m_j \in \mathbb{E}_{c'}^{R'}.~ \forall n' < n.~ \exists m' < m.~ \highlight{O''(O(C_{n'})) = O''(O'(C_{m'}))}$
    \smallnote{This expanded form differs only by the highlighted regions to step \ref{stp:composition_observability_b}.}
  \end{enumerate}
  \item $\forall C, C'.~ O(C) = O'(C') \implies O''(O(C)) = O''(O'(C'))$\label{stp:composition_observability_a}
  \item $\rotor{R}{O}{T}{O'}{R'} \implies \rotor{R}{O'' \circ O}{T}{O'' \circ O'}{R'}$
        \proofby{By expansion of the definitions, it follows directly by application of step \ref{stp:composition_observability_a} to the highlighted regions.}
  \qeditem
\end{enumerate}
\end{proof}

\subsection{Properties of Traces and Causality}\label{appendix:sec:properties-of-streaming}
In this section we define traces and how they relate to executions, and further define a causal order relation on traces, with the goal to show that any causal-order preserving permutation of a trace is a valid trace.
In addition to the material presented in the main body of the paper (\Cref{sec:steps-traces-and-causality}), this section presents the definitions in more detail.

\traceDefinition*

\begin{definition}[Trace Application (Extended Definition)]\label{def:trace-from-execution-full}
A trace $Z$ of length $n$ applied to a configuration $c$ results in a sequence of configurations $C$ of length $n+1$, $Z(c) = C$, if $\forall i < n$:
\[\begin{array}{r@{~}l@{~}l@{~}c@{~}r@{\,}l}
& (\exists p, N_p, X.~& Z_i = \langle \mathsc{I-Event}, p, N_p, X \rangle &\land& \{ \mathsc{I-Event}, \mathsc{S-Step} \} &\vdash C_i \xRightarrow[p]{N_p, X} C_{i+1}) \\
\lor & (\exists p, N_p, X.~& Z_i = \langle \mathsc{I-Border}, p, N_p, X \rangle &\land& \{ \mathsc{I-Border}, \mathsc{S-Step} \} &\vdash C_i \xRightarrow[p]{N_p, X} C_{i+1}) \\
\lor & (\exists p.~& Z_i = \langle \mathsc{F-Fail}, p \rangle &\land& \{ \mathsc{F-Fail}, \mathsc{S-AbsX}, \mathsc{S-Step} \} &\vdash C_i \xRightarrow[p]{\phantom{N_p, X}} C_{i+1}) \\
\lor & (&Z_i = \langle \mathsc{F-Recover} \rangle &\land& \{ \mathsc{F-Recover} \} &\vdash C_i \xRightarrow{\phantom{N_p, X}} C_{i+1})
\end{array}\]
\end{definition}

\validTrace*

Note, that by definition of execution and trace application, any sequence of configurations produced by a trace application is an execution.

It is easy to show that any prefix of a valid trace is also a valid trace.
Still, this is a useful property, so we state it as a lemma.

\begin{restatable}[Validity of Prefix]{lemma}{validityPrefix}\label{lem:validity-prefix}
All prefixes of a valid trace from $c$ are valid traces from $c$.
\end{restatable}
\begin{proof}
This property follows directly from the definition of a valid trace. Given a valid trace $Z$ from $c$, there exists an execution $C$ starting from $c$ such that $Z(c) = C$, then we can simply construct an execution $C'$ and trace $Z'$ as prefixes from $C$ and $Z'$ respectively, such that $Z'(c) = C'$.
\end{proof}

Using the notion of traces, we can define a causal order relation on trace steps similar to the happens-before relation~\cite{DBLP:journals/cacm/Lamport78}.

\begin{definition}[Causal Order (Extended Definition)]\label{def:causal-order-full}
A step $Z_i$ happens before $Z_j$ with $i < j$ if one of the four cases apply:
\begin{bracketenumerate}
\item $Z_i = \mathsc{F-Recover} \lor Z_j = \mathsc{F-Recover}$ then $Z_i$ happens before $Z_j$, by total order on the recovery steps.
\item if they are not recovery steps, \ie for some $p$ and $p'$, $N_p, N_p', X, X'$:
\[Z_i = \langle \mathsc{I-Event}, p, N_p, X \rangle \lor Z_i = \langle \mathsc{I-Border}, p, N_p, X \rangle \lor Z_i = \langle \mathsc{F-Fail}, p \rangle, \text{~and}\]
\[Z_j = \langle \mathsc{I-Event}, p', N_p', X' \rangle \lor Z_j = \langle \mathsc{I-Border}, p', N_p', X' \rangle \lor Z_j = \langle \mathsc{F-Fail}, p' \rangle\]
and $p = p'$ then $Z_i$ happens before $Z_j$, by intraprocessor order.
\item if they are action-producing steps, \ie for some $p$ and $p'$, $N_p$ and $N_p'$, $X$ and $X'$:
\[Z_i = \langle \mathsc{I-Event}, p, N_p, X \rangle \lor Z_i = \langle \mathsc{I-Border}, p, N_p, X \rangle, \text{~and}\]
\[Z_j = \langle \mathsc{I-Event}, p', N_p', X' \rangle \lor Z_j = \langle \mathsc{I-Border}, p', N_p', X' \rangle\]
and there is a message $m$ produced by $Z_i$, that is
\[\exists M, M'', N_p''.~ X(N_p, M) = (N_p'', M'') \land m\in(M''\setminus M')\]
which is consumed by $Z_j$, that is
\[\forall M'.~ m\notin M' \implies (X'(N_p', M') ~\text{is undefined})\]
then $Z_i$ happens before $Z_j$, by interprocessor order.
\item if there exists a step $Z_k$ such that $Z_i$ happens before $Z_k$ and $Z_k$ happens before $Z_j$, then $Z_i$ happens before $Z_j$, by transitivity.
\end{bracketenumerate}
\lipicsEnd
\end{definition}

Reorderings of traces are captured as permutations.
We show that permutations that preserve the causal order relation are valid in the sense that they produce valid traces.

\begin{definition}[Trace Permutation]\label{def:permutation}
A trace $Z'$ of length $n$ is a permutation of another trace $Z$ of length $n$ if there is a bijection $f$ from $\dom{Z}$ to itself such that $\forall i\in\dom{Z}.~ Z_i = Z'_{f(i)}$.
\end{definition}

\begin{definition}[Trace Permutation Preserving a Relation]\label{def:permutationpreserving}
A trace $Z'$ of length $n$ is a permutation of another trace $Z$ of length $n$ preserving relation $R$ if, for a bijection $f$ defining the permutation, $\forall i, j.~ R(Z_i, Z_j) \implies R(Z_{f(i)}, Z_{f(j)})$.
\end{definition}

Finally, we state a lemma about the validity and similarity of permutations preserving causal order.
Intuitively, it follows from the fact that causally unrelated steps should not influence each other.

\validityCausal*\prflabel{lem:validity-similarity-causal}{prf:validity-similarity-causal}
\begin{proof}
The proof is done by first (1) proving the validity; and then (2) proving the similarity of causality-preserving permutations.

\proofsubparagraph{Validity.}
Given a trace $Z$ valid from $c$, and any causality-preserving permutation $Z'$ thereof such that $Z_i = Z'_{f(i)}$,
to show is that: $Z'$ is valid from $c$.
We show this by induction over the length of prefixes of $Z'$.
\begin{enumerate}
  \item Induction base: $\seq{Z'_i}{i}{0}$ is valid from $c$.
  \begin{enumerate}
    \item $\seq{Z'_i}{i}{0} = \varepsilon = \seq{Z_i}{i}{0}$ which is a prefix of $Z$, and, by \Cref{lem:validity-prefix}, is therefore valid from $c$.
  \end{enumerate}
  \item Induction step: assuming $\seq{Z'_i}{i}{j}$ is valid from $c$, to show is that $\seq{Z'_i}{i}{j+1}$ is valid from $c$.
  \begin{enumerate}
    \item By assumption, $\seq{Z'_i}{i}{j}$ is a valid trace from $c$, thereby there exists an execution $C'$ such that $\seq{Z'_i}{i}{j}(c) = C'$.
    \item Let's proceed by case-analysis on the $(j+1)$'th step $Z'_j$.
    It suffices to show that $Z'_j$ can be applied to configuration $C'_{j}$.
    \begin{enumerate}
      \item case $Z'_j = \langle \mathsc{I-Event}, p, N_p, X \rangle$
      \begin{enumerate}
        \item By assumption, $Z'$ is a causality-preserving permutation of $Z$.
        By rule (3) of the causal order relation for step $Z'_j$, any of the consumed messages in $X$ have been produced by previous steps and are thus available for consumption in $C_j$.
        By rule (2) of the causal order relation, all steps $Z_k$ with $k < f(j)$ on the same processor $p$ as $Z'_j$, appear once in the prefix of $Z'$ before $Z'_k$. That is, the intraprocessor order has been preserved.
        For this reason, the state $\Sigma_p$ of $p$ in configuration $C'_j$ is the same as the state of $\Sigma_p$ in $C_{f(j)}$ (for some $C$ with $Z(c) = C$).
        For these reasons, the step $Z'_j$ can be applied to configuration $C'_j$ for actions $X$ on processor $p$, as the consumed messages are available in $C'_j$.
      \end{enumerate}
      \item case $Z'_j = \langle \mathsc{I-Border}, p, X \rangle$, follows analogously.
      \item case $Z'_j = \langle \mathsc{F-Fail}, p \rangle$
      \begin{enumerate}
        \item The rule $\mathsc{F-Fail}$ is always applicable as can be seen by its definition; it does not depend on the configuration. For this reason the step $Z'_j$ can be applied to $C'_j$.
      \end{enumerate}
      \item case $Z'_j = \langle \mathsc{F-Recover} \rangle$
      \begin{enumerate}
          \item By assumption, $Z'$ is a causality-preserving permutation of $Z$.
          By rule (1) of the causal order relation for step $Z'_j$, for all other steps $Z'_k$ with $f(k) < f(j)$, we have that $k < j$.
          As the step was enabled in the trace $Z$, we know that some fail step $\langle \mathsc{F-Fail}, p \rangle$ must occur in one of these other steps before the recovery step. For this reason, the recovery step is enabled and $Z'_j$ can be applied to $C'_j$.
      \end{enumerate}
    \end{enumerate}
  \end{enumerate}
\end{enumerate}

\proofsubparagraph{Similarity.}
It suffices to show that, given any trace $Z$ of length $n$ valid from $c$, for which the second to last step is not causally ordered with respect to the last step, swapping these two steps will result in the same final configuration when applied to $c$.
That is, given $Z = Z_{prefix} : [Z_{n-2}] : [Z_{n-1}]$, we need to show that for $Z' = Z_{prefix} : [Z_{n-1}] : [Z_{n-2}]$ we have that $Z(c)_{n} = Z'(c)_{n}$.
In simpler terms, it suffices to show that $[Z_{n-2}] : [Z_{n-1}](Z_{prefix}(c)) = [Z_{n-1}] : [Z_{n-2}](Z_{prefix}(c))$.
The proof is done by case analysis on the last two steps.
By assumption, $Z_{n-2}$ and $Z_{n-1}$ are not causally ordered, therefore, by rule (1) neither of them can be an $\mathsc{F-Recover}$ step, by rule (2) the steps cannot be on the same processor, by rule (3) one of the steps may not be consuming messages produced by the other step, or by rule (4) there are no transitive causal dependencies between the steps.
This leaves the following two cases.
(Case 1) If one of them is an $\mathsc{F-Fail}$ step, then reordering the two steps will not affect the final configuration.
This is the case because the steps are on different processors, therefore the steps will not be influenced by the local state effect of either step, further, the $\mathsc{F-Fail}$ step does not consume or produce any messages, therefore it cannot influence the other step. For this reason, the application of these two steps will produce the same final configuration.
(Case 2) If the steps are either a $\mathsc{I-Event}$ or $\mathsc{I-Border}$ step on different processors. As the steps are on different processors, they do not influence each other in terms of the local processor state. However, a step on one processor may produce a message which is consumed later by another processor. Fortunately, this case is ruled out by rules (3) and (4). For this reason, either step will be enabled regardless of their order, and consume and produce the same messages. Thus, resulting in the same final configuration.
\end{proof}

\subsection{Properties of the Implementation Model}\label{appendix:sec:properties-of-the-low-level}
This section presents some lemmas about the implementation model which are later used in the proof of failure transparency.

Many of the lemmas center around border steps, in particular, the committing border step, and the non-committing border step.
The committing border step is the last border step for a given epoch, it occurs when all processors have executed the border step for the epoch.
Formally, we capture this step as:
$(\{\mathsc{S-Step}, \mathsc{I-Border}\} \vdash c \Rightarrow c') \land \mathrm{gce}(c) \neq \mathrm{gce}(c')$, that is, it is a border step in which the greatest common epoch number is increased.
A non-committing border step, on the other hand, does not change the greatest common epoch number.

\begin{lemma}\label{lem:committing-border-output}
All execution steps derivable in $I$, except for the committing border step, do not affect the observed output.
\end{lemma}
\begin{proof}
The proof is done by case analysis of the rules in $\mathrm{I}$ given the configurations $c = \langle\, \Pi,\, \Sigma,\, N,\, M,\, M_0 \,\rangle$ and $c' = \langle\, \Pi',\, \Sigma',\, N',\, M',\, M_0' \,\rangle$.

To show is that: $(\{\mathsc{S-Step}, \mathsc{I-Event}, \mathsc{F-Fail}, \mathsc{F-Recover}\} \vdash c \Rightarrow c') \implies \mathrm{out}(c) = \mathrm{out}(c')$ and $(\{\mathsc{S-Step}, \mathsc{I-Border}\} \vdash c \Rightarrow c') \land \mathrm{gce}(c) = \mathrm{gce}(c') \implies \mathrm{out}(c) = \mathrm{out}(c')$.

\proofsubparagraph{Failure Step.}
The $\mathsc{F-Fail}$ rule does not introduce any change of output, since it does not change the epoch number nor the messages, and the output cannot be changed without changing the epoch number or the messages.
Our goal is to show that:
\[
(\{\mathsc{S-Step}, \mathsc{F-Fail}\} \vdash c \Rightarrow c') \implies
\mathrm{out}(c) = \mathrm{out}(c')
\]
First, let's show that:
\[
(\{\mathsc{S-Step}, \mathsc{F-Fail}\} \vdash c \Rightarrow c') \implies
\mathrm{gce}(c) = \mathrm{gce}(c') \land M = M'
\]
Since the only local step is $\mathsc{F-Fail}$, the only possible execution step is:
\[\infer*{
    \Pi_p = \texttt{TK}\langle\, f, S, o \,\rangle
\\
    \Sigma_p = \langle\, a,\, \sigma_\mathrm{V} \,\rangle
\\\\
    X = \varepsilon
\\
    \Sigma_p' = \langle\, a,\, \texttt{fl} \,\rangle
\\
    (N'_p,\, M') = X(N_p,\, M)
}{
    \langle\, \Pi,\, \Sigma,\, N,\, M,\, D \,\rangle \Rightarrow \langle\, \Pi,\, \Sigma\map{p}{\Sigma_p'},\, N\map{p}{N'_p},\, M',\, D \,\rangle
}\]
By the definition of action application, $\varepsilon(N_p,\, M) = (N_p,\, M)$, therefore $M = M'$. $\Sigma' = \Sigma\map{p}{\langle\, a,\, \texttt{fl} \,\rangle}$ where $\Sigma_p = \langle\, a,\, \sigma_\mathrm{V} \,\rangle$. From this follows that:
\[\forall q, a, a', \sigma_\mathrm{V}, \sigma_\mathrm{V}'.~ \Sigma_p = \langle\, a,\, \sigma_\mathrm{V} \,\rangle \land \Sigma_p' = \langle\, a',\, \sigma_\mathrm{V}' \,\rangle \implies a = a'\]
Further, by expanding the definitions, we can show that the $\mathrm{gce}$ is not changed:
\begin{align*}
\mathrm{gce}(c)
    &= \mathrm{min}\setbuild{\mathrm{max}(\mathrm{dom}(a))}{\Sigma_p  = \langle\, a,\, \sigma_\mathrm{V} \,\rangle}\\
    &= \mathrm{min}\setbuild{\mathrm{max}(\mathrm{dom}(a))}{\Sigma_p' = \langle\, a,\, \sigma_\mathrm{V} \,\rangle}\\
    &= \mathrm{gce}(c')
\end{align*}
From which follows our intended result:
\[
(\{\mathsc{S-Step}, \mathsc{F-Fail}\} \vdash c \Rightarrow c') \implies
\mathrm{gce}(c) = \mathrm{gce}(c') \land M = M'
\]
Next, it suffices to show that the observable output is not changed given that the $\mathrm{gce}$ and the messages are not changed:
\[
\mathrm{gce}(c) = \mathrm{gce}(c') \land M = M' \implies \mathrm{out}(c) = \mathrm{out}(c')
\]
The conclusion follows directly from the premises and the definition of $\mathrm{out}$:
\begin{align*}
\mathrm{out}(c) = \setbuild{n\,s\,\langle\, e,\, d \,\rangle}{(n\,s\,\langle\, e,\, d \,\rangle) \in M \land e \leq \mathrm{gce}(c)} &=
\\
\setbuild{n\,s\,\langle\, e,\, d \,\rangle}{(n\,s\,\langle\, e,\, d \,\rangle) \in M' \land e \leq \mathrm{gce}(c')} &= \mathrm{out}(c')
\end{align*}

\proofsubparagraph{Event Step.}
The $\mathsc{I-Event}$ rule does not affect the greatest common epoch number (gce), nor does it produce any messages with an epoch number lower than the gce.
For this reason, by the definition of $\mathrm{out}$, the output is not changed.

\proofsubparagraph{Non-Committing Border Step.}
A non-committing (local) border step, \ie an execution step for which $(\{\mathsc{S-Step}, \mathsc{I-Border}\} \vdash c \Rightarrow c') \land \mathrm{gce}(c) = \mathrm{gce}(c')$, does not affect the observable output.
By assumption, the $\mathrm{gce}$ is not changed. Further, no new messages with $e \leq \mathrm{gce}(c)$ can be produced, since all processors have epoch number $e > \mathrm{gce}(c)$ by definition of $\mathrm{gce}$, therefore these messages are not included in $\mathrm{out}$.

\proofsubparagraph{Recovery Step.}
The recovery rule does not change the observed output. Here we should reason about the latest common snapshot ($\mathrm{lcs}$), in particular, that it preserves the output messages.
Our goal is to show that:
\[
(\{\mathsc{F-Recover}\} \vdash c \Rightarrow c') \implies
\mathrm{out}(c) = \mathrm{out}(c')
\]
By applying the recovery rule, we get that:
\[
c' = \mathrm{lcs}(c)
\]
By the definition of $\mathrm{lcs}$, it the $\mathrm{gce}$ is not changed, \ie $\mathrm{gce}(c') = \mathrm{gce}(lcs(c)) = \mathrm{gce}(c)$, since $\forall \langle a', \sigma_\mathrm{V}' \rangle \in \Sigma'.~ \dom{a'} = \set{\mathrm{gce}(c)}$.
Further, $M' = M_0 \cup out(c)$, leaving the events with epoch number $e \leq \mathrm{gce}(c)$ unchanged. Therefore, the observed output is not changed, \ie $\mathrm{out}(c) = \mathrm{out}(c')$.
\end{proof}

\begin{lemma}\label{lem:epoch-order}
If a trace $Z$ consists of a sequence of $\mathsc{S-Step}$s (\ie excluding $\mathsc{F-Recover}$ steps) valid from a configuration $c$, then a trace step $Z_i$ with epoch number $e$ cannot happen before a trace step $Z_j$ with epoch number $e'$ for which $e' < e$.
\end{lemma}
\begin{proof}
By the definition of local steps in $I$, the epoch number of a processor can only grow.
In particular, for event steps and failure steps it stays the same, and for border steps it is incremented.
Therefore, if the steps are on the same processor, by case (2) of \Cref{def:causal-order-full}, for epoch numbers $e'$ and $e$ such that $e' < e$, any $e'$ step happens before all $e$ steps on the same processor.
Further, if the steps are on different processors, by (3, 4) of \Cref{def:causal-order-full}, if $e' < e$ then $Z_i$ cannot happen before $Z_j$, as: (a) any messages emitted by $Z_i$ cannot be received by $Z_j$ as $e' \neq e$; and (b) nor can there exist a transitive causal dependency, as neither of the rules (2, 3) allow a direct dependency such that a higher epoch-numbered step happens before a lower epoch-numbered step.
\end{proof}

\begin{lemma}\label{lem:reordering-output}
If a trace $Z$ is valid from a configuration $c$, and $Z'$ is a causality-preserving permutation of $Z$, such that the permutation preserves the total order of all trace steps with an epoch number less than $e$, then the observable outputs obtained by the committing border steps of all epochs below $e$ are equal for $Z$ and $Z'$.
\end{lemma}
\begin{proof}
Consider one of the committing border steps for epoch $e'$ with $e' < e$.
We can reorder the trace by moving all trace steps for epochs larger than $e'$ that occur before the committing border step directly after the committing border step.
By \Cref{lem:epoch-order}, this is a causality-preserving permutation, as none of the moved trace steps can happen before the committing border step nor any other step for epochs less than or equal to $e'$.
By \Cref{lem:validity-similarity-causal}, we know that the configuration obtained (i) after the committing border step before the reordering; and (ii) after the last moved step after the reordering, are the same.
Furthermore, by \Cref{lem:committing-border-output} we know that output function is not affected by non-committing (non-global) border steps.
As none of the moved steps can be global committing border steps, we know that the output function's result must be the same (after the reordering) for the committing border step and the last moved step.
And therefore we know that the output function at the committing border step for $e'$ must be the same after the reordering as before the reordering.
This procedure can be repeated for every committing border step up the $e'$, and therefore each committing border step will have the same observable output before as after the reordering.
\end{proof}

\begin{lemma}\label{lem:snapshot-change}
The latest common snapshot is only changed by committing border steps.
\end{lemma}
\begin{proof}
By definition of the latest common snapshot ($\mathrm{lcs}$), it is only affected by changes to $N$ and $\Sigma$ in $c$, as well as by $\mathrm{out}(c)$.
With respect to the sequence numbers $N$, the computed $\mathrm{lcs}$ solely depends on the domain of $N$, which remains the same for all steps as it is a map from processors to sequence numbers.
With respect to the states $\Sigma$, the computed $\mathrm{lcs}$ solely depends on the archives $a$ of the states for epochs less than the $\mathrm{gce}$. The archives of epochs less than $\mathrm{gce}$ are never changed, as no processor can have an epoch number less than $\mathrm{gce}$. Further, as was shown in the proof of \Cref{lem:committing-border-output}, the only steps changing the greatest common epoch are the committing border steps. Therefore, changes to $\Sigma$ do not affect the latest common snapshot.
Lastly, $\mathrm{out}(c)$ is only changed by the committing border steps (\Cref{lem:committing-border-output}).
Therefore, the latest common snapshot is only changed by the committing border steps.
\end{proof}

\begin{lemma}\label{lem:snapshot-small}
If a trace $Z$ of $\mathsc{S-Step}$s is valid from a configuration $c$, such that the last trace step is a committing border step with epoch number $e$, and epoch numbers of all trace steps are less than or equal to $e$, then the last configuration of the execution $\seq{C}{i}{n} = Z(c)$ is equal to its latest common snapshot, \ie $C_{n-1} = \mathrm{lcs}(C_{n-1})$.
\end{lemma}
\begin{proof}
By inspection of the definition of $\mathrm{lcs}$, it suffices to show that the state ($\Sigma$), the sequence numbers ($N$), and the messages ($M$) of the last configuration $C_{n-1}$ are the same as obtained by the latest common snapshot.
($\Sigma$) The last trace step for each processor in $Z$ must have been a border step for epoch $e$. This border step saves the current state to the archive for $e$, and increments the epoch number on the processor to $e + 1$. As can be seen by the definition of $\mathrm{lcs}$, the $\mathrm{lcs}$ computes the same archive and same local state, as we know that $e = \mathrm{gce}(C_{n-1})$.
($N$) The sequence numbers for each processor are incremented by one for every consumed or produced message. As can be seen in the definition of $\mathrm{lcs}$, the computed sequence numbers correspond to the number of output messages on each stream. As a stream has at most one producer, the computed sequence number corresponds to the number of produced messages. For consumed messages, we know that all messages of epochs smaller than the $\mathrm{gce}$ must have been consumed, therefore the computed sequence number is also equal.
($M$) The $\mathrm{lcs}$ computes the new set of messages from the observed output messages $\mathrm{out}$. As all produced messages in $Z$ have epoch numbers less than or equal to $e$, all produced messages are included in the observed output $\mathrm{out}(C_{n-1})$. Therefore, the computed messages are the same as the messages in the configuration $C_{n-1}$.
\end{proof}

\begin{lemma}\label{lem:snapshot}
If a trace $Z$ of $\mathsc{S-Step}$s is valid from $c$ with $Z(c) = C$, such that the first $n$ trace steps have epoch numbers less than or equal to $\mathrm{gce(C_{\size{C}-1})}$, and all later trace steps have epoch numbers greater than $\mathrm{gce(C_{\size{C}-1})}$, then $\mathrm{lcs}(C_{n}) = \mathrm{lcs}(C_{\size{C}-1})$.
\end{lemma}
\begin{proof}
The proof is done by induction on the length of the prefix of $Z$ starting from index $n$.

Induction base: the length the prefix $[Z_i]^n_i$ is $n$.
In this case, for the sequence of configurations $C' = [Z_i]^n_i(c)$, we have that $n = \size{C'} - 1$ and thus $C'_{n} = C'_{\size{C'}-1}$, and therefore $\mathrm{lcs}(C'_n) = \mathrm{C'_{\size{C'}-1}}$.

Induction step: for some prefix $Z' = [Z_i]^{n'}_i$ of size $n' > n$ for which $C' = Z'(c)$, we have to show that $\mathrm{lcs}(C'_{n}) = \mathrm{lcs}(C'_{\size{C'}-1})$, under the assumption that $\mathrm{lcs}(C'_{n}) = \mathrm{lcs}(C'_{\size{C'}-2})$.
By its definition, the latest common snapshot is only affected by all trace steps except for non-committing border steps with an epoch number less than or equal to the greatest common epoch number.
As the new step has an epoch number greater than the greatest common epoch number, and it is known to not be a committing border step (as otherwise the $\mathrm{gce}$ would be increased), it does not affect the latest common snapshot, thus $\mathrm{lcs}(C'_{\size{C'}-1}) = \mathrm{lcs}(C'_{\size{C'}-2})$.
For this reason, following the induction hypothesis, $\mathrm{lcs}(C'_{n}) = \mathrm{lcs}(C'_{\size{C'}-2}) = \mathrm{lcs}(C'_{\size{C'}-1})$.
\end{proof}

\begin{definition}[Fair Execution]\label{def:fair-execution}
An execution $C$ is fair if it is maximal (\ie it is not a prefix of another execution), has a finite amount of failures, and eventually executes any step which is {}eventually always enabled. That is, $\mathrm{fair}(C)$ iff:
\begin{gather*}
(\nexists C'.~ (C\concat C') \in \mathbb{E}_{C_0}^{\rulesI}) ~ \land\\
(\exists C', C''.~ C = (C' \concat C'') \land \size{C'}\in\mathbb{N} \land \forall \langle \Pi'', \Sigma'', M'', N'', D'' \rangle \in C''.~ \nexists \langle a, \mathtt{fl} \rangle \in \Sigma'') ~ \land\\
(\forall C', C''.~ C = (C' \concat C'') \land \size{C''} = \infty \implies \nexists z.~ (\forall c'' \in C''.~ z(c'')) \land \nexists i.~ z(C''_i) = C''_{i+1})
\end{gather*}
\end{definition}

Proving liveness typically involves fairness assumptions.
Specifically, for our proof of liveness in \Cref{thm:liveness}, we assume that executions are maximal, \ie we exclude executions which are prefixes of other executions.
This is necessary as our definition of the set of executions is prefix-closed, that is, even the execution consisting solely of the initial state is included in it.
It would be difficult to reason about the liveness of an execution which takes no execution steps.
Instead, all executions are restricted to either be terminated or infinite.
Further, we assume that the execution has at most a finite amount of failures.
Formally, this means that any fair execution must be decomposable into a finite prefix and a suffix for which the suffix contains no failures.
If this was not the case, then we may have a situation in which the execution always fails and aborts an epoch before it has completed.
Lastly, a fair execution must eventually execute any step which is {}eventually always enabled.
Formally, we describe this by negation: a step cannot eventually always be enabled but not eventually taken.
We use this last property to ensure that execution steps which are enabled to eventually be taken.
If we would not make such an assumption, it would be possible for the execution to always execute one of the processors, but never the others, thus making it difficult to argue about the liveness.
However, with the assumption, we know that, once a step is enabled, and if it remains enabled, then eventually it will be taken.
Using this, we can then show through a chain of eventual steps that the system eventually produces an output.

\livenessOfImplementation*\prflabel{thm:liveness}{prf:liveness}
\begin{proof}
Given an arbitrary initial configuration $k$ and some epoch $e$ which is present in the initial configuration, to show is that eventually the epoch $e$ appears in the output of a fair execution.
To show this, we proceed by the following definitions and statements.

(\texttt{Def0}) 
A configuration is well-formed if the messages on the streams have contiguous sequence numbers and are divided into epochs, and the processors have sequence numbers which correspond with the number of consumed/produced messages.

(\texttt{Def1})
A configuration is failed if there is a processor in a failed state.

(\texttt{Def2})
An epoch is produced if there is a message for that epoch visible in the output as defined by the output function $\mathrm{out}$.

(\texttt{S0}) Any valid initial configuration is well-formed, \ie it maintains \texttt{Def0}.
Proof. This follows from the definition of valid initial configuration.

(\texttt{S1}) Any execution step of $\rulesI$ applied to a well-formed configuration results in a well-formed configuration. 
Proof. This follows from inspection on the rules when applied to an arbitrary well-formed configuration.

(\texttt{S2})
Any execution of $\rulesI$ starting from a well-formed configuration contains only well-formed configurations.
Proof. This follows directly from \texttt{S0} and \texttt{S1}.

(\texttt{S3})
If an execution step is enabled in $\rulesI$, then it will remain enabled at least until a recovery step is taken.
Proof. This follows from inspection on the rules in $\rulesI$.

(\texttt{S4})
It is not possible that there \emph{always} \emph{eventually} is a failure.
Proof. This follows from the fair execution assumption which states that there is a ``finite amount of failures''. In other words, failures do not occur infinitely often, which implies that failures do not always eventually occur.

(\texttt{S5})
Starting from any arbitrary configuration $c$ from an execution starting from $k$, for which $c$ is not failed and the epoch $e$ is not yet produced.
If there are no failures in the execution continuing from $c$, then eventually the epoch $e$ is produced.

Proof.
We proceed by induction on the structure of the processing graph as defined by $\Pi$.
To show is that:
(\texttt{IH}) for each stream in the processing graph, eventually a border event for the epoch $e$ is produced.
That is, eventually there is some configuration for which the border events for $e$ exist on all streams, and snapshots for $e$ exist on all processors.
Note that \texttt{IH} implies \texttt{S5}, as by definition of the output function, the epoch is produced once there is a configuration for which all processors have a snapshot of the epoch.

Base case.
The base case consists of the empty processing graph with no processors, but with the set of input streams to the graph.
\texttt{IH} follows from the well-formedness of the initial configuration on the input streams.

Induction step.
By the induction hypothesis, we have that \texttt{IH} holds for some processing graph $\Pi'$ which is a subgraph of $\Pi$.
The induction step is to show that \texttt{IH} holds for the processing graph $\Pi''$, which adds one of the processors to $\Pi'$ for which all inputs to the processor are either inputs to the graph or are outputs of other processors in $\Pi'$.
This way, all possible allowed acyclic graph structures can be constructed.
We know that all inputs to the new processor will eventually contain the border message for epoch $e$ by \texttt{IH} on $\Pi'$.
Further, we know that these streams are well-formed by \texttt{S2}.
To show is that the processor eventually consumes the epoch border message for $e$ on each one of its input streams, and thus produces a corresponding epoch border message on its output stream.

By construction of the rules $\rulesI$, we know that the processor processes one epoch from its inputs until completion, before starting the next epoch (see the epoch border alignment in \Cref{fig:alignment}).
Further, by well-formedness on the initial configuration, we know that the epochs are finite in size.
From the well-formedness on its inputs, we know that there exists a sequence of execution steps for the processor such that it can process all input steps that are available to it.
We know that at least one such step will be enabled until all inputs have been processed.
Further, we know that any enabled step remains enabled until taken (\texttt{S3}).
Thus, by the fairness assumption, we know that such steps will eventually be taken, and therefore all inputs to the processor will eventually be processed.
Consequently, eventually, by an $\mathsc{I-Border}$ step, the border message for epoch $e$ will have been consumed, and a corresponding border message for epoch $e$ produced.

(\texttt{S6})
It is \emph{always} the case that: \emph{eventually} there is a failure; or \emph{eventually} the epoch is produced. In other terms, starting from any configuration in an execution, eventually there is a failure or eventually the epoch is produced.

Proof.
Given an arbitrary configuration $c$ from an execution starting from $k$.
We show that, starting from $c$, eventually there is a failure or eventually the epoch $e$ is produced.
Let us assume that the epoch has not yet been produced in any configuration leading up to $c$, otherwise we would be done.
Further, let us consider the case in which there is a failed processor in the configuration $c$.
In this case, there is a failure, and thus we are done.
Then, let us consider the case in which there is no failed processor in the configuration $c$.
If the execution continuing from configuration $c$ eventually contains a failure, then we are done.
Therefore, the last case to consider is the execution starting from configuration $c$ contains no failures and for which $e$ has not yet been produced.
This last case is handled by \texttt{S5}.

(\texttt{S7})
Eventually the epoch $e$ appears in the output of a fail execution.
More precisely, starting from any configuration in an execution, eventually the epoch is produced.

Proof. This follows from \texttt{S4} and \texttt{S6} by the following temporal reasoning:
(\emph{always} ((\emph{eventually} A) or (\emph{eventually} B))) and (\emph{not} (\emph{always} (\emph{eventually} A))) implies (\emph{always}(\emph{eventually} B)).
That is, we have shown that \emph{always} \emph{eventually} the epoch is produced.
\end{proof}

\fi
\end{document}